\pdfoutput=1
\documentclass[11pt]{article}

\usepackage[margin=1in]{geometry}
\usepackage{microtype}
\usepackage{lmodern}
\usepackage{graphicx}
\usepackage{subcaption}
\usepackage{booktabs}
\usepackage{tabularx}
\usepackage{algorithm}
\usepackage[noend]{algpseudocode}
\usepackage{amsmath}
\usepackage{amssymb}
\usepackage{mathtools}
\usepackage{amsthm}
\usepackage{thmtools}
\usepackage{thm-restate}
\usepackage{amsfonts}
\usepackage{dsfont}
\newcommand{\mathbbm}[1]{\mathds{#1}}

\usepackage[suppress]{color-edits} 
\addauthor{hb}{red}
\addauthor{am}{orange}

\usepackage[numbers]{natbib}

\usepackage[colorlinks=true,linkcolor=blue,citecolor=blue,urlcolor=blue]{hyperref}
\usepackage[capitalize,noabbrev]{cleveref}

\theoremstyle{plain}
\newtheorem{theorem}{Theorem}[section]
\newtheorem{proposition}[theorem]{Proposition}
\newtheorem{lemma}[theorem]{Lemma}

\theoremstyle{definition}
\newtheorem{definition}[theorem]{Definition}

\theoremstyle{remark}

\theoremstyle{example}
\newtheorem{example}[theorem]{Example}
\theoremstyle{conjecture}
\newtheorem{conjecture}[theorem]{Conjecture}
\newtheorem{observation}{Observation}

\DeclareMathOperator*{\argmax}{arg\,max}

\newcommand{\FGS}{\texttt{GS}}

\newcommand{\CIAP}{\texttt{CIA}}

\newcommand{\EPL}{\hat{O}}
\newcommand{\PL}{O}
\newcommand{\TGS}{t^{\texttt{GS}}}

\newcommand{\AG}{\mathcal{A}}
\newcommand{\FI}{\mathcal{F}}
\newcommand{\MARK}{\mathcal{M}}
\newcommand{\DIS}{\mathcal{D}}

\newcommand{\APP}{\mathcal{P}}

\newcommand{\INT}{\mathcal{I}}
\newcommand{\RR}{f^{\texttt{RR}}}
\newcommand{\PH}{\rho}
\newcommand{\FA}{f^{\texttt{apply}}}
\newcommand{\FM}{f^{\texttt{match}}}
\newcommand{\FFM}{a^{\texttt{match}}}
\newcommand{\MM}{\mathcal{M}(\AG,\FI)}
\newcommand{\NH}{\gamma}
\newcommand{\VF}{\mathcal{V}}
\newcommand{\VFP}{\mathcal{V}^+}

\newcommand{\CSV}{\mathcal{B}}
\newcommand{\CSFP}{\mathcal{B}'}

\usepackage[disable,textsize=tiny]{todonotes}
\usepackage{multirow}
\usepackage{enumitem}
\setlist[itemize]{leftmargin=0pt, itemindent=1em, topsep=2pt, itemsep=1pt, parsep=0pt, partopsep=0pt}
\setlist[enumerate]{leftmargin=0pt, itemindent=1.5em, topsep=2pt, itemsep=1pt, parsep=0pt, partopsep=0pt}

\title{Two-Sided Time-Independent Regret for Matching Markets with Limited Interviews}

\author{
  Amirmahdi Mirfakhar \and
  Xuchuang Wang \and
  Mengfan Xu \and
  Hedyeh Beyhaghi \and
  Mohammad Hajiesmaili \\
  \\
  University of Massachusetts Amherst
}

\date{} 

\begin{document}

\maketitle

\begin{abstract}
Two-sided matching platforms rely on preferences from both sides, yet participants can evaluate only a small fraction of potential partners. In practice, they use low-cost pre-match screening,  e.g., interviews, profile views, or trial tasks, to form noisy impressions before committing to applications and offers. We study bandit learning in matching markets with interviews, modeling these interactions as queried \emph{hints}~\citep{DBLP:conf/innovations/BhaskaraGIKM23} that reveal partial preference information to both sides while constraining subsequent applications. Our framework also allows firm-side uncertainty: firms, like agents, learn their preferences and may make early hiring mistakes. To address this, we introduce strategic deferral, a firm-side action that permits temporary vacancy, corrects premature commitments, and enables decentralized learning under coarse anonymous feedback. We design algorithms for centralized and decentralized markets and show that a constant number of interviews per round suffices for horizon-independent regret, improving over the $O(\log T)$ guarantees known without interviews. Our bounds are near-optimal: the centralized guarantee is within a factor $m$ of an information-theoretic lower bound, while decentralized algorithms match it up to polynomial factors in structured markets and remain horizon-independent in general markets.
\end{abstract}

\section{Introduction}\label{sec:intro}
Two-sided matching markets~\citep{roth1992two} provide a foundational model for settings in which outcomes depend on the preferences and decisions of both sides, from classical applications such as labor markets~\citep{roth1984evolution} and school choice~\citep{abdulkadirouglu2003school} to modern digital platforms such as creator--brand sponsorship markets, reciprocal recommenders, and research  internship or collaboration markets~\citep{hitsch2010matching,
ashlagi2020clearing,kanoria2021facilitating,arnosti2021managing,rios2023improving}. Across these domains, participants face uncertainty about match quality and often rely on limited low-cost screening before committing to a full match; Appendix~\ref{app:digital-platform-motivations}
highlights representative applications. Formally, a market consists of \(n\) agents and \(m\) firms, each with preferences over the other side, and a matching is \emph{stable} if no agent--firm pair
would both prefer to deviate from their current assignments~\citep{gale1962college}. When these
preferences are not known a priori, participants must learn them through interaction while evaluating
only a small fraction of potential partners, motivating \emph{bandit learning} in matching
markets~\citep{liu2020competing,maheshwari2022decentralized,liu2021bandit,kong2023player,
ghosh2024competing}. Much of the bandit-matching literature studies this problem under the simplifying assumptions that only agents are uncertain, while firms have known preferences. These assumptions miss two central features of the motivating settings above.

\noindent\textbf{Two missing ingredients: pre-match interviews and uncertain firms.}
(i) First, participants often engage in low-cost pre-match interactions before submitting a final application or contract. We call these interactions \emph{interviews}, but use the term broadly: in Appendix~\ref{app:digital-platform-motivations}, an interview may correspond to product seeding, a trial post, a profile
view, or a screening call. Such interactions produce \emph{pre-application} noisy side-observations of match quality for both sides and can accelerate learning. This resembles \emph{bandits with hints}~\citep{DBLP:conf/innovations/BhaskaraGIKM23,mirfakhar2025heterogeneous} and, more broadly, \emph{algorithms with predictions}~\citep{lykouris2021competitive}. However, unlike prior hinted models focusing on single-agent settings, interviews in matching markets
simultaneously \emph{reveal information} and \emph{restrict feasible actions}, since agents may only apply to interviewed firms~\citep{lee2017interviewing,beyhaghi2021randomness}. 
Thus, informative interviews can still intensify the matching competition.

(ii) Second, firms may also be uncertain about their preferences. A brand, employer, or user may initially misrank applicants based on noisy screening signals; if forced to always hire, early mistakes can reject strong candidates, slow convergence to stability, and increase regret. To model how firms hedge against such uncertainty, we allow a firm to \emph{defer} hiring in a round, i.e., to remain vacant rather than commit to an applicant it currently deems suboptimal. We introduce \emph{deferral} as a modular primitive for two-sided learning: it expands the firm action space, helps stabilize decentralized dynamics under limited feedback, and is consistent with evidence that firm-side hiring deferrals can shape stable outcomes in static/non-learning settings~\citep{kupfer2018influence}. We note that~\citep{pagare2023two} studied coordinated learning with two-sided uncertainty; however, without interviews or firm-side deferral, which are central to our decentralized two-sided learning. Detailed distinctions are discussed in Appendix~\ref{sec:remarks2}.

\noindent\textbf{Model at a glance.} We study two-sided bandit learning with interviews and uncertainty on both sides. Our model follows the platform workflow above: \textit{limited screening/interview first, then a higher-commitment application, then acceptance or deferral, followed by coarse firm-side status feedback}. In round \(t\), every agent \(a\) selects \(k\) firms to interview.
Each interviewed pair \((a,f)\) yields stochastic noisy feedback to both sides with unknown means \(\mu_{a,f}\) and \(\mu_{f,a}\), representing their expected utilities from matching.
After interviewing, each agent applies to one of its interviewed firms; a firm may accept its most-preferred applicant or defer and remain vacant, with both sides receiving stochastic rewards upon a match and zero otherwise. We evaluate performance via agent regret: the gap between the expected reward of the agent's stable match and its accrued reward over \(T\) rounds. 
Since stability is bilateral, agent-side regret bounds induce analogous firm-side guarantees under the corresponding stable benchmark.
We study both a centralized and decentralized settings. In the decentralized settings, agents observe only minimal anonymous firm-side signals: either vacancy-only feedback, indicating whether a firm filled its position, or hiring-change feedback, indicating whether a firm's hire changed, without revealing identities.
As shown in~\cref{sec:decwithvf}, some coarse firm-side signal is necessary for decentralized algorithms, especially in unstructured markets.

\noindent\textbf{A guiding question and technical challenges.} Existing bandit-matching work without interviews typically yields regret scaling at least logarithmically in $T$~\citep{liu2020competing,maheshwari2022decentralized,liu2021bandit,dai2021learning,li2024two}, with problem-dependent constants hidden in 
\(O(\log T)\) term. Interviews provide side observations before applications are made, suggesting that stable outcomes may be learned much faster. This leads to our central question: \textit{Can a two-sided market learn a stable matching with near-optimal horizon-independent regret using only a constant number of low-cost pre-match interactions per participant in each round,
in a decentralized manner with only anonymous firm-side status signals, even when firms are uncertain and may defer hiring?} Answering this requires addressing two challenges: 
(1) \textit{Uncertain firms and controlled deferral:}
when firms learn their rankings, early misrankings can induce rejections of strong applicants, creating feedback loops that redirect agents and slow convergence. Deferral can mitigate premature hires, but must be controlled so that it corrects instability without becoming a new source of regret.
(2) \textit{Interviews as hints under competition and stability:}
in single-agent hinted bandits~\citep{DBLP:conf/innovations/BhaskaraGIKM23, wei2020taking}, hints provide
side-information alongside a reward each round. In our multi-agent two-sided setup, interviews provide side-information but do not guarantee rewards: when interview sets overlap, competition intensifies and many agents may remain unmatched since each firm hires at most one. Combined with two-sided uncertainty, deferral, and anonymous status signals, this makes horizon-independent regret substantially harder than in either classical bandit matching or single-agent hinted bandits.

\begin{table}[t]
    \caption{Summary of regret and per-round time complexity results; 
$\Delta$-gap dependence for each algorithm is detailed in the 
corresponding theorem.}
    \label{sample-table}
    \centering
    \small
    \begin{tabular}{@{} l l l l l @{}}
        \toprule
        Setting                                    & Algorithm            & Market             & Feedback        & Regret                                             \\
        \midrule
        Centralized                                & \cref{alg:ciarr}     & General            & N/A             & $O(nm^2\Delta^{-2})$                               \\
        \midrule
        \multirow{2}{*}{Coordinated decentralized} & \multirow{2}{*}{\cref{alg:drr}}     & $\alpha$-reducible & \multirow{2}{*}{Vacancy-only}   & $O(n^3m^2\Delta^{-2})$                             \\
                                                   &                      & General            &                 & $O(n^4m^2\Delta^{-2})$                             \\
        \midrule
        \multirow{2}{*}{Coordination-free decentralized} & \cref{alg:ancdrr}    & $\alpha$-reducible & \multirow{2}{*}{Hiring-change}  & $O(n^3m^2\Delta^{-2})$                             \\
                                                   & \cref{alg:Eancdrr}   & General            &                 & $O\!\left(\epsilon^{-1} n^5 m^4\Delta^{-2}\right)$ \\
        \midrule
        Lower bound (information-theoretic)        & ---                  & General            & N/A             & $\Omega(nm\Delta^{-2})$                            \\
        \bottomrule
    \end{tabular}
\end{table}

\noindent\textbf{Contributions.} We introduce a two-sided bandit model 
with interviews and uncertain firms, showing for the first time that a constant number of 
interviews enable near-optimal \emph{horizon-independent} regret for 
stable matchings in both centralized and decentralized markets, 
improving upon the \(O(\log T)\) bounds of prior 
work~\citep{liu2020competing, liu2021bandit, maheshwari2022decentralized} 
under significantly weaker assumptions.

\noindent$\blacktriangleright$ \textit{Modeling uncertain firms with a 
deferral option.} We introduce \emph{deferral} as a novel modular primitive for 
two-sided learning (\autoref{sec:model}), the first to extend firms' 
action space beyond always hiring (\cref{examp:introstrategic}), 
pluggable into prior works to relax the strong assumption of 
certain firms. We design a deferral policy (\cref{alg:fdrr}) provably 
avoiding uncontrolled deferrals while enabling stable two-sided 
learning, with broader implications for strategic firm manipulations 
in matching markets.

\noindent$\blacktriangleright$ \textit{Two interviews per round suffice 
for horizon-independent regret.} We formalize interviews as bandit 
hints in a two-sided market and show, for the first time, that 
\emph{two interviews per round} using empirical mean estimators suffice 
for horizon-independent regret (\cref{thm:ciarr}), provably converging 
to an \emph{actual stable matching}, not merely matchings with zero 
regret that need not be stable. Even in the simpler single-agent case 
with guaranteed rewards, this resolves an open conjecture 
of~\citep{DBLP:conf/innovations/BhaskaraGIKM23}, who proved three hints sufficient via 
UCB-V but conjectured three necessary under empirical means; we show 
\emph{two suffice}. We extend this guarantee to the richer two-sided 
setting with uncertain firms and strategic deferral under anonymous 
signals.

\noindent$\blacktriangleright$ \textit{Centralized allocation: learning 
the agent-optimal stable matching.} In the centralized setting 
(\autoref{sec:centralized}), we design a central interview allocator 
(\texttt{CIA}) that coordinates interviews and applications. Building 
on deferred acceptance~\citep{gale1962college}, we give an algorithm 
that learns the agent-optimal stable matching in \(O(nm^2)\) regret 
under both certain and uncertain firms (\cref{thm:ciarr}). We show 
this bound is within a factor $m$ of the information-theoretic lower 
bound \(\Omega(nm)\), and conjecture \(O(nm^2)\) is tight, leaving the 
matching \(\Omega(nm^2)\) lower bound as an open problem.

\noindent$\blacktriangleright$ \textit{Decentralized learning.} Under 
vacancy-only feedback \(\VF\) (\cref{sec:decwithvfmain}), a coordinated 
algorithm (\cref{alg:drr}) achieves \(O(n^3m^2)\) regret in structured and \(O(n^4m^2)\) in general markets 
(\cref{thm:decregmain}), reflecting the natural \(O(n^2)\) cost of 
distributing \(\FGS\), within factors $m$ and $nm$ of optimal. Under 
anonymous hiring-change feedback \(\VFP\) (\cref{sec:non-coopmain}), a 
coordination-free algorithm (\cref{alg:ancdrr}) matches the coordinated 
bound within factor $m$ of optimal in structured markets and achieves time-independent regret in 
general markets with three interviews (\cref{thm:noncodecreg}), 
where cyclic blocking-pair dynamics introduce large constants, common 
in~\citep{liu2021bandit} but unlike~\citep{maheshwari2022decentralized} 
absent from our structured-market bounds. Tight lower bounds here 
remain a compelling open problem. All results hold under significantly 
weaker assumptions than~\citep{liu2021bandit}, which requires firm 
preferences known to both sides apriori and match identities revealed 
each round.
\vspace{-0.5em}
\section{Model and Preliminaries}\label{sec:model}

Consider a two-sided market $\MARK(\AG,\FI)$ with $n$ agents $\AG$ and $m$ firms $\FI$, where $n\leq m$.
The terminology is generic: in the platform applications motivating this work, agents may be creators, job seekers, users, or mentees, while firms may be brands, employers, campaigns, mentors, or reciprocal-recommendation opportunities; Appendix~\ref{app:digital-platform-motivations} gives a detailed mapping.
Each agent--firm pair $(a,f)$ is associated with two reward distributions: $\DIS_{a,f}$ for agent $a\in\AG$, with mean $u_{a,f}$, and $\DIS_{f,a}$ for firm $f\in\FI$, with mean $u_{f,a}$.
Rewards are independent across agent--firm pairs and time steps, and are bounded in $[0,1]$.
We assume a heterogeneous market with no ties: for every agent $a$ and distinct firms $f,f'\in\FI$, $u_{a,f}\neq u_{a,f'}$, and for every firm $f$ and distinct agents $a,a'\in\AG$, $u_{f,a}\neq u_{f,a'}$.
Thus preferences are strict and induce well-defined rankings on both sides.

For each agent $a\in\AG$, the \emph{ground-truth preference list} $\PL_a$ ranks firms by decreasing $u_{a,f}$, and we write $f\underset{\PL_a}{>}f'$ if $u_{a,f}>u_{a,f'}$.
Analogously, each firm $f\in\FI$ has a ground-truth preference list $\PL_f$ ranking agents by decreasing $u_{f,a}$, with $a\underset{\PL_f}{>}a'$ defined similarly.
Agents do not know their preferences \emph{a priori} and must learn them from observations.
On the firm side, we consider two cases: \emph{certain firms}, which know their preference lists in advance, and \emph{uncertain firms}, which must learn them over time.
Each agent maintains an estimate $\hat{u}_{a,f}(t)$ of $u_{a,f}$, and each firm maintains an estimate $\hat{u}_{f,a}(t)$ of $u_{f,a}$.
The induced estimated preference lists are denoted by $\EPL_a(t)$ and $\EPL_f(t)$.
For a certain firm $f$, $\hat{u}_{f,\cdot}(t)=u_{f,\cdot}$ and hence $\EPL_f(t)=\PL_f$ for all $t\in\mathcal{T}$.

\textbf{Decision Process.}
Given \(T\) rounds, let \(\mathcal{T}\) denote the set of decision-making rounds with \(|\mathcal{T}|=T\). We exhibit the following sequential interaction dynamics between agents and firms (whether \emph{certain} or \emph{uncertain}), consisting of three stages: \emph{interview}, \emph{application}, and \emph{firm-side feedback revelation}. Specifically, in each round \(1 \leq t \leq T\):

$\blacktriangleright$ \emph{Interview stage.}
Each agent $a \in \AG$ selects a subset of firms $\INT_a(t) \subseteq \FI$ to interview\footnote{Although, we refer to this stage as an \emph{interview}, it is more akin to an open house or participation in a career fair.}, with size $2 \leq |\INT_a(t)| \leq k$, where $k \in \mathbb{N}^+$ is the \emph{interview budget}.
For each interviewed firm $f \in \INT_a(t)$, the agent observes a stochastic signal $X_{a,f}(t) \sim \DIS_{a,f}$ drawn from distribution $\DIS_{a,f}$ with mean $u_{a,f}$; notably, these signals are not counted as rewards.
Similarly, each firm $f$ interviews agents in $\INT_f(t)$ and observes $X_{f,a}(t) \sim \DIS_{f,a}$ for each $a \in \INT_f(t)$.

$\blacktriangleright$ \emph{Application stage.}
After obtaining the interview results, each agent \(a\) applies to one of the interviewed firms \(\FA_a(t)\in\INT_a(t)\).

     (i) Each firm \(f\) collects applications from agents, denoted by \(\APP_f(t)=\{a\in\AG:\FA_a(t)=f\}\).
            If firm \(f\) decides to hire, it admits
            its most preferred applicant
            within \(\APP_f(t)\) according to its current preference as \(\FFM_f(t) = \argmax_{\APP_f(t)} \hat{u}_{f,a}(t).\)

     (ii) If \(\FFM_{\FA_a(t)}(t) = a\), then agent \(a\) is hired by firm \(\FA_a(t)\) and receives a stochastic reward \(X_{a,\FA_a(t)}(t)\sim \DIS_{a,\FA_a(t)}\); otherwise, agent \(a\) receives no reward.
           We denote the match of agent $a$ at round $t$ by
           \(\FM_a(t)\)
           where $\FM_a(t) = \FA_a(t)$ if $a$ is admitted by its applied firm, and $\FM_a(t) = \emptyset$ otherwise.
           The match of firm $f$ is similarly denoted as $\FFM_f(t)$. Precisely, the reward received by agent $a$, denoted $X_{a,\FM_a(t)}$, is defined as
           \begin{align} X_{a,\FM_a(t)}=X_{a,\FA_a(t)}\mathbbm{1}\left\{\FM_a(t) \neq \emptyset\right\}.\label{def:rewards1}
           \end{align}
           
$\blacktriangleright$ \emph{Firm-side feedback revelation stage.}
At the end of round $t$, agents observe one of the following firm-side feedback.

\begin{minipage}{0.40\linewidth}
{\small
\begin{align}
     \VF(t)  &\doteq \{f \in \FI : \FFM_f(t) = \emptyset\}, \label{def:vf}
\end{align}
}
\end{minipage}
\hfill
\begin{minipage}{0.60\linewidth}
{\small
\begin{align}
     \VFP(t) &\doteq \VF(t)\cup\{f \in \FI : \FFM_f(t-1) \neq \FFM_f(t)\}. \label{def:vfp}
\end{align}
}
\end{minipage}

Here, the \emph{vacancy-only} feedback $\VF(t)$ is the set of firms that are vacant at time $t$, whereas \emph{anonymous hiring changes} $\VFP(t)$ includes firms whose hiring status has changed since the previous round, without revealing their current matches, on top of the vacancies; i.e., $\VF \subseteq \VFP$. We view $\VF(t)$ as the weaker signal and $\VFP(t)$ as the stronger one; these feedback determine what information agents observe and play a central role in our algorithm design.

\noindent\textbf{Optimal and Pessimal Regrets.}
For policy $\pi$, we use regret to quantify performance from the agents' perspective. Since multiple stable matchings may exist, we focus on agent-optimal and agent-pessimal stable matches with corresponding analysis depending on whether we are in the centralized or decentralized model. 
More precisely, following the prior work~\cite{liu2021bandit}, we define the \emph{optimal} and \emph{pessimal} regrets under policy \(\pi\) for agent $a$ over $T$ rounds as follows

     {\small
\setlength{\abovedisplayskip}{-15pt}
\begin{align}
    \underbrace{\mathbb{E}[\overline{R}_a(\pi,T)]}_{\text{$a$-Optimal}} \doteq \sum_{t=1}^T u_{a,\overline{f^*_a}} - \mathbb{E}[X_{a,\FM_a(t)}(t)], 
    \qquad
    \underbrace{\mathbb{E}[\underline{R}_a(\pi,T)]}_{\text{$a$-Pessimal}} \doteq \sum_{t=1}^T u_{a,\underline{f^*_a}} - \mathbb{E}[X_{a,\FM_a(t)}(t)], \label{eq:optreg}
\end{align}
}
where $\overline{f^*_a}$ and $\underline{f^*_a}$ are the best and worst stable matches of agent $a$.
For brevity, we omit the policy $\pi$ from the regret and use $\mathbb{E}[\overline{R}_a(T)]$ and $\mathbb{E}[\underline{R}_a(T)]$, hereinafter. For each agent $a \in \AG$ and firm $f \in \FI$, we define the \emph{optimal reward gap} as $\overline{\Delta}_{a,f} \doteq |u_{a,\overline{f^*_a}} - u_{a,f}|$ and the \emph{pessimal reward gap} as $\underline{\Delta}_{a,f} \doteq |u_{a,\underline{f^*_a}} - u_{a,f}|$. Similarly, we define optimal reward gap as $\overline{\Delta}_{f,a} \doteq |u_{f,\underline{a^*_f}} - u_{f,a}|$ and pessimal reward gap as $\underline{\Delta}_{f,a} \doteq |u_{f,\overline{a^*_f}} - u_{f,a}|$, where $\overline{a^*_f}$ and $\underline{a^*_f}$ denote the firm-optimal and pessimal stable matches for firm $f$, noting the agent-optimal matching is firm-pessimal and vice versa~\cite{gale1962college}.

We note that when the market admits a unique stable matching, the pessimal and optimal regrets are identical, implying consistency. As such, for those markets, we denote the regret and reward gaps by $\mathbb{E}[R_a(T)]$, $\Delta_{a,f}$, and $\Delta_{f,a}$. An example of such a market is presented as follows.

\textbf{$\alpha$-Reducible Markets.} A key class of markets with a unique stable matching is $\alpha$-reducible markets, introduced for uniqueness by~\cite{clark2006uniqueness} and later used by~\cite{maheshwari2022decentralized} for coordination-free decentralized bandit learning. This structure is less restrictive than serial dictatorship~\citep{sankararaman2021dominate,wang2024optimal}, while still covering a broad family of unique-stable-matching markets. Formally, a pair $(a,f)\in \AG\times\FI$ is a \emph{fixed pair} if $a$ and $f$ are mutual top choices, i.e.,
$f \underset{O_a}{>} f'$ for all $f' \in \FI \setminus \{f\}$, and
$a \underset{O_f}{>} a'$ for all $a' \in \AG \setminus \{a\}$.

\vspace{-0.7em}
\begin{definition}
A market $\MARK(\AG,\FI)$ is \emph{$\alpha$-reducible} if every $\MARK(\AG',\FI')$ with $\AG'\subseteq\AG$ and $\FI'\subseteq\FI$ contains a fixed pair.
\end{definition}
\vspace{-0.7em}

In $\alpha$-reducible markets, we can iteratively remove fixed pairs: there exists a fixed pair $(a_1,f_1)$; removing it yields an $\alpha$-reducible sub-market; repeating yields pairs $(a_i,f_i)_{i\in[n]}$, which match all agents and induce the unique stable matching with firms in $\FI\setminus\{f_i\}_{i\in[n]}$ left unmatched. We focus on $\alpha$-reducible markets in the main body, as this layered structure highlights the analysis and the main proof ideas and defer the extension to general markets to the appendix.

\subsection{Extended Action Space for Firms' Uncertainty}
A key departure from prior work is that firms may be \emph{uncertain} about their own preferences.
If an uncertain firm is forced to hire whenever it has applicants, early estimation errors can ``lock in'' an unstable outcome: the firm keeps hiring the agent it currently, but incorrectly, believes is best, while agents may receive no local evidence that the firm has learned otherwise.
We therefore extend the firm action space to allow a firm to \emph{strategically abstain} from hiring in selected rounds.
In platform terms, this corresponds to leaving a campaign slot open or postponing a hire or mentorship.
Crucially, abstention makes the firm publicly vacant through $\VF(t)$ and $\VFP(t)$, giving previously rejected agents a coarse signal that can unwind decentralized deadlocks.

\begin{example}[Why abstention is necessary]\label{examp:introstrategic}
Let $\MM$ have agents $\AG=\{a_1,a_2\}$ and (strategic) uncertain firms $\FI=\{f_1,f_2\}$ with ground-truth preferences
\[
\begin{array}{lr}
\PL_{a_1}:\langle f_1^\dagger \succ f_2\rangle & \PL_{f_1}:\langle a_1^\dagger \succ a_2\rangle\\
\PL_{a_2}:\langle f_1 \succ f_2^\dagger\rangle & \PL_{f_2}:\langle a_2^\dagger \succ a_1\rangle .
\end{array}
\]
The unique stable matching is $(a_i,f_i)_{i\in[2]}$ (marked by $\dagger$s). Suppose that at time $t$ the estimated preferences are
\[
\begin{array}{lr}
\EPL_{a_1}(t):\langle f_1 \succ f_2^*\rangle & \EPL_{f_1}(t):\langle a_2^* \succ a_1\rangle\\
\EPL_{a_2}(t):\langle f_1^* \succ f_2\rangle & \EPL_{f_2}(t):\langle a_2 \succ a_1^*\rangle ,
\end{array}
\]
so only $f_1$ is wrong (it swaps $a_1$ and $a_2$). The stable matching under $\EPL(t)$ is $(a_1,f_2)$ and $(a_2,f_1)$, marked by $(*)$s, which is unstable under the ground truth.

Now assume that at some later time $t'>t$, firm $f_1$ learns the correct order, i.e., $\EPL_{f_1}(t')=\PL_{f_1}$. In a decentralized setting, this correction may be \emph{unobservable} to agents, so if firms always hire, the market can remain stuck at $(a_1,f_2),(a_2,f_1)$. If instead $f_1$ can abstain at $t'$, it can \emph{reject} its current hire: by not hiring, $f_1$ becomes vacant and appears in the firm-side feedback, which all agents observe. This vacancy triggers $a_2$ to move to $f_2$, which displaces $a_1$ back to $f_1$, restoring the stable matching $(a_i,f_i)_{i\in[2]}$.
\end{example}

Formally, we allow each firm $f\in\FI$ to choose a private binary hiring decision $\NH_f(t)\in\{0,1\}$ each round, where $\NH_f(t)=1$ means hire and $\NH_f(t)=0$ means remain vacant. The decision is taken after interviews, once $f$ observes its applicant set $\APP_f(t)$. With this modification, the realized reward of agent $a$ becomes \refstepcounter{equation}\label{eq:reward} \noindent \(X_{a,\FM_a(t)}
=
\NH_{\FA_a(t)}(t)\,X_{a,\FA_a(t)}\,\mathbf{1}\{\FM_a(t)\neq \emptyset\}.\)\hfill(\theequation)

Accordingly, beyond agent-side learning, we introduce a \emph{strategic rejection policy} (\cref{alg:fdrr}) that selects $\NH_f(t)$ each round. We call any firm running \cref{alg:fdrr} \emph{strategic}, even if it never abstains; for a strategic \emph{certain} firm, \cref{alg:fdrr} never triggers abstention, so $\NH_f(t)=1$ for all $t$. Firms that hire every round are \emph{non-strategic} (even if uncertain). We defer the formal policy to~\cref{sec:uncertain-firms}.

\vspace{-1em}
\subsection{Algorithmic Paradigms and Preliminaries} \label{sec:preliminaries}
\vspace{-0.5em}
Here, we first formally present two different algorithm design paradigms \emph{centralized} and \emph{decentralized}, and then introduce the preliminaries common to both.

\textbf{Centralized vs. Decentralized Algorithms.}
We study two settings. In the \emph{centralized} setting, a coordinator allocates interviews to avoid collisions. In the \emph{decentralized} setting, agents choose interview sets independently, so multiple agents may apply to the same firm; such collisions must be resolved using firm-side feedback, either $\VF(t)$ or $\VFP(t)$. 

$\blacktriangleright$ \textit{Centralized algorithms.}
We assume an omniscient \textbf{C}entral \textbf{I}nterview \textbf{A}llocator ($\CIAP$) that observes the estimated preference lists of all agents and firms and selects the interview sets $\INT_a(t)$ at each round.

$\blacktriangleright$ \textit{Decentralized algorithms.}
Without $\CIAP$, each agent $a$ selects $\INT_a(t)$ from its \emph{local observations} (Definition~\ref{def:decentralized}), which include: (i) its estimated preference lists, (ii) whether its application was rejected, and (iii) the firm-side feedback from the previous round.

\textbf{Algorithmic Preliminaries.}
Here, we introduce the preliminaries of our algorithmic designs.


\textit{Empirical means as estimators.}
We use empirical averages to estimate expected utilities for each agent--firm pair \((a,f)\). For agent $a$ and firm $f$ at time $t$, the estimator is
$
    \hat{u}_{a,f}(t)
    = \frac{\sum_{t'=1}^{t} X_{a,f}(t')\, \mathbbm{1}\{f \in \INT_a(t')\}}{N_{a,f}(t)},
$
where $N_{a,f}(t) = \sum_{t'=1}^{t} \mathbbm{1}\{f \in \INT_a(t')\}$ is the number of interviews of $a$ at $f$ up to time $t$. For uncertain firms, the estimator $\hat{u}_{f,a}(t)$ is defined analogously.

\textit{The structure of the interview set $\INT_a(t)$}
Across all our algorithms, each agent interviews between two and $k$ firms per round, i.e., $2 \le |\INT_a(t)| \le k$. In the main body, we focus on the case $k=2$, i.e., showing only two interviews is sufficient for time-independent regret, and defer a special extension to unstructured markets with $k=3$ to~\cref{sec:noncoopgen}. The interviewing set is $\INT_a(t) = \{\RR_a(t), \FA_a(t)\}$, where $\FA_a(t)$ is the firm to which agent $a$ applies and $\RR_a(t)$ is selected in a round-robin manner for exploration, as follows.

$\blacktriangleright$ \textit{Interviewing round-robin firm $\RR_a(t)$.}
Each agent interviews a designated firm $\RR_a(t)$ such that over each block of $m$ rounds, every firm appears exactly once as the round-robin firm for every agent, ensuring uniform coverage. Formally, for agents indexed as $a_i$ and firms as $f_j$, we set $\RR_{a_i}(t) = f_{(t+i \bmod m)+1}$. This round-robin exploration resembles the exploration phase of Explore-then-Commit ($\texttt{ETC}$)  schemes~\citep{lattimore2020bandit}, but here it is specifically designed to achieve \emph{time-independent regret} rather than the $O(\log T)$ behavior typical of $\texttt{ETC}$.

$\blacktriangleright$ \textit{Interviewing to-be-applied firm $\FA_a(t)$.}
The other firm $\FA_a(t)$ that agent $a$ interviews with and applies to is chosen via the Gale–Shapley algorithm ($\FGS$) (or its distributed variant), which computes a stable matching with respect to the given preference lists of both sides.


\vspace{-0.7em}
\section{Centralized Learning}\label{sec:centralized}
\vspace{-0.7em}
We study the centralized setting where $\CIAP$ controls interview assignments, and present \emph{Centralized Interview Allocation with Round-Robin Interviews} (\cref{alg:ciarr}), \amedit{which also creates a backbone for our} decentralized design.  At each round $t$, $\CIAP$ runs $\FGS$~\citep{gale1962college} on the estimated lists $\{\EPL_a(t)\}_{a\in\AG}$ and $\{\EPL_f(t)\}_{f\in\FI}$, producing a stable matching that assigns each agent $a$ a target firm $\FA_a(t)$. The allocator then sets $\INT_a(t)=\{\FA_a(t),\RR_a(t)\}$; agents interview both firms, update their estimates, and apply only to $\FA_a(t)$. The following theorem gives \emph{time-independent} regret for all agents under \cref{alg:ciarr} with uncertain firms with the full proof presented in~\cref{sec:proof-ciarr}.
\hbedit{We show a factor \(m\) tightness of the regret in~\cref{thm:ciarr} in~\cref{sec:centopt}.}
%
\begin{restatable}{theorem}{thmciarr}\label{thm:ciarr}
In market $\MM$ with non-strategic firms, agent $a$'s optimal regret under~\cref{alg:ciarr} is
\(
\mathbb{E}[\overline{R}_a(\mathcal{T})]\in O(nm^2\min\!\{\overline{\Delta}_{\AG},\underline{\Delta}_{\FI}\}^{-2}),
\) with $\overline{\Delta}_{\AG} \doteq \min_{a\in\AG}\overline{\Delta}_{a}$ and $\underline{\Delta}_{\FI} \doteq \min_{f\in\FI}\underline{\Delta}_{f}$.
\end{restatable}%
\begin{proof}[Proof Sketch]
Following~\cite{liu2020competing}, for each agent $a\in\AG$ and firm $f\in\FI$, let $\mathcal{H}_{a,f}$ (resp.\ $\mathcal{L}_{a,f}$) denote the firms ranked above (resp.\ below) $f$ in $\PL_a$, with estimated analogues $\hat{\mathcal{H}}_{a,f}(t)$ and $\hat{\mathcal{L}}_{a,f}(t)$ defined w.r.t.\ $\EPL_a(t)$; firm-side sets are defined analogously.
\begin{definition}[Valid Preference Lists]\label{def:valid}
In a matching market $\MARK(\AG,\FI)$, $\EPL_a(t)$ is \emph{valid} w.r.t.\ $f\in\FI$ if $\hat{\mathcal{H}}_{a,f}(t)\subseteq \mathcal{H}_{a,f}$. Analogously, $\EPL_f(t)$ is \emph{valid} w.r.t.\ $a\in\AG$ if $\hat{\mathcal{H}}_{f,a}(t)\subseteq \mathcal{H}_{f,a}$.
\end{definition}
Let $\Bar{\mathcal{E}}_{a,f}$ denote the set of invalid rounds for pair $(a,f)$ (with $\Bar{\mathcal{E}}_{f,a}$ defined similarly). Regret accrues only in invalid rounds, and the following lemma (proved in~\cref{appx:cent}) bounds their number:
\begin{restatable}{lemma}{validrr}\label{lemma:valid-rr}
In a matching market $\MM$, for any agent $a\in\AG$ with optimal stable match $\overline{f^*_a}\in\FI$,~\cref{alg:ciarr} guarantees
\(
\mathbb{E}\![|\Bar{\mathcal{E}}_{a,\overline{f^*_a}}|]
\in O(|\mathcal{L}_{a,\overline{f^*_a}}|\cdot m\cdot \Delta_{a}^{-2}),
\) where $\Delta_{a} \doteq \min_{f \neq \overline{f^*_a}} \overline{\Delta}_{a,f}$.
\end{restatable}
An analogous bound $O(n|\mathcal{L}_{f,\underline{a^*_f}}|\cdot \underline{\Delta}_{f}^{-2})$ holds for firms. Summing over all agents and firms and bounding $|\mathcal{L}_{a,\overline{f^*_a}}|\leq m$ and $|\mathcal{L}_{f,\underline{a^*_f}}|\leq n$ yields~\cref{thm:ciarr}.
\end{proof}

 \vspace{-0.5em}
\section{Decentralized Learning}\label{sec:decentralized}
\vspace{-0.5em}
We study decentralized learning without $\CIAP$, where interview decisions rely solely on local observations. Agents select interviewed firms $\RR_a(t)$ via round-robin and apply to $\FA_a(t)$ computed through distributed Gale-Shapley ($\FGS$) executions. Since $\FGS$ is inherently global, each execution spans multiple rounds and incurs regret, so agents should trigger it only when local signals justify it. We use firm-side feedback as a shared coordination signal and design algorithms under both feedback models: \emph{anonymous hiring changes} $(\VFP)$ and \emph{vacancy-only} $(\VF)$.

We introduce a firm-side \emph{strategic rejection policy} (\cref{alg:fdrr}, \cref{sec:uncertain-firms}) applicable under both feedback models. Under \emph{vacancy-only} feedback $\VF(t)$ (\cref{sec:decwithvfmain}), agents jointly trigger distributed $\FGS$, yielding \cref{alg:drr} with $O(n^3m^2)$ and $O(n^4m^2)$ regret for structured and general markets. Under the richer $\VFP(t)$ (\cref{sec:non-coopmain}), agents run independent $\FGS$ threads without explicit coordination, yielding \cref{alg:ancdrr} with $O(n^3m^2)$ and \cref{alg:Eancdrr} with $O(\epsilon^{-1}n^5m^4)$ regret, respectively.

\begin{algorithm}[t]
\caption{Decentralized Strategic Rejection}
\label{alg:fdrr}
\begin{algorithmic}[1]
    \State \textbf{Input:} $f$, $\AG$, $\mathcal{T}$, switching conditions $\mathcal{S}_f$ \eqref{def:uncertainfirmstra}
    \State \textbf{Initialize:} $c_f(1) \gets 0$, $r_{f,a}(1) \gets 0$ for all $a \in \AG$
    \For{$t \in \mathcal{T}$}
        \State $\NH_f(t) \gets 1$
        \State Interview with agents in $\INT_f(t)$ and update $\hat{u}_{f,a}(t+1)$ for all $a \in \INT_f(t)$
        \State Receive the set of applicants $\APP_f(t)$
        \State $a^*_f(t) \gets \arg\max_{a \in \APP_f(t)} \hat{u}_{f,a}(t)$
        \If{$t \notin \mathcal{S}_f$} 
            $\FFM_f(t) \gets a^*_f(t)$ \Comment{hiring top applicant}
        \Else 
            $\NH_f(t) \gets 0$, $\FFM_f(t) \gets 0$ \Comment{rejecting all applicants}
        \EndIf
        \State \textsc{UpdateFirmRejVars}$(f,t)$ \Comment{Subroutine~\ref{alg:frejvar}}
    \EndFor
\end{algorithmic}
\end{algorithm}
\vspace{-0.5em}
\subsection{Strategic Firm's Rejection Policy}\label{sec:uncertain-firms}

We introduce a firm-side policy, given in \cref{alg:fdrr}, that specifies whether a firm hires in each round and, if so, which applicant it hires. The policy is the same under both firm-side feedback models, $\VF(t)$ and $\VFP(t)$. At a high level, a strategic uncertain firm $f\in\FI$ abstains from hiring when its updated estimates suggest it may have mis-ranked a previously rejected agent relative to its current top applicant. After interviewing and observing the applicant set, the firm compares its current top applicant to previously rejected agents: if the top applicant is not estimated to dominate every such \emph{active} rejected agent, the firm sets $\NH_f(t)=0$, rejects all applicants, and becomes vacant (appearing in $\VF$ and $\VFP$); otherwise, it hires its estimated top applicant and updates its rejection records.


We now formalize the conditions under which a strategic uncertain firm performs a strategic rejection. Each firm $f\in\FI$ maintains two private state variables, $r_{f,a}(t)$ and $c_f(t)$. The variable $r_{f,a}(t)$ records the most recent round prior to $t$ in which $f$ rejected agent $a$ in favor of another hire, while $c_f(t)$ records the most recent round prior to $t$ in which $f$ remained vacant. We interpret $r_{f,a}(t)\ge c_f(t)$ as: agent $a$ was rejected after the last vacancy, so $f$ may need to reconsider $a$ as estimates evolve. Let
\(
a_f^*(t)\doteq \argmax_{a\in\APP_f(t)} \hat{u}_{f,a}(t)
\)
denote the firm’s current top applicant in $\APP_f(t)$. Formally,

{\small
\setlength{\abovedisplayskip}{-12pt}
\begin{align}
    r_{f,a}(t)
    &\doteq
    \argmax_{t'< t}
    \left\{
        \NH_f(t')=1 \;\land\;
        \begin{aligned}
            &\FA_a(t')=f\\
            &\FFM_f(t') \neq a
        \end{aligned}
    \right\},\qquad
    c_f(t)
    &\doteq
    \argmax_{t''< t}
    \left\{
        \FFM_f(t'') = \emptyset
    \right\}\label{eq:rfa}
\end{align}
}
\vspace{-5pt}
The set of rounds where firm $f$ abstains from hiring is
\refstepcounter{equation}\label{def:uncertainfirmstra}
\noindent
\(\mathcal{S}_f \doteq \{t \in \mathcal{T} : \exists a \in
\hat{\mathcal{H}}_{f,a_f^*(t)}(t) \ \land\ r_{f,a}(t) \geq c_f(t).\}\)
\hfill(\theequation)

This condition captures the case where a previously rejected agent (displaced by another hire) is now estimated to be preferred to the current top applicant, triggering a \emph{strategic rejection} (deferred hiring) rather than a permanent one.  One can observe that under~\cref{alg:fdrr} a certain firm always hires.

\vspace{-0.5em}
\subsection{Coordinated Decentralized Algorithm with Vacancy-Only Feedback $\VF(t)$ }\label{sec:decwithvfmain}
In this section, we design a novel \emph{coordinated} decentralized learning algorithm under vacancy-only feedback $\VF(t)$. We first highlight the high-level idea and the main technical challenges of designing such algorithms under limited feedback while maintaining robustness to strategic rejections by firms. We then present the key algorithmic components and the resulting regret guarantees. An extensive explanation of the details and analysis is deffered to the~\cref{sec:decwithvf}.

\textbf{Necessity of Coordination.}
A key challenge under vacancy-only feedback $\VF(t)$ is that hiring changes remain hidden. This can obscure critical market transitions — for instance, two agents may simultaneously swap matches without changing the vacancy set, leaving others to act on stale information (see~\cref{examp:coordfgs}). An agent rejected by a firm, for example, may never learn that the firm's match has since changed in their favor.

\textbf{Achieving $O(1)$ Regret under Indistinguishable Vacancy Signals.}
Under $\VF(t)$, the vacancy signal is the only public information shared among agents, so coordination must be embedded within it. An agent can signal by abstaining from applying, thereby creating a vacancy. However, since this feedback is anonymous, agents cannot distinguish coordination signals from ordinary interviewing set change or strategic firm rejections. The challenge is to design a scheme that triggers and encodes coordination correctly while maintaining time-independent regret and bounded coordination cost.

\textbf{Algorithmic Design.} To address the challenges above, we design a coordinated framework that partitions time into alternating \emph{updating} and \emph{committing} phases, encoded by $\PH_a(t)\in\{0,1\}$, where $\PH_a(t)=0$ denotes updating and $\PH_a(t)=1$ denotes committing. The process starts with an updating phase: agents coordinate on the interview sets for the next committing phase using estimated preference lists fixed at $\TGS$, the synchronously maintained start of the most recent updating phase, which is initially set to be $0$. After updating, they commit to this outcome while acting independently. During committing, agents monitor local events indicating that the lists fixed at $\TGS$ may have become invalid, and hence that the coordinated matching may no longer be stable. Upon detecting such an event, an agent signals coordination by refusing to apply, thereby altering the public vacancy feedback $\VF(t)$, causing all agents to synchronously set a new $\TGS$ and enter another updating phase.

It remains to specify how agents coordinate, when an updating phase ends, and which local events during committing should trigger coordination. These rules must be robust to strategic rejections. We address this via fixed-length updating phases and a small set of local inconsistency and vacancy-based triggers, all defined through a private rejection-time variable $r_{a,f}(t)$, which records the most recent time $t' < t$ at which agent $a$ applied to firm $f$ and was non-strategically rejected.
\begin{itemize}
    \item \textbf{Updating (Coordination).}
    Agents coordinate for a fixed $3n^2$ rounds starting at $\TGS$, using a common snapshot of their estimated preferences at time $\TGS$. At the start of this phase, all agents reset $r_{a,f}$ to zero. Each agent $a$ then selects actions from firms that have not rejected it since $\TGS$:
    $
        \CSV_a(t) \doteq \{ f \in \FI : r_{a,f}(t) < \TGS \},
        \FA_a(t) = \arg\max_{f \in \CSV_a(t)} \hat{u}_{a,f}(\TGS).
    $
    This phase is shown, by~\cref{obs:drr1}, to end with applications $(a,\FA_a(\TGS+3n^2))_{a\in\AG}$ forming a perfect matching, which agents then commit to in the subsequent committing phase.

    \item \textbf{Switching to Updating.}
    During committing, agents continue applying to the firm selected at the end of the last updating phase. Agents monitor a small set of local triggers; upon detecting any such event, an agent signals coordination by altering $\VF(t)$, e.g., by not applying, which creates an observable vacancy because the committed applications form a matching. This causes all agents to synchronously set $\TGS \gets t+1$ and enter a new updating phase. The set of switching events are (i) preference-list inconsistencies $\mathcal{S}^{\mathrm{inc}}_a$, (ii) rejections due to strategic firm behavior $\mathcal{S}^{\mathrm{rej}}_a$, and (iii) unexpected vacancy as coordination signals $\mathcal{S}^{\mathrm{vac}}_a$.
\end{itemize}
\begin{algorithm}[t]
\caption{Modular Coordinated Learning Process}
\label{alg:modularcoord}
\begin{algorithmic}[1]
    \State \textbf{Initialize:} $\TGS\gets 1$, $\rho_a\gets 0$, and $r_{a,f}\gets 0$ for all $f\in\FI$
    \For{$t\in\mathcal{T}$}
        \State If $\rho_a=0$, run \textsc{CoordinatedUpdate}  $3n^2$ rounds from $\TGS$ using $\hat u_{a,f}(\TGS)$; then set $\rho_a\gets 1$
        \State If $\rho_a=1$, commit to the output of the last update
        \State If $t\in \mathcal{S}^{\mathrm{inc}}_a\cup\mathcal{S}^{\mathrm{rej}}_a\cup\mathcal{S}^{\mathrm{vac}}_a$, set $\FA_a(t)\gets\emptyset$, reset $r_{a,f}\gets0$, set $\TGS\gets t+1$, and set $\rho_a\gets0$
    \EndFor
\end{algorithmic}
\end{algorithm}
\noindent \textbf{Regret Guarantee.}
Under the above policies, we obtain time-independent regret for both $\alpha$-reducible and general markets under $\VF(t)$, with a proof sketch (full proofs in~\cref{proof:thmalphareg} and~\cref{proof:decreg}).

\begin{restatable}{theorem}{thmdecregmain}\label{thm:decregmain}
In a matching market $\MM$, with firms and agents following~\cref{alg:fdrr} and~\cref{alg:drr}, the expected regret is time-independent. In particular, 
{\small\[
\mathbb{E}\!\left[\underbar{R}_{a}(\mathcal{T})\right] \in 
\begin{cases}
O(n^3m^2\min\!\left\{\overline{\Delta}_{\AG},\underline{\Delta}_{\FI}\right\}^{-2} + m\sum_{j=1}^{i}(\left|\mathcal{L}_{a_j,f_j}\right|\overline{\Delta}_{a_j}^{-2}
+
\left|\mathcal{L}_{f_j,a_j}\right|\underline{\Delta}_{f_j}^{-2})), & \text{$\alpha$-\text{\emph{reducible}}},\\[4pt]
O(n^4m^2\Delta^{-2}), & \text{\emph{General Markets}}.
\end{cases}
\]}
\end{restatable}

\begin{proof}[Proof Sketch.]
Although the analysis differs between $\alpha$-reducible and general markets, the high-level approach is to decompose regret into updating and committing phases. Since each updating phase (a distributed $\FGS$ run) has fixed length $O(n^2)$, the key challenge is to bound the number of such phases. 
Under~\cref{alg:drr}, the number of updating phases is $O(nm^2)$ in $\alpha$-reducible markets (\cref{lemma:alphadec}) and $O(n^2m^2)$ in general markets (\cref{lemma:dec}), which explains the gap in the final bounds. The key idea is to charge each coordination trigger to a specific invalid estimated preference list, either on the agent side or the firm side, thereby limiting how often updates can occur. 
%
The committing phases have random duration and are handled separately for the two market classes. In both cases, we show that if the matching induced by the committed interview sets is not stable, then a new coordination trigger occurs after a time-independent number of rounds. Thus, the process not only incurs bounded regret between updates, but also eventually converges to a stable matching. 
\end{proof}

\cref{thm:decregmain}'s bounds are within factors $m$ and $nm$ of optimal for structured and general markets, reflecting the decentralization cost of $\FGS$ over the centralized bound ( see the details in~\cref{opt:coord}).

\vspace{-0.5em}
\subsection{Coordination-Free 
Decentralization
with Anonymous Hiring Changes $\VFP(t)$ Feedback}\label{sec:non-coopmain}

In this section, we design a \emph{coordination-free} decentralized learning algorithm under the richer firm-side feedback $\VFP(t)$. Utilizing globally revealed hiring changes, the algorithm avoids coordination while remaining robust to uncertain firms' strategic rejections, and achieves time-independent regret. We first outline the main idea and technical challenges, then present the key algorithmic components and regret analysis. Detailed descriptions and analysis are deferred to the~\cref{sec:non-coop}.

\textbf{Heterogeneous reactions and instability under anonymity.}
In the absence of coordination, agents react independently to observed hiring changes based on their local estimates, leading to inconsistent updates that can slow or block convergence to stable matching. This is further complicated by anonymous firm-side feedback and strategic rejections, which may induce several misleading signals.
\textbf{Cyclic blocking-pair dynamics.}
Independent reactions may also lead to simultaneous attempts to resolve blocking pairs. In general markets with multiple stable matchings, this can induce persistent cycles and linear regret, requiring additional randomization to break such dynamics.

\noindent \textbf{Algorithmic Design.}
To address the challenges above, we use a simple coordination-free rule for agents, whose analysis is substantially more delicate. Each agent maintains the same rejection-time variable $r_{a,f}(t)$ and forms a candidate set of firms that either have never rejected it or have exhibited a hiring change since their last rejection:
\begin{align}\label{eq:candidate_setsss}
    \CSFP_a(t) \doteq 
    \left\{ f : \exists\, t' \in [r_{a,f}(t), t),\ f \in \VFP(t') \right\},
    \qquad
    \FA_a(t) = \arg\max_{f \in \CSFP_a(t)} \hat{u}_{a,f}(t).
\end{align}
As in~\cref{alg:ancdrr}, the agent interviews $\INT_a(t)=\{\FA_a(t),\RR_a(t)\}$ and applies to $\FA_a(t)$ that suffices for $\alpha$-reducible markets to achieve time independent regret. For general markets, where cyclic blocking-pair resolutions may arise, we use a randomized variant deferred to the appendix (\cref{alg:Eancdrr}) with $k=3$ interviews: $\INT_a(t)=\{\FA_a(t),\FA_a(t-1),\RR_a(t)\}$. The agent then randomizes between applying to $\FA_a(t)$ and $\FA_a(t-1)$ with parameter $\lambda\in(0,1)$.
\begin{algorithm}[t]
\caption{Coordination-Free Decentralized Learning}
\label{alg:ancdrr}
\begin{algorithmic}[1]
    \State \textbf{Input:} $a$, $\FI$, anonymous hiring changes $\VFP$,
    \textbf{Initialize:} $r_{a,f}(1) \gets 0$ for all $f \in \FI$
    \For{$t \in \mathcal{T}$}
        \State Construct candidate set $\CSFP_a(t)$ according to \eqref{eq:candidate_setsss}, 
        $\FA_a(t) \gets \arg\max_{f \in \CSFP_a(t)} \hat{u}_{a,f}(t)$
        \State Interview with $\INT_a(t) = \{\FA_a(t), \RR_a(t)\}$ and apply to $\FA_a(t)$
        \State \textsc{UpdateAgentRejVars}$(a,t)$
    \EndFor
\end{algorithmic}
\end{algorithm}

\noindent \textbf{Regret Guarantee.}
Under the above policies, we obtain time-independent regret for both $\alpha$-reducible and general markets under $\VFP(t)$ (full proofs in~\cref{Sec:genuncertainproofalpha} and~\cref{proof:noncoordgen} respectively).

\begin{restatable}{theorem}{thmnoncodecregmain}\label{thm:noncodecreg}
In a matching market $\MM$, with firms following~\cref{alg:fdrr} and agents following coordination-free ~\cref{alg:ancdrr} (for $\alpha$-reducible markets with $k=2$) and~\cref{alg:Eancdrr} (for general markets with $k=3$), the expected regret is time-independent. In particular, 
{\small\[
\mathbb{E}\!\left[\underbar{R}_{a}(\mathcal{T})\right] \in 
\begin{cases}
O\!((\sum_{j \in [i]} |\mathcal{H}_{a_{j},f_{j}}| + 1)\cdot (\sum_{j\in[i]}(m^2\overline{\Delta}_{a_j}^{-2} + nm\,\underline{\Delta}_{f_j}^{-2}) + 1)), & \text{$\alpha$-reducible $\MM$},\\[4pt]
O\!\left(\epsilon^{-1}n^5m^4\Delta^{-2}\right),\; \hbedit{\text{for $\epsilon = \bigl(\lambda(1-\lambda)^{n-1}\bigr)^{n^4m+nm}$ with $\lambda \in (0,1)$,}} & \text{General Markets}.
\end{cases}
\]}
\end{restatable}

\begin{proof}[Proof Sketch.]
We first prove the result for \(\alpha\)-reducible markets.  Consider consecutive blocks of rounds during which the estimated preference list of an agent $a_i$ remains valid. Within each such block, $a_i$ matches with its unique stable partner $f_i$ after a time-independent number of rounds and remains matched thereafter. Hence, regret is incurred only at the beginning of each valid block. Since the total number of rounds in which estimates are invalid is itself time-independent, the overall regret is time-independent. 
We then extend the proof to general markets. In general markets, the absence of the layered structure and the presence of multiple stable matchings introduce cyclic blocking-pair dynamics, which may persist even under valid estimates. To break such cycles, we use a randomized variant that introduces an additional interview $\FA_a(t-1)$ (i.e., $k=3$). This ensures that, with probability at least $\lambda(1-\lambda)^{n-1}$, exactly one relevant agent applies to its current candidate $\FA_a(t)$ while the remaining agents keep their previous actions $\FA_a(t-1)$, avoiding simultaneous blocking-pair resolutions. Repeating this over a specific sequence of $O(n^4)$ blocking-pair resolutions~\cite{abeledo1995paths} yields progress toward stability. Once such a sequence occurs, the same valid-block argument as above applies, giving the desired time-independent regret bound.
\end{proof}


\textbf{Conclusion.}
We study bandit learning in matching markets with interviews and two-sided uncertainty. We present centralized and decentralized algorithms that achieve time-independent regret with an anonymous firm-side signal. The key ingredient is handling firm uncertainty via an extended action space enabling strategic deferral, which stabilizes decentralized learning under limited feedback and improves prior work along multiple dimensions. Incentive compatibility, optimality, comparisons, and future directions are in Appendixes~\ref{sec:ic},~\ref{sec:remarks1},~\ref{sec:remarks2}, and~\ref{sec:remarks3}.


\section*{Acknowledgments}
We thank Yair Zick\footnote{Department of Computer Science, University of Massachusetts Amherst. Email: \texttt{yzick@umass.edu}} for his invaluable guidance and insightful discussions throughout this project. His feedback and mentorship greatly contributed to shaping the ideas presented in this work.


\bibliographystyle{plainnat}
\bibliography{references}

\appendix




\clearpage            
\onecolumn   

\section{Motivating Applications}
\label{app:digital-platform-motivations}
Our model is motivated by several real-world matching applications. Below, we highlight a few representative examples that capture the main primitives of our model: low-cost pre-match screening, higher-commitment applications, firm-side acceptance or deferral, stochastic rewards, and limited feedback about other participants' outcomes. While our framework is designed to abstract these features, a detailed domain-specific mapping and performance analysis for any particular application would require substantial additional modeling effort and is beyond the scope of this work.

\textbf{Creator--Brand Sponsorship Markets.}
Consider a platform that matches content creators to brands with scarce campaign slots. Each round
$t$ corresponds to a sponsorship cycle: every creator selects a small number of brands to screen or
interview---through product seeding, trial affiliate links, or limited test posts---before committing to one pitch. These interactions are cheaper than a full
sponsorship contract, but they provide noisy signals of match quality to both sides. The creator
learns whether a brand fits their audience and long-run reputation, while the brand learns whether
the creator is likely to deliver strong campaign performance with acceptable content quality and
brand-safety risk. After screening, the creator pitches to one screened brand; the brand then accepts
its currently most preferred applicant or defers by leaving the campaign slot open when the evidence
is weak. A successful match yields stochastic value to both sides, such as compensation and audience
fit for the creator and predicted campaign performance or product--audience alignment for the brand.
The learning goal is to approach a stable benchmark under the true two-sided preferences while using
only a constant number of low-cost pre-match observations per round.

\textbf{Reciprocal Recommender and Recruiting Platforms.}
Consider a platform that recommends bilateral connections, such as online dating, professional
networking, or recruiting~\citep{hitsch2010matching,arnosti2021managing,rios2023improving}.
Agents may be users, job seekers, mentees, or candidates, while firms may be potential partners or employers. Each round $t$ corresponds to a search or recruiting cycle: every agent screens a small number of potential counterparties---through profile views, likes, resume screens, introductory calls, or technical screens---before committing to one serious request or application. These screening interactions are lower cost than a final match, but they provide noisy signals of match quality to both sides. The agent learns which counterparty best fits their preferences, career goals, or compatibility, while the firm learns which applicant is most promising given its own preferences, capacity constraints, and expected match quality. After screening (interview), the agent applies to one screened counterparty; the firm then accepts its currently most preferred applicant or defers, for example, by keeping a
slot open, or continuing to search. A successful match yields a stochastic value to both sides. The learning goal is again to approach a stable benchmark under the true two-sided preferences using only a constant number of pre-match observations per round, with coarse feedback such as whether a slot remains open or whether a counterparty's match status changed.


\textbf{Academic Research Internships.}
Consider a market in which research labs (agents) place their students and postdocs into industry internship positions (firms) each cycle. Each round $t$ corresponds to a hiring cycle: every research lab sends one of its members to interview at a small number of industry positions---through campus visits, informational calls, or trial projects---before the member commits to a single application. These interviews serve as low-cost explorations that yield noisy signals of match quality for both sides, without constituting a formal commitment. After interviewing, the member applies to one position; the industry host then decides whether to extend an offer to its most preferred applicant or defer if uncertain about the fit. A successful match yields a reward reflecting the quality of the collaboration for both sides. Importantly, it is the lab as an institution---not just any individual member---that bears the objective: across cycles, different members may be placed, but the lab accumulates experience and refines its understanding of which industry partners best complement its research agenda. Performance is measured via regret: the gap between the cumulative reward the lab would have obtained by always matching with its best stable industry partner (in hindsight) and its actual accumulated reward over $T$ cycles. The goal is to minimize this regret while using only a constant number of interviews per cycle, operating in a decentralized manner with limited feedback about other labs' outcomes. 

\section{General Notation, Lemmas, and Observations for the Regret Analysis}\label{appx:gen-notation}
In this section, we introduce general notation and collect auxiliary lemmas and observations used throughout our proofs, in particular in the proofs of~\cref{thm:noncodecreg} and~\cref{thm:decregmain} for unstructured general markets.
\subsection{Top-$k$ Ground-Truth and Estimated Preferences}\label{appx:topk}
We define, for each agent (resp., firm), the set of its top-$k$ firms (resp., agents) under the ground-truth and estimated preference lists.

\begin{definition}[Top-$k$ agents and firms]\label{def:topnfirms}
In a matching market $\MM$, for each agent $a\in\AG$ and integer $1\le k\le m$, let $\FI^{(k)}_a$ denote the top-$k$ firms in $a$'s ground-truth list $\PL_a$, and let $\hat{\FI}^{(k)}_a(t)$ denote the top-$k$ firms under $a$'s estimated list $\EPL_a(t)$ at time $t$. Formally,
\begin{align}
    \FI^{(k)}_a
    &\doteq
    \left\{
        f \in \FI :
        \left|\mathcal{H}_{a,f}\right| \le k-1
    \right\},\\
    \hat{\FI}^{(k)}_a(t)
    &\doteq
    \left\{
        f \in \FI :
        \left|\hat{\mathcal{H}}_{a,f}(t)\right| \le k-1
    \right\}.
\end{align}
Similarly, for each firm $f\in\FI$ and integer $1\le k\le n$, let $\AG^{(k)}_f$ denote the top-$k$ agents in $f$'s ground-truth list $\PL_f$, and let $\hat{\AG}^{(k)}_f(t)$ denote the top-$k$ agents under $f$'s estimated list $\EPL_f(t)$ at time $t$. Formally,
\begin{align}
    \AG^{(k)}_f
    &\doteq
    \left\{
        a \in \AG :
        \left|\mathcal{H}_{f,a}\right| \le k-1
    \right\},\\
    \hat{\AG}^{(k)}_f(t)
    &\doteq
    \left\{
        a \in \AG :
        \left|\hat{\mathcal{H}}_{f,a}(t)\right| \le k-1
    \right\}.
\end{align}
\end{definition}

\subsection{Rounds of Top-$k$ Alignment}\label{appx:topk-alignment}
We next define the rounds in which estimated preferences agree with the ground truth up to the top-$k$ items. We first introduce an agreement operator and then relate it to validity events.

\begin{definition}[Top-$k$ set agreement]\label{def:topksetoperator}
Fix an integer $k$. For an agent $a\in\AG$ and time $t$, we write
\[
\hat{\FI}^{(k)}_a(t)\ \equiv\ \FI^{(k)}_a
\]
to mean that the top-$k$ firms under $\EPL_a(t)$ and $\PL_a$ agree both element-wise and order-wise, i.e., $\hat{\FI}^{(k)}_a(t)=\FI^{(k)}_a$ and $\EPL_a(t)$ and $\PL_a$ induce the same order on $\FI^{(k)}_a$. Similarly, for a firm $f\in\FI$ and time $t$, we write
\[
\hat{\AG}^{(k)}_f(t)\ \equiv\ \AG^{(k)}_f
\]
to mean that the top-$k$ agents under $\EPL_f(t)$ and $\PL_f$ agree both element-wise and order-wise, i.e., $\hat{\AG}^{(k)}_f(t)=\AG^{(k)}_f$ and $\EPL_f(t)$ and $\PL_f$ induce the same order on $\AG^{(k)}_f$.
\end{definition}

\begin{lemma}\label{lem:truthiffvalid}
Fix an integer $k$. For any agent $a\in\AG$ and time $t$, we have $\hat{\FI}^{(k)}_a(t)\equiv \FI^{(k)}_a$ iff
\[
t \in \bigcap_{f\in \FI^{(k)}_a}\mathcal{E}_{a,f}.
\]
Likewise, for any firm $f\in\FI$ and time $t$, we have $\hat{\AG}^{(k)}_f(t)\equiv \AG^{(k)}_f$ iff
\[
t \in \bigcap_{a\in \AG^{(k)}_f}\mathcal{E}_{f,a}.
\]
\end{lemma}

\begin{proof}
We prove the claim for agents; the proof for firms is analogous.

First, assume $\hat{\FI}^{(k)}_a(t)\equiv \FI^{(k)}_a$. By definition of $\equiv$, for every $f\in \FI^{(k)}_a$ the relative order of the firms above $f$ is correct under $\EPL_a(t)$. Equivalently, for each such $f$ we have $\hat{\mathcal{H}}_{a,f}(t)=\mathcal{H}_{a,f}$, so $\EPL_a(t)$ is valid with respect to $f$ at time $t$, i.e., $t\in \mathcal{E}_{a,f}$. Hence $t\in \bigcap_{f\in \FI^{(k)}_a}\mathcal{E}_{a,f}$.

Conversely, assume $t\in \bigcap_{f\in \FI^{(k)}_a}\mathcal{E}_{a,f}$ but $\hat{\FI}^{(k)}_a(t)\not\equiv \FI^{(k)}_a$. Then there exists some $f'\in \FI^{(k)}_a$ such that $\EPL_a(t)$ and $\PL_a$ induce different orders on the relevant set $\mathcal{L}_{a,f'}$. Equivalently, there exists $f''\in \mathcal{L}_{a,f'}$ with $f''\in \hat{\mathcal{H}}_{a,f'}(t)$, which implies $t\in \Bar{\mathcal{E}}_{a,f'}$, contradicting $t\in \mathcal{E}_{a,f'}$. Therefore $\hat{\FI}^{(k)}_a(t)\equiv \FI^{(k)}_a$.
\end{proof}

By \cref{lem:truthiffvalid}, top-$k$ alignment can be characterized in terms of validity events. We therefore define the rounds in which \emph{all} agents (resp., \emph{all} firms) have top-$k$ alignment.

\begin{definition}[Rounds of global top-$k$ alignment]\label{def:globaltopkalign}
Fix an integer $k$. We define $\Gamma^{(k)}_\AG$ (resp., $\Gamma^{(k)}_\FI$) as the set of time steps at which \emph{all} agents (resp., \emph{all} firms) have their top-$k$ preferences aligned with the ground truth (element-wise and order-wise). Formally,
\begin{align}
\Gamma^{(k)}_\AG 
&\doteq 
\left\{
t \in \mathcal{T} :
\forall a \in \AG,\;
t \in \bigcap_{f \in \FI^{(k)}_a} \mathcal{E}_{a,f}
\right\},\label{def:agenttopnsort}\\
\Gamma^{(k)}_\FI
&\doteq 
\left\{
t \in \mathcal{T} :
\forall f \in \FI,\;
t \in \bigcap_{a \in \AG^{(k)}_f} \mathcal{E}_{f,a}
\right\}.\label{def:firmsort}
\end{align}
\end{definition}

In particular, when $k$ equals the full list size, top-$k$ alignment reduces to exact recovery of the entire preference lists.

\begin{observation}\label{obs:globaltruth}
We have $t\in \Gamma^{(m)}_\AG \cap \Gamma^{(n)}_\FI$ if and only if every agent and every firm has fully learned its preference list at time $t$, i.e., $\EPL_a(t)=\PL_a$ for all $a\in\AG$ and $\EPL_f(t)=\PL_f$ for all $f\in\FI$. Equivalently, $\Gamma^{(m)}_\AG \cap \Gamma^{(n)}_\FI$ is exactly the set of rounds at which the estimated market induced by $\EPL(t)$ coincides with the ground-truth matching market induced by $\PL$.
\end{observation}

The last one is an observation for the strategic rejection policy~\cref{alg:fdrr}.

\begin{observation}\label{obs:rejpol}
For any certain firm $f$, we have $t\notin\mathcal{S}_f$ for all $t$, so under~\cref{alg:fdrr}$,\ \NH_f(t)=1$ in every round.
\end{observation}
\section{Rejection Variables Update Pseudo-Codes (Algorithms~\ref{alg:fdrr},~\ref{alg:drr},\ref{alg:ancdrr}, and~\ref{alg:Eancdrr})}\label{appx:missing-pseudocodes}

We collect two short subroutines for updating the \emph{rejection variables} used throughout our decentralized algorithms. 
\cref{alg:frejvar} updates the firm-side private state: $r_{f,a}(t)$ is the last round before $t$ in which firm $f$ rejected agent $a$ while hiring someone else, and $c_f(t)$ is the last round before $t$ in which $f$ remained vacant. 
\cref{alg:adrejvar} updates the agent-side rejection time $r_{a,f}(t)$, the last round before $t$ in which agent $a$ applied to firm $f$, remained unmatched, and $f$ did \emph{not} appear vacant (i.e., $f$ hired another agent).

\floatname{algorithm}{Subroutine}
\begin{algorithm}[t]
\caption{\textsc{UpdateFirmRejVars}$(f,t)$}
\label{alg:frejvar}
\begin{algorithmic}[1]
    \If{$t \notin \mathcal{S}_f$}
        \State $r_{f,a}(t+1) \gets t$ for all $a \in \APP_f(t) \setminus \{\FFM_f(t)\}$
        \State $c_f(t+1) \gets c_f(t)$
    \Else
        \State $c_f(t+1) \gets t$ \Comment{update latest vacancy}
    \EndIf
\end{algorithmic}
\end{algorithm}

\floatname{algorithm}{Subroutine}
\begin{algorithm}[t]
\caption{\textsc{UpdateAgentRejVars}$(a,t)$}
\label{alg:adrejvar}
\begin{algorithmic}[1]
    \For{all firms $f \in \FI$}
        \State $r_{a,f}(t+1) \gets r_{a,f}(t)$
    \EndFor
    \If{$\FM_a(t)=\emptyset$ \textbf{and} $\FA_a(t)\notin \VF(t)$}
        \State $r_{a,\FA_a(t)}(t+1) \gets t$ \Comment{rejected in favor of another hire}
    \EndIf
\end{algorithmic}
\end{algorithm}


\section{Deferred Concepts from the Model (Section~\ref{sec:model}) and Preliminaries (Section~\ref{sec:preliminaries})}\label{appx:prel}
This section collects definitions and technical details deferred from Sections~\ref{sec:model} and~\ref{sec:preliminaries} to streamline the main presentation.

\subsection{Definition of Local Observations}
We define the local observations available to each agent which determine the information used to update estimates and make decisions in our decentralized algorithms.

\begin{definition}[Local Observations]\label{def:decentralized}
     In the \emph{decentralized} setup, each agent $a \in \AG$ updates its interviewing set $\INT_a(t)$ solely based on its own local observations at time $t$, which include:
     \begin{enumerate}
          \item Changes in its estimated preference list $\EPL_a(t-1)$ to $\EPL_a(t)$.
          \item Whether it was rejected by the firm $\FA_a(t-1)$.
          \item The firm-side feedback $\VF(t-1)$.
     \end{enumerate}
\end{definition}

\subsection{Definition of Convergence}
We formalize the notion of convergence used throughout the paper, specifying when the sequence of matchings stabilizes under our learning dynamics.

\begin{definition}[Convergence of a Learning Algorithm]\label{def:convergence}
A learning algorithm $\pi$ is said to \emph{converge} to a perfect matching $(a,\FM_a(t))_{a \in \AG}$ if there exists a time step $t \in [T]$ such that for all $t' \ge t$, we have $\FM_a(t') = \FM_a(t)$ and $\FM_a(t') \neq \emptyset$ for every agent $a \in \AG$.
\end{definition}

\section{Deferred Concepts and Proofs from Section~\ref{sec:centralized}}\label{appx:cent}
This section collects additional definitions and deferred proofs from Section~\ref{sec:centralized} to streamline the main presentation.

\subsection{Pseudocode for the Centralized Algorithm}\label{sec:cent-pseudo}
Here we provide the pseudocode of the centralized algorithm referenced in Section~\ref{sec:centralized}.
\floatname{algorithm}{Algorithm}
\begin{algorithm}[h]
\caption{Centralized Interview Allocation}
\label{alg:ciarr}
\begin{algorithmic}[1]
    \State \textbf{Input:} $\AG$, $\FI$, $\mathcal{T}$
    \For{$t \in \mathcal{T}$}
        \State $\CIAP$ collects the current estimated preference lists $\{\EPL_a(t)\}_{a \in \AG}$ and $\{\EPL_f(t)\}_{f \in \FI}$
        \State $\CIAP$ runs $\FGS\!\left(\{\EPL_a(t)\}_{a \in \AG}, \{\EPL_f(t)\}_{f \in \FI}\right)$ and sets $\FA_a(t)$ to be agent $a$'s match in the resulting estimated stable matching
        \State $\CIAP$ assigns the interview set $\INT_a(t) = \{\FA_a(t), \RR_a(t)\}$ to each agent $a \in \AG$
        \State Each agent interviews the firms in $\INT_a(t)$ and observes the corresponding feedback \label{line:alg11}
        \State Each agent $a \in \AG$ observes $X_{a,\FA_a(t)}(t)$
        \State Agents update $\hat{u}_{a,\FA_a(t)}(t+1)$ and $\hat{u}_{a,\RR_a(t)}(t+1)$
        \State Each firm $f \in \FI$ observes $X_{f,a}(t)$ and $X_{a,\RR_a(t)}(t)$ from interviewed agents $a \in \INT_f(t)$
        \State Agents apply to $\FA_a(t)$ and receive reward $X_{a,\FM_a(t)}(t)$ \label{line:alg12}
        \State Each firm admits its most preferred applicant $\FFM_f(t)$ \label{line:alg13}
    \EndFor
\end{algorithmic}
\end{algorithm}
\subsection{Proof of an Auxiliary Lemma}\label{sec:aux-lemma-order}
We first state and prove a simple auxiliary bound used in the regret analysis.

\begin{lemma}\label{lemma:order}
Let $a > 0$, $c > 0$, and $\Delta > 0$ be constants, and define
\(
R(m,\Delta) \;=\; \frac{a\, e^{-\frac{\Delta^2}{c m}}}{1 - e^{-\frac{\Delta^2}{c m}}}.
\)
Then
\[
R(m,\Delta) \in O\!\left(\frac{m}{\Delta^2}\right)
\quad \text{for } m>0.
\]
\end{lemma}

\begin{proof}
Let $x = \frac{\Delta^2}{c m}$. Then $x>0$ and we can rewrite
\[
R(m,\Delta)
= \frac{a\, e^{-x}}{1 - e^{-x}}
= \frac{a}{e^x - 1}.
\]

Using the standard inequality valid for all $x>0$,
\(
e^x - 1 \;\ge\; x,
\)
we obtain
\[
R(m,\Delta)
\le \frac{a}{x}
= \frac{a c m}{\Delta^2}.
\]

Thus, there exists a constant $K = ac$ such that
\[
R(m,\Delta) \le K \cdot \frac{m}{\Delta^2},
\]
and hence
\(
R(m,\Delta) \in O\!\left(\frac{m}{\Delta^2}\right).
\)
\end{proof}

\subsection{Proof of Lemma~\ref{lemma:valid-rr}}\label{sec:proof-valid-rr}
We next present the proof of Lemma~\ref{lemma:valid-rr}, which is a key ingredient in our time-independent regret analysis.
\validrr*
\begin{proof}
    By union bounding over firms and applying Hoeffding's inequality to estimation errors, we obtain
            \begin{align}
                 & \mathbb{E}[|\Bar{\mathcal{E}}_{a,\overline{f^*_a}}|] = \sum_{t=1}^T \Pr\!\left(\exists f' \in \mathcal{L}_{a,\overline{f^*_a}}:\ f' \in \hat{\mathcal{H}}_{a,\overline{f^*_a}}(t)\  \right), \nonumber                                                                 \\
                 & = \sum_{t=1}^T \Pr\!\left(\exists f' \in \mathcal{L}_{a,\overline{f^*_a}}:\ \hat{u}_{a, f'}(t) > \hat{u}_{a, \overline{f^*_a}}(t) \right), \nonumber                                                              \\
                 & \leq \sum_{t=1}^T \sum_{f' \in \mathcal{L}_{a,\overline{f^*_a}}} \Pr\!\left( \hat{u}_{a, f'}(t) - u_{a,f'} + u_{a,\overline{f^*_a}} - \hat{u}_{a,\overline{f^*_a}}(t) > \overline{\Delta}_{a,f'} \right),\nonumber \\
                 & \leq \sum_{t=1}^T \sum_{f' \in \mathcal{L}_{a,\overline{f^*_a}}}
                \Pr\left(\big|\hat{u}_{a,f'}(t) - u_{a,f'}\big| > \tfrac{\overline{\Delta}_{a,f'}}{2}\right) + 
            \Pr\left(\big|u_{a,\overline{f^*_a}} - \hat{u}_{a,\overline{f^*_a}}(t)\big| > \tfrac{\overline{\Delta}_{a,f'}}{2}\right),\nonumber                                                                                                           \\
                 & \overset{(a)}{\leq} \sum_{t=1}^T \sum_{f' \in \mathcal{L}_{a,\overline{f^*_a}}} 4 \exp\!\left(-\frac{t \overline{\Delta}_{a,f'}^2}{2m}\right),\nonumber\\
                 &\leq \frac{4|\mathcal{L}_{a,\overline{f^*_a}}| \exp\!\left(-\tfrac{(\overline{\Delta}_{a})^2}{2m}\right)}{1 - \exp\!\left(-\tfrac{(\overline{\Delta}_{a})^2}{2m}\right)}, \nonumber\\
                 &\leq 8m\,|\mathcal{L}_{a,\overline{f^*_a}}| \cdot \left(\overline{\Delta}_{a}\right)^{-2}
                 \in O\!\left(
                 m\,|\mathcal{L}_{a,\overline{f^*_a}}| \cdot \left(\overline{\Delta}_{a}\right)^{-2}
                 \right),
                 \label{eq:validagentrr}
            \end{align}
    where $\overline{\Delta}_{a} \doteq \min_{f \neq \overline{f^*_a}} \overline{\Delta}_{a,f}$ denotes the minimum optimal utility gap. Step (a) follows from Hoeffding's inequality. Step (b), implied by~\cref{lemma:order}, shows that the expected number of invalid time steps for agent $a$ is bounded by a constant independent of $T$, completing the proof.
\end{proof}

\subsection{Proof of Theorem~\ref{thm:ciarr}}\label{sec:proof-ciarr}
\thmciarr*
\begin{proof}
We track invalidity on both sides.  
For each agent $a\in\AG$, let $\Bar{\mathcal{E}}_{a,\overline{f^*_a}}$ be the rounds in which $\EPL_a(t)$ is invalid w.r.t.\ $\overline{f^*_a}$, and for each firm $f\in\FI$, let $\Bar{\mathcal{E}}_{f,\underline{a^*_f}}$ be the rounds in which $\EPL_f(t)$ is invalid w.r.t.\ $\underline{a^*_f}$ (defined analogously). Define the globally invalid rounds by
\[
\Bar{\mathcal{E}}
~\doteq~
\Bigl(\bigcup_{a\in\AG}\Bar{\mathcal{E}}_{a,\overline{f^*_a}}\Bigr)\ \cup\
\Bigl(\bigcup_{f\in\FI}\Bar{\mathcal{E}}_{f,\underline{a^*_f}}\Bigr),
\qquad
\mathcal{E}\doteq [T]\setminus \Bar{\mathcal{E}}.
\]
Decompose the expected regret as
\begin{align}
\mathbb{E}[\overline{R}_a(T)]
&= \sum_{t\in[T]} u_{a,\overline{f^*_a}}
-\Bigl(\sum_{t\in\mathcal{E}}\mathbb{E}[X_{a,\FM_a(t)}(t)]
      +\sum_{t\in\Bar{\mathcal{E}}}\mathbb{E}[X_{a,\FM_a(t)}(t)]\Bigr).
\label{eq:sepreg-unified}
\end{align}

\noindent\textbf{Valid rounds ($t\in\mathcal{E}$).}
Fix any $t\in\mathcal{E}$ and consider the agent-optimal stable matching under the ground-truth preferences, $(a,\overline{f^*_a})_{a\in\AG}$.
If this matching had a blocking pair under the estimated preferences at time $t$, then the same pair would also block under the ground-truth preferences: on $\mathcal{E}$, any agent (resp., firm) that is ranked above the relevant stable partner in the estimated list is also ranked above it in the ground truth. This contradicts stability of $(a,\overline{f^*_a})_{a\in\AG}$ under $\PL$.
Hence $(a,\overline{f^*_a})_{a\in\AG}$ is also stable under $(\EPL_a(t),\EPL_f(t))$.

Since $\CIAP$ runs $\FGS$ on the estimated market at time $t$, it returns the agent-optimal stable matching w.r.t.\ $(\EPL_a(t),\EPL_f(t))$.
Therefore, for every agent $a$, the resulting match $\FM_a(t)$ is either $\overline{f^*_a}$ or a firm that $a$ ranks strictly above $\overline{f^*_a}$ in $\EPL_a(t)$; by validity, this firm is also strictly preferred in the ground truth.
Thus agents incur zero regret on all rounds $t\in\mathcal{E}$.

\noindent\textbf{Invalid rounds ($t\in\Bar{\mathcal{E}}$).}
By a union bound,
\begin{align}
\mathbb{E}[|\Bar{\mathcal{E}}|]
~\le~
\sum_{a\in\AG}\mathbb{E}\!\left[\left|\Bar{\mathcal{E}}_{a,\overline{f^*_a}}\right|\right]
+\sum_{f\in\FI}\mathbb{E}\!\left[\left|\Bar{\mathcal{E}}_{f,\underline{a^*_f}}\right|\right].
\label{eq:union-invalid}
\end{align}
Lemma~\ref{lemma:valid-rr} gives, for every $a\in\AG$,
\(
\mathbb{E}\!\left[\left|\Bar{\mathcal{E}}_{a,\overline{f^*_a}}\right|\right]
\in O\!\left(m|\mathcal{L}_{a,\overline{f^*_a}}|\cdot \overline{\Delta}_{a}^{-2}\right),
\)
where $\overline{\Delta}_{a} \doteq \min_{f \neq \overline{f^*_a}} \overline{\Delta}_{a,f}$.
Moreover, since each (strategic uncertain) firm participates in round-robin interviews, the same argument yields the firm-side analogue
\begin{align}
\mathbb{E}\!\left[\left|\Bar{\mathcal{E}}_{f,\underline{a^*_f}}\right|\right]
\in O\!\left(m|\mathcal{L}_{f,\underline{a^*_f}}|\cdot \underline{\Delta}_{f}^{-2}\right),
\label{eq:firmvalidrr}
\end{align}
where $\underline{\Delta}_{f} \doteq \min_{a \neq \underline{a^*_f}} \underline{\Delta}_{f,a}$.
Using $|\mathcal{L}_{a,\overline{f^*_a}}|\le m$ and $|\mathcal{L}_{f,\underline{a^*_f}}|\le n$, we obtain $\mathbb{E}[|\Bar{\mathcal{E}}|]\in O\left(nm^2\min\left\{\overline{\Delta}_{\AG},\underline{\Delta}_{\FI}\right\}^{-2}\right)$,
where $\overline{\Delta}_{\AG} \doteq \min_{a\in\AG}\overline{\Delta}_{a}$ and $\underline{\Delta}_{\FI} \doteq \min_{f\in\FI}\underline{\Delta}_{f}$.
Since regret is zero on $\mathcal{E}$ and per-round regret is at most $1$, we have $\mathbb{E}[\overline{R}_a(T)] \le \mathbb{E}[|\Bar{\mathcal{E}}|]\in O\left(nm^2\min\left\{\overline{\Delta}_{\AG},\underline{\Delta}_{\FI}\right\}^{-2}\right)$, completing the proof.
\end{proof}

\section{Deferred Concepts and Proofs from Section~\ref{sec:decwithvfmain}}\label{appx:decwithvf}
This section collects extensive explanation of the algorithm followed by additional definitions and deferred proofs from Section~\ref{sec:decwithvfmain}.

\subsection{Necessary Coordination under $\VF(t)$ compared to $\VFP(t)$}\label{sec:movevtovp}
Under the only-vacancy firm-side feedback $\VF(t)$, agents do not observe firms’ hiring updates---information that is available under $\VFP(t)$. Since $\VF(t)$ only announces vacancies, agents must coordinate when adjusting their strategies. Thus, while $\VFP(t)$ allows independent adaptation, the limited feedback in $\VF(t)$ requires coordinated distributed executions of $\FGS$. Example~\ref{examp:coordfgs} demonstrates why such coordination becomes essential under $\VF(t)$ and clarifies the form of interaction our algorithm adopts.

\begin{example}\label{examp:coordfgs}
Consider a market with agents \( \mathcal{A} = \{a_1, a_2, a_3\} \) and certain firms \( \mathcal{F} = \{f_1, f_2, f_3\} \), whose true preferences are:
\[
    \begin{array}{lr}
        \PL_{a_1}: \langle f^\dagger_1 \succ f_2 \succ f_3 \rangle & 
        \PL_{f_1}: \langle a^\dagger_1 \succ a_2 \succ a_3 \rangle \\[3pt]
        \PL_{a_2}: \langle f^\dagger_2 \succ f_1 \succ f_3 \rangle & 
        \PL_{f_2}: \langle a_1 \succ a^\dagger_2 \succ a_3 \rangle \\[3pt]
        \PL_{a_3}: \langle f_1 \succ f^\dagger_3 \succ f_2 \rangle & 
        \PL_{f_3}: \langle a_1 \succ a_2 \succ a^\dagger_3 \rangle
    \end{array}
\]
The unique stable matching is \((a_1,f_1), (a_2,f_2), (a_3,f_3)\), as indicated by daggers. Assume agent $a_2$ has already learned its true preferences.  
At some time step \(t\), the estimated preference lists are:
\[
    \begin{array}{lr}
        \EPL_{a_1}(t): \langle f_2 \succ f_1 \succ f_3 \rangle & I_{a_1}(t) = \{f_2, \RR_{a_1}(t)\} \\[3pt]
        \EPL_{a_2}(t): \langle f_2 \succ f_1 \succ f_3 \rangle & I_{a_2}(t) = \{f_1, \RR_{a_2}(t)\} \\[3pt]
        \EPL_{a_3}(t): \langle f_3 \succ f_1 \succ f_2 \rangle & I_{a_3}(t) = \{f_3, \RR_{a_3}(t)\}
    \end{array}
\]
Since $a_1$ is the most preferred by all firms, it applies to $\FA_{a_1}(t)=f_2$, which causes $a_2$ to be rejected by $f_2$ when it considers $a_2$.
Consequently, $a_2$ proceeds to consider $f_1$, and since $a_2$ is the second most preferred agent by all firms, the only remaining option for $a_3$ is to consider $f_3$.

Now suppose that at time step $t' > t$, the updated estimates become:
\[
    \begin{array}{lr}
        \EPL_{a_1}(t'): \langle f_3 \succ f_1 \succ f_2 \rangle & I_{a_1}(t') = \{f_3, \RR_{a_1}(t')\} \\[3pt]
        \EPL_{a_2}(t'): \langle f_2 \succ f_1 \succ f_3 \rangle & I_{a_2}(t') = \{f_1, \RR_{a_2}(t')\} \\[3pt]
        \EPL_{a_3}(t'): \langle f_2 \succ f_3 \succ f_1 \rangle & I_{a_3}(t') = \{f_1, \RR_{a_3}(t')\}
    \end{array}
\]
Here, $a_1$ switches to applying to $f_3$, again causing other applicants to be rejected due to being the most preferred agent. Meanwhile, $a_3$, which was never rejected by $f_2$, now reconsiders $f_2$ as its most preferred estimated firm at that round.
Since $a_1$ and $a_3$ simultaneously change their targets, and $a_2$ observes no hiring updates under $\VF(t')$ (only vacancies), it may assume that $f_2$ remains unavailable and commit to the suboptimal firm $f_1$.

Under $\VFP(t')$, the hiring update at $f_2$ would be visible to $a_2$, prompting it to reconsider $f_2$ and avoid a suboptimal match.
Hence, the limited feedback in $\VF(t')$ prevents agents from detecting such preference updates, necessitating coordinated distributed runs of $\FGS$.
If, for instance, $f_1$---upon detecting a change in its local estimates $\EPL_{f_1}(t')$---could signal others to initiate a joint distributed $\FGS$, then $a_2$ would reassess $f_2$, ensuring convergence to a stable matching.

Otherwise, $a_2$ must continually reconsider firms it was previously rejected from (with some positive probability), and in some scenarios this prevents achieving $O(1)$ regret.
\end{example}

\subsection{Extended Coordinated Decentralized Algorithm with Vacancy-Only $\VF(t)$ Feedback }\label{sec:decwithvf}
We now design decentralized learning algorithms under the only-vacancy feedback $\VF(t)$, which agents use as the coordination signal. Unlike $\VFP(t)$, which reveals hiring changes, $\VF(t)$ only announces vacancies, so important market transitions can remain invisible to agents not directly involved (see~\cref{examp:coordfgs}); this weaker feedback makes learning more challenging. Accordingly, under $\VF(t)$ we rely on explicit coordination to ensure stability and avoid linear regret; a formal justification is deferred to \cref{sec:movevtovp}.

We next describe a learning framework in which agents alternate between independent behavior and deliberate signaling through their application decisions; these signals shape the vacancy feedback and trigger coordinated updates.

\noindent\textbf{Learning Framework.}
 The time horizon is divided into two alternating stages: \emph{committing} and \emph{updating}. In a {committing} stage, agents adhere to their current interviewing sets $\INT_a(t)$ and act independently. In an {updating} stage, agents coordinate to run a distributed execution of $\FGS$ on their current estimated preference lists in order to revise their interviewing decisions. 
 
Each agent $a\in\AG$ maintains a binary state $\PH_a(t)$, where $\PH_a(t)=1$ indicates a committing stage and $\PH_a(t)=0$ indicates an updating stage. This state is updated at the end of each round after observing $\VF(t)$. Since $\VF(t)$ is globally revealed, all agents update their states identically, synchronizing transitions between the two stages. Formally, we define the sets of updating and committing rounds as
$\mathcal{U}\doteq\{\,t\in\mathcal{T}:\forall a\in\AG,\ \PH_a(t)=0\,\}$ and
$\mathcal{C}\doteq\{\,t\in\mathcal{T}:\forall a\in\AG,\ \PH_a(t)=1\,\}$.

The central challenge is deciding when to switch stages. Transitions from $\rho_a(t)=1$ (committing) to $\rho_a(t)=0$ (updating) occur in two ways.  
First, an agent may detect a market change and intentionally refrain from applying, creating a vacancy that serves as a coordination signal via the firm-side feedback.  
Second, a strategic firm may abstain from hiring under~\cref{alg:fdrr} and appear in $\VF(t)$. Since agents do not observe the source of new vacancies, both events are treated as coordination signals. Conversely, switching back is deterministic: agents wait $O(n^2)$ rounds for the distributed $\FGS$ run and return to committing. The procedure is detailed in~\cref{alg:drr}.

\begin{algorithm}[h]
\caption{Coordinated Decentralized Learning (agent $a$)}
\label{alg:drr}
\begin{algorithmic}[1]
    \State \textbf{Input:} $a$, $\FI$, $\VF$, switching conditions $\mathcal{S}_a^1,\mathcal{S}_a^2,\mathcal{S}_a^3$
    \State \textbf{Initialize:} $t \gets 1$, $\TGS \gets 1$, $r_{a,f}(1) \gets 0$ for all $f \in \FI$
    \While{$t \leq T$}
        \State Construct the candidate set $\CSV_a(t)$ according to \eqref{def:candset2}
        \If{$t \in [\TGS, \TGS + 3n^2)$}
            \State $\PH_a(t) \gets 0$ \Comment{updating stage}
            \State $\FA_a(t) \gets \arg\max_{f \in \CSV_a(t)} \hat{u}_{a,f}(\TGS)$
            \State $\INT_a(t) \gets \{\FA_a(t), \RR_a(t)\}$ and interview
        \ElsIf{$t \notin \mathcal{S}^{\mathrm{inc}}_a \cup \mathcal{S}^{\mathrm{rej}}_a \cup \mathcal{S}^{\mathrm{vac}}_a$}
            \State $\PH_a(t) \gets 1$ \Comment{committing stage}
            \State $\FA_a(t) \gets \arg\max_{f \in \CSV_a(t)} \hat{u}_{a,f}(t)$
            \State $\INT_a(t) \gets \{\FA_a(t), \RR_a(t)\}$ and interview
        \Else
            \State $\PH_a(t+1) \gets 0$ \Comment{switching to updating stage}
            \State $\TGS \gets t+1$
            \State $r_{a,f}(t+1) \gets 0$ for all $f \in \FI$ \Comment{resetting}
            \State $\FA_a(t) \gets \emptyset$ \Comment{not applying to signal others}
        \EndIf
        \State Apply to $\FA_a(t)$
        \State \textsc{UpdateAgentRejVars}$(a,t)$ \Comment{Subroutine~\ref{alg:adrejvar}}
        \State $t \gets t+1$
    \EndWhile
\end{algorithmic}
\end{algorithm}
At each round $t$, agent $a$ interviews $\INT_a(t)=\{\FA_a(t),\RR_a(t)\}$, where $\RR_a(t)$ is chosen via round-robin and $\FA_a(t)$ is the firm to which $a$ applies. We next describe how agents choose their applied firms and interview sets within each stage, and how transitions between committing and updating are triggered. 
Each agent $a$ maintains a private rejection-time variable $r_{a,f}(t)$ for each firm $f$, recording the most recent round $t'<t$ in which $a$ was \emph{non-strategically rejected} by $f$:

{\small
\setlength{\abovedisplayskip}{-12pt}
\begin{align}
r_{a,f}(t) \doteq \argmax_{t'<t}\left\{ \NH_f(t') = 1\land  \begin{aligned}
            &\FA_a(t') = f\\
            &\FM_a(t') = \emptyset 
        \end{aligned} \ \right\}.
\label{eq:rejtime}
\end{align}
}
Thus, if $a$ applies to $\FA_a(t)=f$ and remains unmatched, it updates $r_{a,f}(t+1)=t$ whenever $f\notin \VF(t)$.
We also maintain a global variable $\TGS$ denoting the first round of the most recent updating stage, i.e., the start of the current $\FGS$ execution. This value is common to all agents and is updated whenever they switch from committing to updating. 

\noindent \textbf{Interview Set During Updating.}
The goal of the updating phase is to execute a fresh coordinated run of $\FGS$ using a fixed snapshot of the agents’ current estimated preference lists. At the beginning of an updating stage, all agents reset their rejection variables by setting $r_{a,f}(\TGS)=0$ for every firm $f$ so that past rejections do not affect the new run of $\FGS$. Throughout the updating phase, agents coordinate to execute a distributed version of $\FGS$ and make all interviewing decisions based on their estimated preference lists fixed at time $\TGS$, namely $\hat{u}_{a,f}(\TGS)$. In contrast, firms continue to follow~\cref{alg:fdrr} and make hiring decisions using their current estimates at each round.

To coordinate a new run of $\FGS$, agents enter an updating phase of fixed length $3n^2$ rounds starting at $\TGS$. During this phase, all agents follow a synchronized distributed execution of $\FGS$. Specifically, for all
$t \in \{\TGS, \TGS+1, \ldots, \TGS+3n^2-1\}$,
each agent $a$ selects its applied firm from candidate set
\begin{align}\label{def:candset2}
\CSV_a(t) \doteq \{ f \in \FI : r_{a,f}(t) < \TGS \},
\end{align}
which consists of all firms that have not rejected agent $a$ since the beginning of the current updating phase. The agent then chooses the most preferred firm in $B'_a(t)$ according to its estimated preferences at time $\TGS$, formally as $
\FA_a(t) = \argmax_{f \in \CSV_a(t)} \hat{u}_{a,f}(\TGS)
$.

If agent $a$ is rejected because firm $\FA_a(t)$ hires another agent, that is, if $\FFM_{\FA_a(t)}(t) \neq \emptyset$, the rejection variable is updated as $r_{a,\FA_a(t)}(t)=t$. This procedure continues until round $\TGS + 3n^2 - 1$, after which all agents set $\PH_a(\TGS + 3n^2)=1$ and transition back to the committing stage. By~\cref{obs:drr1}, each updating phase ends with a matching $(a,\FA_a(\TGS+3n^2))_{\forall a\in\AG}$ that matches all agents; hence, during the subsequent committing phase no two agents apply to the same firm. 

\noindent\textbf{Interview Set During Committing.} \label{par:switch}
After an updating phase, agents enter a committing stage, acting independently under their current estimates until a coordination trigger is detected. During committing, at any time step $t \ge \TGS + 3n^2 + 1$, each agent sets
\(
    \FA_a(t) = \arg\max_{f \in \CSV_a(t)} \hat{u}_{a,f}(t),
\)
where the candidate set $\CSV_a(t)$ is fixed throughout the committing stage. As long as the maximizer does not change, agent $a$ keeps applying to the same firm, even as estimates continue to update.

When a switching condition is detected at some $t \in \mathcal{C}$, agent $a$ sets $\FA_a(t)=\emptyset$; all agents then synchronously switch to an updating stage in the next round. For each agent $a$, define the committing rounds that trigger this transition as follows (the first two are detected locally, while the third is a global signal common to all agents):

(i) \textbf{Preference-list inconsistency.}
Learning may change agent $a$’s estimated top choice relative to the firm selected at the end of the most recent updating phase, suggesting that the last distributed $\FGS$ run used an invalid list. Formally,
\begin{align*}
\mathcal{S}^{\mathrm{inc}}_a
\doteq
\bigl\{
t \in \mathcal{C} :
\FA_a(t) \neq \FA_a(\TGS+3n^2)
\bigr\}.
\end{align*}
(ii) \textbf{Strategic rejection.}
Agent $a$ may be rejected due to a firm’s strategic abstention after the beginning of the last updating phase; this is detected when the relevant firm appears in the vacancy feedback, invalidating the matching obtained at the end of the updating phase.

{\small
\setlength{\abovedisplayskip}{-13pt}
\begin{align*}
\mathcal{S}^{\mathrm{rej}}_a 
\doteq 
\Bigl\{\, t \in \mathcal{C} \;:\;
\exists\, t' \in [\TGS,t)
\ \text{s.t.}\ 
\FA_a(t') \in \VF(t')
\Bigr\}.
\end{align*}
}
(iii) \textbf{New vacancy signal.}
A matched firm becomes newly vacant, so vacancies exceed $m-n$. Formally,
{\small
\[
\mathcal{S}^{\mathrm{vac}}_a 
\doteq 
\{\, 
t \in \mathcal{C} : |\VF(t)| > m-n\}.\]
}
If $t \in \mathcal{S}^{\mathrm{inc}}_a \cup \mathcal{S}^{\mathrm{rej}}_a \cup \mathcal{S}^{\mathrm{vac}}_a$, agent $a$ sets $\PH_a(t+1)=0$ and updates $\TGS = t+1$. All agents then synchronously enter a new updating stage in the following round.

\subsection{Committing to a Perfect Matching~\cref{obs:drr1}}\label{sec:obs-drr1}
This lemma shows that after each updating phase, the agents' applications induce a perfect matching, which serves as the baseline structure for the subsequent regret analysis.
\begin{restatable}{lemma}{obsdrrone}\label{obs:drr1}
In a matching market $\MM$ where firms follow~\cref{alg:fdrr} and agents follow~\cref{alg:drr}, at the end of any updating phase that starts at time $\TGS$, the applications $(a,\FA_a(\TGS+3n^2))_{a\in\AG}$ form a matching that matches all agents.
\end{restatable}
\begin{proof}
Since the beginning of the updating phase at time $\TGS$, each agent $a$ selects its applied firm according to
\[
\FA_a(t) = \argmax_{f \in \CSV_a(t)} \hat{u}_{a,f}(\TGS),
\]
that is, with respect to its estimated preference list $\EPL_a(\TGS)$. Consequently, during the updating phase, each agent effectively traverses $\EPL_a(\TGS)$ in descending order, applying to increasingly less preferred firms until reaching one that does not reject it in favor of another agent. Once an agent $a$ becomes matched to a firm $f=\FA_a(t)$, it continues to apply to $f$ in subsequent rounds until either the updating phase ends or $f$ rejects $a$ by hiring another agent. Therefore, whenever the applications $(a,\FA_a(t))_{a\in\AG}$ do not form a perfect matching that matches all agents, there exists at least one agent who has not been hired in that round.

In addition to rejections caused by competing applicants, agents may experience strategic rejections. However, by construction, the rejection variable $r_{a,f}$ is updated only when firm $f$ hires another agent, i.e., when $\NH_f(t)=1$, and not when $f$ strategically abstains from hiring. As a result, following a strategic rejection, agents reapply to the same firm in the subsequent round. Moreover, under the rejection policy~\cref{alg:fdrr}, no firm performs strategic rejections in two consecutive rounds during an updating phase. Further, by the definition of the candidate set~\eqref{def:candset2} under~\cref{alg:drr}, an agent does not apply to a firm after being rejected in favor of another hire. Hence, each firm can reject a given agent at most twice during the updating phase starting from $\TGS$: at most once strategically, and at most once in favor of hiring another agent.

It follows that as long as the applications do not form a perfect matching, within at most three rounds there exists at least one agent who moves to a less preferred firm. Indeed, in the worst case an agent is first strategically rejected at time $t$ (and thus re-applies to the same firm at time $t+1$), and then rejected at time $t+1$ in favor of another hire (which updates $r_{a,f}$), in which case it can consider a new firm only at time $t+2$.

Since, by~\cref{obs:topnescorrect}, each agent’s match during the updating phase is restricted to the top $n$ firms in $\EPL_a(\TGS)$, the applications form a perfect matching within at most $3n^2$ rounds. In particular, after at most $3n^2$ rounds from the start of an updating phase, the applications $(a,\FA_a(t))_{a\in\AG}$ form a perfect matching. The $O(n^2)$ bound follows since each agent applies to at most the top $n$ firms (with respect to $\EPL_a(\TGS)$) during a single updating phase.
\end{proof}

\subsection{Counting Updating Stages via Triggered Events}\label{sec:obs-drr2}
This observation expresses the number of updating stages in terms of the total number of triggered inconsistency/rejection events across all agents, which is later used to bound the overall regret.
\begin{restatable}{observation}{obsdrrtwo}\label{obs:drr2}
The total number of updating stages is given by 
\(
    \mathbb{E}\!\left[\left|\bigcup_{a \in \AG} \left(\mathcal{S}^{\mathrm{inc}}_{a} \cup \mathcal{S}^{\mathrm{rej}}_{a}\right)\right|\right].
\)
\end{restatable}
\begin{proof}
This holds because whenever an agent $a \in \AG$ detects an event in $\mathcal{S}^{\mathrm{inc}}_a$, it abstains from applying in that round, which is subsequently revealed to all other agents through their detection of events in $\mathcal{S}^{\mathrm{vac}}_{a'}$. Similarly, if $a$ detects an event in $\mathcal{S}^{\mathrm{rej}}_a$, it indicates that at that time step the agent was strategically not hired, an event that is also reflected to others via $\mathcal{S}^{\mathrm{vac}}_{a'}$ in the firm-side feedback. Therefore, every initialization of a new distributed $\FGS$ run is ultimately triggered by at least one agent $a$ detecting the current time step $t$ as belonging to either $\mathcal{S}^{\mathrm{inc}}_a$ or $\mathcal{S}^{\mathrm{rej}}_a$.
\end{proof}

\subsection{Bounding the Number of Updating Stages in $\alpha$-Reducible Markets}\label{sec:proof-alphadec}
\begin{restatable}{lemma}{lemmalphadec}\label{lemma:alphadec}
In an $\alpha$-reducible market, with strategic firms following~\cref{alg:fdrr} and agents following~\cref{alg:drr}, the total number of updating stages is $O\!\left(nm^2\min\!\left\{\overline{\Delta}_{\AG},\underline{\Delta}_{\FI}\right\}^{-2}\right)$, where $\overline{\Delta}_{\AG} \doteq \min_{a\in\AG}\overline{\Delta}_{a}$ and $\underline{\Delta}_{\FI} \doteq \min_{f\in\FI}\underline{\Delta}_{f}$.
\end{restatable}
\begin{proof}
Consider the unique stable matching \((a_i, f_i)_{i \in [n]}\) in an $\alpha$-reducible market.
For each \(i \in [n]\), define the set of time steps
\[
\mathcal{G}_i \doteq \{t \in \mathcal{C} : \forall j \in [i],\ \FA_{a_j}(t) = f_j\},
\]
namely, the rounds in which agents \(a_1, \ldots, a_i\) all apply to their respective stable partners.
Its complement \(\Bar{\mathcal{G}}_i\) consists of rounds in which at least one of these agents is not matched to its stable partner.

We now bound the expected number of switching events across all agents. In particular, with respect to~\cref{obs:drr2}, we bound
\(
\mathbb{E}\!\left[\left|\bigcup_{a_i\in \AG}\mathcal{S}^{\mathrm{inc}}_{a_i}\cup \mathcal{S}^{\mathrm{rej}}_{a_i} \right|\right]
\)
as follows:
{\small
\begin{align}
\mathbb{E}\!\left[\!\left|\bigcup_{a_i\in \AG}\mathcal{S}^{\mathrm{inc}}_{a_i}\cup \mathcal{S}^{\mathrm{rej}}_{a_i} \right|\!\right]
&=\mathbb{E}\!\left[\!\left|\bigcup_{a_i\in \AG}
\!\!\left(
\!\{t\!\in\!\mathcal{S}^{\mathrm{inc}}_{a_i}\!\cup\!\mathcal{S}^{\mathrm{rej}}_{a_i}:t\!\in\!\Bar{\mathcal{G}}_i\}
\cup
\{t\!\in\!\mathcal{S}^{\mathrm{inc}}_{a_i}\!\cup\!\mathcal{S}^{\mathrm{rej}}_{a_i}:t\!\in\!\mathcal{G}_i\}
\!\right)\!\right|\!\right] \nonumber\\
&\leq
\mathbb{E}\!\left[\!\left|\bigcup_{a_i\in \AG}
\!\!\left(
\Bar{\mathcal{G}}_i
\cup
\{t\!\in\!\mathcal{S}^{\mathrm{inc}}_{a_i}\!\cup\!\mathcal{S}^{\mathrm{rej}}_{a_i}:t\!\in\!\mathcal{G}_i\}
\right)\!\right|\!\right] \nonumber\\
&\overset{(a)}{\leq}
\mathbb{E}\!\left[\!\left|
\bigcup_{a_i\in \AG}\!\left(\Bar{\mathcal{E}}_{a_i,f_i}\!\cup\!\Bar{\mathcal{E}}_{f_i,a_i}\right)
\cup
\bigcup_{a_i\in \AG}\!\left(
\{t\!\in\!\mathcal{S}^{\mathrm{inc}}_{a_i}\!\cup\!\mathcal{S}^{\mathrm{rej}}_{a_i}:t\!\in\!\mathcal{G}_i\}
\right)
\!\right|\!\right] \nonumber\\
&\overset{(b)}{\leq}
\mathbb{E}\!\left[\!\left|
\bigcup_{a_i\in \AG}\!\left(\Bar{\mathcal{E}}_{a_i,f_i}\!\cup\!\Bar{\mathcal{E}}_{f_i,a_i}\right)
\cup
\bigcup_{a_i\in \AG}\!\left(
\!\{t\!\in\!\mathcal{S}^{\mathrm{inc}}_{a_i}\!\cup\!\mathcal{S}^{\mathrm{rej}}_{a_i}:
(t\!\in\!\mathcal{G}_i)
\land
(t,t\!-\!1\!\in\!\mathcal{E}_{a_i,f_i}\!\cap\!\mathcal{E}_{f_i,a_i})\!\}
\right)
\!\right|\!\right] \nonumber\\
&\overset{(c)}{\leq}
\sum_{i\in[n]}\Bigg(
\mathbb{E}\!\left[\!\left|\Bar{\mathcal{E}}_{a_i,f_i}\right|\!\right]
+\mathbb{E}\!\left[\!\left|\Bar{\mathcal{E}}_{f_i,a_i}\right|\!\right]
+\!\!\sum_{f\in\mathcal{H}_{a_i,f_i}}\!\!
\mathbb{E}\!\left[\!\left|\!\{t:
f\!\in\!\hat{\mathcal{L}}_{a_i,f_i}(t\!-\!1)
\!\cap\!\hat{\mathcal{H}}_{a_i,f_i}(t)\}\!\right|\!\right]\Bigg) \nonumber\\
&\overset{(d)}{\leq}
\sum_{a_i\in\AG}
\frac{4m \exp\!\left(-\tfrac{\overline{\Delta}_{a_i}^2}{2m}\right)}{1 - \exp\!\left(-\tfrac{\overline{\Delta}_{a_i}^2}{2m}\right)}
+
\sum_{f_i\in\FI}
\frac{4n \exp\!\left(-\tfrac{\underline{\Delta}_{f_i}^2}{2m}\right)}{1 - \exp\!\left(-\tfrac{\underline{\Delta}_{f_i}^2}{2m}\right)} \nonumber\\
&\overset{(e)}{\leq}
\sum_{a_i\in\AG} O\!\left(m^2\overline{\Delta}_{a_i}^{-2}\right)
+
\sum_{f_i\in\FI} O\!\left(nm\underline{\Delta}_{f_i}^{-2}\right)
\in O\!\left(nm^2\min\!\left\{\overline{\Delta}_{\AG},\underline{\Delta}_{\FI}\right\}^{-2}\right),
\label{eq:alphadGS}
\end{align}
}

Here, inequality~(a) follows from the fact that if \(t \in \Bar{\mathcal{G}}_i\), then either some agent \(a_j\) has an invalid estimated preference list at time \(\TGS\), or there exists a time step \(t' \in [\TGS, \TGS + 3n^2]\) at which the firm \(f_j\), for some \(j \in [i]\), has an invalid estimated preference list. Otherwise, if the estimated preference lists of agents \(a_1,\ldots,a_i\) are valid at time \(\TGS\) and the estimated preference lists of firms \(f_1,\ldots,f_i\) remain valid throughout \([\TGS,\TGS+3n^2]\), then by the definition of $\alpha$-reducibility (and its inductive structure), these agents must be matched to their unique stable partners, which implies \(t \in \mathcal{G}_i\).

The justification of inequality~(b) follows by considering the possible cases when
$t-1 \in \mathcal{G}_i$.

\begin{itemize}
\item[(i)]
Either $\EPL_{a_i}(t-1)$ or $\EPL_{f_i}(t-1)$ is invalid.
In this case,
$t-1 \in \Bar{\mathcal{E}}_{a_i,f_i} \cup \Bar{\mathcal{E}}_{f_i,a_i}$,
and the subsequent execution of $\FGS$ can be charged to round $t-1$ as an instance of an invalid preference list.

\item[(ii)]
Both $\EPL_{a_i}(t-1)$ and $\EPL_{f_i}(t-1)$ are valid, but at least one becomes invalid at time $t$.
If $\EPL_{a_i}(t)$ becomes invalid, then $t \in \mathcal{S}^{\mathrm{inc}}_{a_i}$, corresponding to a switching event by the agent.
If instead $\EPL_{f_i}(t)$ becomes invalid, firm $f_i$ may perform a strategic rejection, yielding
$t \in \mathcal{S}^{\mathrm{rej}}_{a_i}$.
In either case, the ensuing updating phase, and hence the next execution of $\FGS$, is associated with
$t \in \Bar{\mathcal{E}}_{a_i,f_i} \cup \Bar{\mathcal{E}}_{f_i,a_i}$.

\item[(iii)]
Both $\EPL_{a_i}(t-1)$ and $\EPL_{a_i}(t)$, as well as $\EPL_{f_i}(t-1)$ and $\EPL_{f_i}(t)$, are valid.
Suppose first that firm $f_i$ performs a strategic rejection at time $t$, i.e., $t\in \mathcal{S}_{f_i}$.
This action is observed by agent $a_i$ through $\VFP(t)$, and therefore implies $t\in \mathcal{S}^{\mathrm{rej}}_{a_i}$.

Since all agents $a_j$ for $j\in[i-1]$ are matched to their stable partners $f_j$, and since $t\in \mathcal{E}_{f_i,a_i}$, we claim that there must exist some agent $a_j$ whose estimated preference list was invalid at some time during the previous updating phase $[\TGS,\TGS+3n^2]$.
Indeed, the strategic rejection by $f_i$ can only be triggered if there exists an agent $a_j$ that $f_i$ truly prefers (i.e., $a_j \in \mathcal{H}_{f_i,a_i}$) but that was previously rejected by $f_i$ during the updating phase.
Concretely, during the updating phase, agent $a_j$ applied to $f_i$ and was rejected at some time $t'$, which means that at the time of rejection $a_j \in \hat{\mathcal{L}}_{f_i,a_i}$.
After this rejection, $a_j$ subsequently applied to (and matched with) some firm $f_j$ that it estimated to be less preferred.
Given that $f_i$ later performs a strategic rejection, this sequence of events implies that $a_j$'s estimated preference list must have been invalid at time $t'$.

Otherwise, if no strategic rejection occurs at time $t$ and $t \in \mathcal{S}^{\mathrm{inc}}_{a_i}$, then since
$t-1,t \in \mathcal{E}_{a_i,f_i} \cap \mathcal{E}_{f_i,a_i}$,
a switching event of type $\mathcal{S}^{\mathrm{inc}}_{a_i}$ implies that there exists a firm
$f \in \mathcal{H}_{a_i,f_i}$ such that
$f \in \hat{\mathcal{L}}_{a_i,f_i}(\TGS)$ but
$f \in \hat{\mathcal{H}}_{a_i,f_i}(t)$.
\end{itemize}

These cases together establish inequality~(c). Inequality~(d) then follows with the same approach as used in the proof of Lemma~\ref{lemma:valid-rr}.
\end{proof}

\subsection{Proof of Theorem~\ref{thm:decregmain} for $\alpha$-reducible Markets}\label{proof:thmalphareg}
\thmdecregmain*
\begin{proof}
We decompose the regret of agent $a_i$ into two components:
(i) regret incurred during updating phases, denoted by $\mathcal{U}$, and
(ii) regret incurred during committing phases, denoted by $\mathcal{C}$.
Formally, the expected regret satisfies
\begin{align}
\mathbb{E}\!\left[R_{a_i}(T)\right]
&=
\mathbb{E}\!\left[R_{a_i}(\mathcal{U})\right]
+
\mathbb{E}\!\left[R_{a_i}(\mathcal{C})\right] \nonumber\\
&=
\underbrace{
\sum_{t \in \mathcal{U}}
\left(
u_{a_i,f_i}
-
\mathbb{E}[X_{a_i,\FM_{a_i}(t)}]
\right)
}_{\text{updating regret}}
+
\underbrace{
\sum_{t \in \mathcal{C}}
\left(
u_{a_i,f_i}
-
\mathbb{E}[X_{a_i,\FM_{a_i}(t)}]
\right)
}_{\text{committing regret}}.
\label{eq:decregalpha}
\end{align}
We bound these two terms separately.

We begin with the cumulative regret incurred during the updating phases. By construction, each updating phase lasts at most $3n^2$ rounds, and in each such round an agent incurs regret at most one. Moreover, by~\cref{obs:drr2} and~\cref{lemma:alphadec}, the expected number of updating phases is bounded by the number of switching events across all agents. Consequently,
\begin{align}
\sum_{t \in \mathcal{U}}
\left(
u_{a_i,f_i}
-
\mathbb{E}[X_{a_i,\FM_{a_i}(t)}]
\right)
&\le
3n^2 \cdot\,
\mathbb{E}\!\left[
\left|
\bigcup_{a \in \AG}
\left(
\mathcal{S}^{\mathrm{inc}}_a \cup \mathcal{S}^{\mathrm{rej}}_a
\right)
\right|
\right].
\label{ineq:update-regret}
\end{align}

We now focus on the second term, corresponding to regret incurred during committing phases.
Let
\[
\bigcup_{a \in \AG} \left(\mathcal{S}^{\mathrm{inc}}_a \cup \mathcal{S}^{\mathrm{rej}}_a\right)
= \{t_1,t_2,\ldots\}
\]
denote the set of switching times, ordered increasingly.
The $j$-th committing phase is defined as the interval
\[
[t'_j : t_{j+1}],
\qquad
\text{where } t'_j \doteq t_j + 3n^2 + 1,
\]
corresponding to the period between the $j$-th and $(j+1)$-th distributed executions of $\FGS$.
The collection of committing intervals is therefore
\[
\mathcal{K}
\doteq
\left\{
[t'_j,\, t_{j+1}] :
j \in \left[\left|\!\bigcup_{a \in \AG}(\mathcal{S}^{\mathrm{inc}}_a \cup \mathcal{S}^{\mathrm{rej}}_a)\right|\right]
\right\}.
\]

The expected regret incurred during committing phases can then be written as
{\small\begin{align}
\mathbb{E}\!\left[R_{a_i}(\mathcal{C})\right]
&=
\sum_{t \in \mathcal{C}}
\left(
u_{a_i,f_i}
-
\mathbb{E}[X_{a_i,\FM_{a_i}(t)}]
\right) \nonumber\\
&=
\mathbb{E}\!\left[
\sum_{j=1}^{|\mathcal{K}|}
\sum_{t \in [t'_j,\, t_{j+1}]}
\Big(
u_{a_i,f_i}
-
X_{a_i,\FM_{a_i}(t'_j)}(t)
\Big)
\right],\nonumber\\
&\overset{(a)}{\le} \mathbb{E}\!\left[
\left|
\bigcup_{a \in \AG}
\left(
\mathcal{S}^{\mathrm{inc}}_a \cup \mathcal{S}^{\mathrm{rej}}_a
\right)
\right|
\right] +
\mathbb{E}\!\left[
\sum_{j \in [|\mathcal{K}|]}
\bigl(t_{j+1}-1-t'_j\bigr)
\;\middle|\;
\FM_{a_i}(t'_j) \in \mathcal{L}_{a_i,f_i}
\right],  \nonumber\\
&\le
\mathbb{E}\!\left[
\left|
\bigcup_{a \in \AG}
\left(
\mathcal{S}^{\mathrm{inc}}_a \cup \mathcal{S}^{\mathrm{rej}}_a
\right)
\right|
\right]+\mathbb{E}\!\left[
\sum_{j \in [|\mathcal{K}|]}
\bigl(t_{j+1}-1-t'_j\bigr)
\;\middle|\;
\FM_{a_i}(t'_j) \in \mathcal{L}_{a_i,f_i} \land t'_j  \in \mathcal{G}_{i-1}
\right] \nonumber\\
&\quad+
\mathbb{E}\!\left[
\sum_{j \in [|\mathcal{K}|]}
\bigl(t_{j+1}-1-t'_j\bigr)
\;\middle|\;
t'_j \in \Bar{\mathcal{G}}_{i-1}
\right], \nonumber\\
&\overset{(b)}{\le}
\mathbb{E}\!\left[
\left|
\bigcup_{a \in \AG}
\left(
\mathcal{S}^{\mathrm{inc}}_a \cup \mathcal{S}^{\mathrm{rej}}_a
\right)
\right|
\right]+ \mathbb{E}\!\left[\left|\Bar{\mathcal{E}}_{a_i,f_i}\right|\right]
+
\mathbb{E}\!\left[\left|\Bar{\mathcal{E}}_{f_i,a_i}\right|\right]
+
\mathbb{E}\!\left[
\sum_{j \in [|\mathcal{K}|]}
\bigl(t_{j+1}-1-t'_j\bigr)
\;\middle|\;
t'_j \in \Bar{\mathcal{G}}_{i-1}
\right], \nonumber\\
&\overset{(c)}{\le}
\mathbb{E}\!\left[
\left|
\bigcup_{a \in \AG}
\left(
\mathcal{S}^{\mathrm{inc}}_a \cup \mathcal{S}^{\mathrm{rej}}_a
\right)
\right|
\right]+
\sum_{j=1}^{i}
\mathbb{E}\!\left[\left|\Bar{\mathcal{E}}_{a_j,f_j}\right|\right]
+
\mathbb{E}\!\left[\left|\Bar{\mathcal{E}}_{f_j,a_j}\right|\right],\label{ineq:finalalphadec}
\end{align}
}

Inequality~(a) follows by accounting for the single round at the end of each committing phase, $t_{j+1}$, in which an agent intentionally refrains from applying in order to signal coordination to other agents. Such a signaling round occurs at most once per committing phase, immediately before the subsequent updating phase. By~\cref{obs:drr2}, the expected number of these signaling rounds is bounded by
\(
\mathbb{E}\!\left[
\left|
\bigcup_{a \in \AG}
\left(
\mathcal{S}^{\mathrm{inc}}_a \cup \mathcal{S}^{\mathrm{rej}}_a
\right)
\right|
\right].
\)
This term also includes the cases when $a_i$ was strategically rejected during the last updating stage.

Inequality~(b) follows from the structure of $\alpha$-reducible markets. 
Condition on $t'_j \in \mathcal{G}_{i-1}$, so that the first $i-1$ agents are matched to their optimal firms, while agent $a_i$ is matched suboptimally to some firm in $\mathcal{L}_{a_i,f_i}$. 
By $\alpha$-reducibility, one of the following must hold.

(i) Either agent $a_i$ has applied to its optimal firm $f_i$ and was rejected during the last updating stage. In this case, firm $f_i$ must be matched to a suboptimal agent due to an invalid estimated preference list during the preceding updating phase. This firm-side invalidity can persist during the committing phase, but at most until $f_i$'s preference list becomes valid, at which point it strategically rejects its current match since $a_i$, which was previously rejected, now becomes preferred to its current applicant.

(ii) Alternatively, agent $a_i$ has not applied to $f_i$ while updating and is matched to a less preferred firm. This implies that $a_i$’s estimated preference list was already invalid at the start of the updating phase, and this agent-side invalidity may persist until it is corrected.

In both cases, the committing regret incurred while $t'_j \in \mathcal{G}_{i-1}$ can be uniquely charged to an invalidity event in either the agent’s or the firm’s estimated preference list. Consequently, the total committing regret of agent $a_i$ over all such disjoint intervals is bounded by
\(
\mathbb{E}\!\left[\left|\Bar{\mathcal{E}}_{a_i,f_i}\right|\right]
+
\mathbb{E}\!\left[\left|\Bar{\mathcal{E}}_{f_i,a_i}\right|\right].
\)

Moreover, to establish inequality~(c), we bound
\begin{align*}
\mathbb{E}\!\Bigg[
\sum_{j \in [|\mathcal{K}|]}
\bigl(t_{j+1}-1-t'_j\bigr)
\;\Bigg|\;
t'_j \in \Bar{\mathcal{G}}_{i-1}
\Bigg]
&\le
\mathbb{E}\!\Bigg[
\sum_{j \in [|\mathcal{K}|]}
\bigl(t_{j+1}-1-t'_j\bigr)
\;\Bigg|\;
\FM_{a_{i-1}}(t'_j) \neq f_{i-1}
\;\land\;
t'_j \in \mathcal{G}_{i-2}
\Bigg] \\
&\quad\quad\qquad+
\mathbb{E}\!\Bigg[
\sum_{j \in [|\mathcal{K}|]}
\bigl(t_{j+1}-1-t'_j\bigr)
\;\Bigg|\;
t'_j \in \Bar{\mathcal{G}}_{i-2}
\Bigg] \\
&\le
\mathbb{E}\!\left[\left|\Bar{\mathcal{E}}_{a_{i-1},f_{i-1}}\right|\right]
+
\mathbb{E}\!\left[\left|\Bar{\mathcal{E}}_{f_{i-1},a_{i-1}}\right|\right] +
\mathbb{E}\!\Bigg[
\sum_{j \in [|\mathcal{K}|]}
\bigl(t_{j+1}-1-t'_j\bigr)
\;\Bigg|\;
t'_j \in \Bar{\mathcal{G}}_{i-2}
\Bigg] \\
&\vdots\\
&\le\sum_{j=1}^{i-1}
\mathbb{E}\!\left[\left|\Bar{\mathcal{E}}_{a_j,f_j}\right|\right]
+
\mathbb{E}\!\left[\left|\Bar{\mathcal{E}}_{f_j,a_j}\right|\right],
\end{align*}
where each step follows with the same rationale used to obtain inequality~(b).

Finally, applying Lemma~\ref{lemma:valid-rr} to~\eqref{ineq:finalalphadec} and using~\eqref{ineq:update-regret} and~\cref{lemma:alphadec}, the total regret is bounded as
{\small
\begin{align*}
    \mathbb{E}\!\left[R_{a_i}(T)\right]&\leq \left(3n^2+1\right)\cdot  O\left(n^3m^2\min\!\left\{\overline{\Delta}_{\AG},\underline{\Delta}_{\FI}\right\}^{-2}\right) +
               \left( \sum_{j=1}^i
\frac{4\left|\mathcal{L}_{a_j,f_j}\right| \exp\!\left(-\tfrac{\overline{\Delta}_{a_j}^2}{2m}\right)}{1 - \exp\!\left(-\tfrac{\overline{\Delta}_{a_j}^2}{2m}\right)}
+
\frac{4\left|\mathcal{L}_{f_j,a_j}\right| \exp\!\left(-\tfrac{\underline{\Delta}_{f_j}^2}{2m}\right)}{1 - \exp\!\left(-\tfrac{\underline{\Delta}_{f_j}^2}{2m}\right)}\right)\\
&\in  O\left(n^3m^2\min\!\left\{\overline{\Delta}_{\AG},\underline{\Delta}_{\FI}\right\}^{-2} + m\cdot \left(\sum_{j=1}^i \left|\mathcal{L}_{a_j,f_j}\right|\overline{\Delta}_{a_j}^{-2}
+
\left|\mathcal{L}_{f_j,a_j}\right|\underline{\Delta}_{f_j}^{-2}\right)\right)
\end{align*}
        }
\end{proof}

\subsection{Extension of Algorithm~\ref{alg:drr}'s Analysis to General Markets}\label{sec:deccoordgen}
We extend the results established for $\alpha$-reducible markets in the main body to general, unstructured matching markets, where no structural assumptions are imposed and multiple stable matchings may exist.

\subsubsection{An Upper Bound on the Expected Number of Updating Stages}\label{sec:gen-upper-updates}
We first extend \cref{lemma:alphadec} to general matching markets in order to bound the expected number of updating stages. In the absence of structural assumptions, this bound becomes market-independent and necessarily larger than in the $\alpha$-reducible case.

\begin{restatable}{lemma}{lemmadec}\label{lemma:dec}
In a matching market $\MM$, with strategic firms following~\cref{alg:fdrr} and agents following~\cref{alg:drr}, the total number of distributed runs of $\FGS$ (i.e., the number of updating phases) is time-independent
and of $O(n^2m^2)$.
\end{restatable}

\begin{proof}
In a general matching market $\MM$, as introduced earlier, we denote the agent-optimal and agent-pessimal stable matches of agent $a \in \AG$ by $\overline{f^*_a}$ and $\underline{f^*_a}$, respectively. By the lattice structure of stable matchings, the agent-optimal stable matching coincides with the firm-pessimal one: in particular, if $\overline{f^*_a}=f$, then $\underline{a^*_f}=a$, and if $\underline{f^*_a}=f$, then $\overline{a^*_f}=a$.

Before proceeding, we introduce the event
\(
\psi_{\mathcal{F}} \doteq \psi^{(1)}_{\mathcal{F}} \cap \psi^{(2)}_{\mathcal{F}}.
\)
We first define
\begin{align*}
\psi^{(1)}_{\FI}
&\doteq
\left\{
t \in [T] :
\forall f \in \FI,\;
t \in \bigcap_{a \in \mathcal{L}_{f,\overline{a^*_f}}}
\mathcal{E}_{f,a}
\right\},\\[4pt]
\psi^{(2)}_{\FI}
&\doteq
\left\{
t \in [T] :
\forall f \in \mathcal{F},\;
\Bigl(
t \in \mathcal{E}_{f,\overline{a^*_f}}
\ \land\
\forall a \in \mathcal{H}_{f,\overline{a^*_f}},\;
a \in \hat{\mathcal{H}}_{f,\overline{a^*_f}}(t)
\Bigr)
\right\}.
\end{align*}

Accordingly, $\psi_{\FI}$ is the set of time steps at which, for every firm \( f \in \FI \), the estimated preference list of \( f \) correctly orders all agents that are strictly less preferred than its agent-optimal match \( \overline{a^*_f} \), and every agent strictly preferred to \( \overline{a^*_f} \) is ranked above it (possibly without being ordered among themselves).

We also define the event $\zeta_{\AG}$ as
\begin{align*}
\zeta_{\AG}
\doteq
\left\{
t \in [T] :
\forall a \in \AG,\;
t \in
\bigcap_{f \in \{\underline{f^*_a}\} \cup \mathcal{H}_{a,\underline{f^*_a}}}
\mathcal{E}_{a,f}
\right\}.
\end{align*}
That is, $\zeta_{\AG}$ is the set of time steps at which, for every agent $a$, the estimated preference list correctly orders its agent-pessimal stable match $\underline{f^*_a}$ and all firms strictly preferred to it, according to the ground-truth preferences.

We also define $\TGS(t)$ as the initial time step of the most recent updating stage by time $t$.

We now follow the proof idea in~\cref{obs:drr2} that the expected number of updating phases is bounded by
\(
\mathbb{E}\!\left[\left|\bigcup_{a_i\in \AG}\mathcal{S}^{\mathrm{inc}}_{a_i}\cup \mathcal{S}^{\mathrm{rej}}_{a_i} \right|\right].
\)
Thus we write
\begin{align*}
\mathbb{E}\!\left[\left|\bigcup_{a_i\in \AG}\mathcal{S}^{\mathrm{inc}}_{a_i}\cup \mathcal{S}^{\mathrm{rej}}_{a_i} \right|\right]
&=
\mathbb{E}\!\left[\sum_{t=1}^T\mathbbm{1}\left\{t \in \bigcup_{a_i\in \AG}\mathcal{S}^{\mathrm{inc}}_{a_i}\cup \mathcal{S}^{\mathrm{rej}}_{a_i} \right\}\right]\\
&\leq
\mathbb{E}\!\left[\sum_{t=1}^T\mathbbm{1}\left\{t \in \bigcup_{a_i\in \AG}\mathcal{S}^{\mathrm{inc}}_{a_i}\cup \mathcal{S}^{\mathrm{rej}}_{a_i} \land \exists t' \in[\TGS(t),t-1], t' \in \Bar{\psi}_\FI \cup \Bar{\zeta}_{\AG} \right\}\right] \\
&\quad+
\mathbb{E}\!\left[\sum_{t=1}^T\mathbbm{1}\left\{t \in \bigcup_{a_i\in \AG}\mathcal{S}^{\mathrm{inc}}_{a_i}\cup \mathcal{S}^{\mathrm{rej}}_{a_i} \land \forall t' \in[\TGS(t),t-1], t' \in \psi_\FI \cap \zeta_{\AG} \right\}\right]\\
&\overset{(a)}\leq \mathbb{E}\!\left[\sum_{t=1}^T\mathbbm{1}\left\{ \exists t' \in[\TGS(t),t-1], t' \in \Bar{\psi}_\FI \cup \Bar{\zeta}_{\AG} \right\}\right] \\
&\qquad\quad+
\mathbb{E}\!\left[\sum_{t=1}^T\mathbbm{1}\left\{t \in \bigcup_{a_i\in \AG}\mathcal{S}^{\mathrm{inc}}_{a_i}\cup \mathcal{S}^{\mathrm{rej}}_{a_i} \land t \in \Bar{\psi}_\FI \cup \Bar{\zeta}_{\AG} \right\}\right]\\
&\overset{(b)}\leq\mathbb{E}\!\left[\left|\Bar{\psi}_\FI \cup \Bar{\zeta}_{\AG}\right|\right].
\end{align*}
Inequality~(a) follows by a case analysis.

(i) If there exists a time step $t' \in [\TGS(t),t-1]$ such that
$t' \in \Bar{\psi}_\FI \cup \Bar{\zeta}_{\AG}$, then the updating phase initiated at $\TGS(t)$ is charged to this violation.

(ii) Otherwise, for all $t' \in [\TGS(t), t-1]$ we have
$t' \in \psi_\FI \cap \zeta_{\AG}$.
In this case, throughout the updating phase starting at $\TGS(t)$, all firms have estimated preference lists that are valid with respect to their agent-optimal matches $\overline{a^*_f}$ in the sense of $\psi_{\FI}$, and all agents have estimated preference lists that are valid with respect to their agent-pessimal matches $\underline{f^*_a}$ in the sense of $\zeta_{\AG}$.
Consequently, the matching $(a,\underline{f^*_a})_{\forall a \in \AG}$ remains stable with respect to the estimated preferences at all time steps in $[\TGS(t), t-1]$, and no switching condition can be triggered.

Moreover, since the distributed $\FGS$ run by the agents returns the agent-optimal (and firm-pessimal) stable matching under the current estimated preferences, it follows that for every agent $a \in \AG$, $\FA_a(t)$ is ranked no worse than $\underline{f^*_a}$ under $\EPL_a(t)$, and for every firm $f \in \FI$, the agent matched to $f$ is ranked no better than $\overline{a^*_f}$ under $\EPL_f(t)$.

Therefore, any switching event at time $t$ must correspond to a violation of either $\zeta_{\AG}$ or $\psi_\FI$:
\begin{itemize}
    \item If $t \in \mathcal{S}^{\mathrm{inc}}_a$ for some agent $a$, then $a$ detects a preference invalidity with respect to some
    $f \in \mathcal{H}_{a,\underline{f^*_a}}$, implying that
    $t \in \Bar{\zeta}_{\AG}$.
    \item If $t \in \mathcal{S}^{\mathrm{rej}}_a$ for some agent $a$, then the firm
    $f=\FA_a(t)$ performs a strategic rejection because its estimated
    preference list becomes invalid with respect to some
    $a' \in \mathcal{L}_{f,\overline{a^*_f}}$, implying that
    $t \in \Bar{\psi}_{\FI}$.
\end{itemize}

Inequality~(b) then follows since each updating phase can be uniquely associated with a time
step in $\Bar{\psi}_\FI \cup \Bar{\zeta}_{\AG}$, and such violations are counted at most once.

We complete the proof by bounding the two terms in
\(
\mathbb{E}\!\left[\left|\Bar{\psi}_\FI \cup \Bar{\zeta}_{\AG}\right|\right]
\;\le\;
\mathbb{E}\!\left[\left|\Bar{\zeta}_{\AG}\right|\right]
+
\mathbb{E}\!\left[\left|\Bar{\psi}_\FI\right|\right]
\)
separately.

To bound $\mathbb{E}\!\left[\left|\Bar{\zeta}_{\AG}\right|\right]$, we write
\begin{align}
 \mathbb{E}\!\left[\left| \Bar{\zeta}_{\AG}\right|\right]
 &= \mathbb{E}\!\left[\left|
 \left\{t \in [T]:
 \exists a \in \AG,\;
 \exists f \in \{\underline{f^*_a}\} \cup \mathcal{H}_{a,\underline{f^*_a}},
 \ t \in \Bar{\mathcal{E}}_{a,f}
 \right\}
 \right|\right],\nonumber\\
 &\leq
 \sum_{a\in \AG}
 \sum_{f \in \{\underline{f^*_a}\} \cup \mathcal{H}_{a,\underline{f^*_a}}}
 \mathbb{E}\!\left[\left| \Bar{\mathcal{E}}_{a,f}\right|\right],\nonumber\\
 &\leq
 \sum_{a\in \AG}
 \sum_{f \in \{\underline{f^*_a}\} \cup \mathcal{H}_{a,\underline{f^*_a}}}
 \frac{4|\mathcal{L}_{a,f}|
 \exp\!\left(-\tfrac{(\Delta^{\texttt{min}}_{a,f})^2}{2m}\right)}
 {1 - \exp\!\left(-\tfrac{(\Delta^{\texttt{min}}_{a,f})^2}{2m}\right)},
 \label{ineq:zeta}\\
 &\in
 \sum_{a\in \AG}
 \sum_{f \in \{\underline{f^*_a}\} \cup \mathcal{H}_{a,\underline{f^*_a}}}
 O\!\left(m\cdot|\mathcal{L}_{a,f}|\right)
 \;\in\; O(n^2m^2),\nonumber
\end{align}
where the final bound follows from~\cref{lemma:valid-rr}.

For $\mathbb{E}\!\left[\left|\Bar{\psi}_\FI\right|\right]$, we then write
\begin{align}
   \mathbb{E}\!\left[\left|\Bar{\psi}_\FI\right|\right]
   &= \mathbb{E}\!\left[\left|\Bar{\psi}^{(1)}_\FI \cup \Bar{\psi}^{(2)}_\FI\right|\right],\nonumber\\
   &\leq
   \mathbb{E}\!\left[\left|\Bar{\psi}^{(1)}_\FI\right|\right]
   +
   \mathbb{E}\!\left[\left|\Bar{\psi}^{(2)}_\FI\right|\right],\nonumber\\
   &\leq
   \mathbb{E}\!\left[\left|\Bar{\psi}^{(1)}_\FI\right|\right]
   +
   \mathbb{E}\!\left[
   \left|
   \left\{t \in [T]:
   \exists f \in \FI,\;
   t \in \Bar{\mathcal{E}}_{f,\overline{a^*_f}}
   \ \lor\
   \left(\exists a \in \mathcal{H}_{f,\overline{a^*_f}},
   \ a \in \hat{\mathcal{L}}_{f,\overline{a^*_f}}(t)\right)
   \right\}
   \right|
   \right],\nonumber\\
&\leq
\mathbb{E}\!\left[\left|\Bar{\psi}^{(1)}_\FI\right|\right]
+
\sum_{f \in \FI}\mathbb{E}\!\left[\left|\Bar{\mathcal{E}}_{f,\overline{a^*_f}}\right|\right]
+
\mathbb{E}\!\left[
\left|
\left\{t \in [T]:
\exists a \in \mathcal{H}_{f,\overline{a^*_f}},
\ a \in \hat{\mathcal{L}}_{f,\overline{a^*_f}}(t)
\right\}
\right|
\right],\nonumber\\
   &\leq
   \sum_{f\in \FI}
   \sum_{a \in \mathcal{L}_{f,\overline{a^*_f}}}
   \frac{4|\mathcal{L}_{f,a}|
   \exp\!\left(-\tfrac{(\Delta^{\texttt{min}}_{f,a})^2}{2m}\right)}
   {1 - \exp\!\left(-\tfrac{(\Delta^{\texttt{min}}_{f,a})^2}{2m}\right)}
   +
   \frac{4n
   \exp\!\left(-\tfrac{(\overline{\Delta}^{\texttt{min}}_{f})^2}{2m}\right)}
   {1 - \exp\!\left(-\tfrac{(\overline{\Delta}^{\texttt{min}}_{f})^2}{2m}\right)},
   \label{ineq:psi}\\
   &\in
   \sum_{f\in \FI}
   \sum_{a \in \mathcal{L}_{f,\overline{a^*_f}}}
   O\!\left(m\cdot|\mathcal{L}_{f,a}| + nm\right)
   \;\in\; O(n^2m^2),\nonumber
\end{align}
where the final bound follows by applying the same argument used to bound
$\mathbb{E}\!\left[\left|\Bar{\zeta}_{\AG}\right|\right]$ to
$\mathbb{E}\!\left[\left|\Bar{\psi}^{(1)}_\FI\right|\right]$,
and then invoking~\cref{lemma:valid-rr} for the remaining term.
Finally, we got \hbcomment{fix the overflow if possible.}
{\footnotesize
\begin{align*}
    \mathbb{E}\!\left[\left|\bigcup_{a_i\in \AG}\mathcal{S}^{\mathrm{inc}}_{a_i}\cup \mathcal{S}^{\mathrm{rej}}_{a_i} \right|\right] &\leq \sum_{a\in \AG}
 \sum_{f \in \{\underline{f^*_a}\} \cup \mathcal{H}_{a,\underline{f^*_a}}}
 \frac{4|\mathcal{L}_{a,f}|
 \exp\!\left(-\tfrac{(\Delta^{\texttt{min}}_{a,f})^2}{2m}\right)}
 {1 - \exp\!\left(-\tfrac{(\Delta^{\texttt{min}}_{a,f})^2}{2m}\right)} \\ & \qquad\qquad+ \sum_{f\in \FI}
   \sum_{a \in \mathcal{L}_{f,\overline{a^*_f}}}
   \frac{4|\mathcal{L}_{f,a}|
   \exp\!\left(-\tfrac{(\Delta^{\texttt{min}}_{f,a})^2}{2m}\right)}
   {1 - \exp\!\left(-\tfrac{(\Delta^{\texttt{min}}_{f,a})^2}{2m}\right)}
   +
   \frac{4n
   \exp\!\left(-\tfrac{(\overline{\Delta}^{\texttt{min}}_{f})^2}{2m}\right)}
   {1 - \exp\!\left(-\tfrac{(\overline{\Delta}^{\texttt{min}}_{f})^2}{2m}\right)},
\end{align*}
which is of $O(n^2m^2)$.
}
\end{proof}
\subsubsection{Proof of Convergence to a Stable Matching}\label{sec:gen-conv-stable}
We outline here the key convergence argument and discuss the impossibility of guaranteeing convergence to the agent-optimal stable matching.

\begin{proposition}\label{prop:coopgenmarket}
In a matching market $\MM$, where agents follow Algorithm~\ref{alg:drr} and strategic firms follow Algorithm~\ref{alg:fdrr}, the algorithms converge to a stable matching, which is not necessarily agent-optimal.
\end{proposition}

\begin{proof}
By Definition~\ref{def:convergence}, convergence implies the existence of a time $t$ such that
$\FM_a(t') = \FM_a(t'-1)$ for all $t' > t$ and all agents $a$.
Let $t_1$ denote the start of the final committing phase.
By Observation~\ref{obs:drr1}, agents commit to a perfect matching
$(a,\FA_a(t_1))_{\forall a \in \AG}$ throughout $[t_1,T]$.

Since no switching condition is triggered after $t_1$, no agent detects a preference-list inconsistency and no firm performs a strategic rejection.
Consequently, each agent $a$ has already been rejected by every firm in
$\hat{\mathcal{H}}_{a,\FA_a(t_1)}(t_1)$, and each firm $f$ is matched with its most preferred agent among those who applied during the final updating phase.

Therefore, no blocking pair exists, and the matching
$(a,\FA_a(t_1))_{\forall a \in \AG}$ is stable.
Finally, as illustrated in Example~\ref{examp:coordfgs}, the limiting stable matching need not be agent-optimal.
\end{proof}

\subsubsection{Crucial Observations}\label{sec:gen-crucial-obs}
We now present two crucial observations that will be used repeatedly in the proof of~\cref{thm:decregmain}. 

The following observations clarify how the set of top-$n$ firms constrains stable outcomes in general markets.

\begin{observation}\label{obs:gmarket1}
In any matching market $\MM$, for every agent $a \in \AG$, the agent-optimal stable match
$\overline{f^*_a}$ belongs to the set of top-$n$ firms $\FI^{(n)}_a$. Formally,
\[
\overline{f^*_a} \in \FI^{(n)}_a.
\]
\end{observation}

\begin{proof}
Running $\FGS$ on the ground-truth preference lists yields the agent-optimal stable matching.
Since there are $n$ agents, each agent can be rejected at most $n-1$ times during the execution
of $\FGS$. Therefore, the agent-optimal stable match of any agent must lie among its top $n$
ranked firms, that is, in $\FI^{(n)}_a$.
\end{proof}

\begin{observation}\label{obs:gmarket2}
In any matching market $\MM$, if for some agent $a \in \AG$ the agent-pessimal stable match
$\underline{f^*_a}$ does not belong to the set $\FI^{(n)}_a$, that is,
\(
\underline{f^*_a} \notin \FI^{(n)}_a,
\)
then the market admits more than one stable matching. Equivalently,
\(
\underline{f^*_a} \neq \overline{f^*_a}.
\)
\end{observation}

\begin{proof}
If $\underline{f^*_a} \notin \FI^{(n)}_a$ while $\overline{f^*_a} \in \FI^{(n)}_a$ by
Observation~\ref{obs:gmarket1}, then the agent-optimal and agent-pessimal stable matches of $a$
must be distinct. Hence, the market admits more than one stable matching.
\end{proof}

\begin{observation}\label{obs:topnescorrect}
In a matching market $\MM$ with strategic firms following~\cref{alg:fdrr} and agents following~\cref{alg:drr}, consider any updating stage starting at time $\TGS$, and let
\(
t' \doteq \TGS + 3n^2 + 1
\)
denote the start of the subsequent committing phase. Then, in the perfect matching
$(a,\FM_a(t'))_{\forall a \in \AG}$ to which agents commit, each agent $a$ is matched to a firm
$\FM_a(t')$ that belongs to its top-$n$ estimated firms at time $\TGS$, i.e.,
\[
\FM_a(t') \in \hat{\FI}^{(n)}_a(\TGS).
\]
\end{observation}

\begin{proof}
The key observation is that, throughout the updating stage, agents select their actions based on the estimated preference lists fixed at time $\TGS$. As a result, during the ensuing distributed execution of $\FGS$, each agent applies only to firms among its top-$n$ estimated choices at time $\TGS$, which implies the claim.
\end{proof}

Now we conclude this section with the following lemma, which plays a crucial role in the proof of~\cref{thm:decregmain}.

\begin{lemma}\label{lemma:sortedcommit}
In a matching market $\MM$ with strategic firms following~\cref{alg:fdrr} and agents following~\cref{alg:drr}, consider any updating stage starting at time $\TGS$, and let
\(
t' \doteq \TGS + 3n^2 + 1
\)
denote the start of the subsequent committing phase. Suppose that the perfect matching
$(a,\FM_a(t'))_{\forall a \in \AG}$ to which agents commit is not stable with respect to the ground-truth preferences.
If there exists a time step $t'' \ge t'$ during the committing phase such that
\[
t'' \in \Gamma^{(n)}_\AG \cap \Gamma^{(n)}_\FI,
\]
then agents initiate the next updating phase at time $t''+1$.
\end{lemma}

\begin{proof}
The claim follows from~\cref{obs:topnescorrect} together with the definition of
$\Gamma^{(n)}_\AG$ and $\Gamma^{(n)}_\FI$ in~\cref{def:globaltopkalign}.

Since the perfect matching $(a,\FM_a(t'))_{\forall a \in \AG}$ is not stable with respect to the ground-truth preferences, there must exist at least one agent or firm whose top-$n$ estimated preference list was not aligned with the ground truth (either element-wise or order-wise) during the updating phase starting at $\TGS$. Equivalently, there exists a time step
\(
t_1 \in [\TGS,\TGS+3n^2]
\)
such that
\(
t_1 \notin \Gamma^{(n)}_\AG \cap \Gamma^{(n)}_\FI.
\)

Now consider the committing phase, and let $t'' \ge t'$ be the first time step such that
\(
t'' \in \Gamma^{(n)}_\AG \cap \Gamma^{(n)}_\FI.
\)
At time $t''$, all agents’ top-$n$ estimated firms and all firms’ top-$n$ estimated agents coincide with the ground truth. Since the matching $(a,\FM_a(t'))_{\forall a \in \AG}$ is unchanged throughout the committing phase and is unstable with respect to the ground-truth preferences, it is also unstable with respect to the estimated preferences at time $t''$. Hence a blocking pair exists.

Let $(a_i, f)$ be such a blocking pair at time $t''$, where $f=\FM_{a_j}(t')$ for some $a_j$. By definition of a blocking pair,
\(
f \underset{\PL_{a_i}}{>} \FM_{a_i}(t')
\)
and
\(
a_i \underset{\PL_f}{>} a_j.
\)
We consider the following cases.

\paragraph{(i)}
If $\FM_{a_i}(t') \notin \FI^{(n)}_{a_i}$, then since
\(
\FI^{(n)}_{a_i} = \hat{\FI}^{(n)}_{a_i}(t'')
\)
element-wise and order-wise, it follows that
\(
\FM_{a_i}(t') \notin \hat{\FI}^{(n)}_{a_i}(t''),
\)
which induces a preference-list inconsistency for agent $a_i$, and therefore
\(
t'' \in \mathcal{S}^{\mathrm{inc}}_{a_i}.
\)

\paragraph{(ii)}
If $\FM_{a_i}(t') \in \FI^{(n)}_{a_i}$ and agent $a_i$ did not apply to $f$ during the previous updating phase, then necessarily
\(
f \underset{\PL_{a_i}}{<} \FM_{a_i}(t')
\)
at time $\TGS$. At time $t''$, top-$n$ alignment implies the ordering reverses to
\(
f \underset{\PL_{a_i}}{>} \FM_{a_i}(t'),
\)
which again yields a preference-list inconsistency for agent $a_i$, and hence
\(
t'' \in \mathcal{S}^{\mathrm{inc}}_{a_i}.
\)

\paragraph{(iii)}
If $\FM_{a_i}(t') \in \FI^{(n)}_{a_i}$ and agent $a_i$ did apply to $f$ during the previous updating phase but was rejected, then $f$ must have rejected $a_i$ in favor of its current match $a_j$ due to a preference-list inconsistency. At time $t''$, the estimated list of $f$ is fully aligned with the ground truth and satisfies
\(
a_i \underset{\PL_f}{>} a_j,
\)
which triggers a firm-side strategic rejection, and therefore
\(
t'' \in \mathcal{S}^{\mathrm{rej}}_{a_j}.
\)

In all cases, a switching condition is triggered at time $t''$, and therefore $t''+1$ is the first round of the next updating phase.
\end{proof}

\subsubsection{Proof of~\cref{thm:decregmain} for General Markets}\label{proof:decreg}
\thmdecregmain*
\begin{proof}
We bound the pessimal regret $\mathbb{E}\!\left[\underline{R}_{a}(T)\right]$ by decomposing it into regret incurred during the updating and committing phases and bounding each term separately:
\begin{align*}
\mathbb{E}\!\left[\underline{R}_{a}(T)\right]
=
\mathbb{E}\!\left[\underline{R}_{a}(\mathcal{U})\right]
+
\mathbb{E}\!\left[\underline{R}_{a}(\mathcal{C})\right].
\end{align*}

For $\mathbb{E}\!\left[\underline{R}_{a}(\mathcal{U})\right]$, we write:
\begin{align}
\mathbb{E}\!\left[\underline{R}_{a}(\mathcal{U})\right] =\sum_{t \in \mathcal{U}}
\left(
u_{a,\underline{f^*_a}}
-
\mathbb{E}[X_{a,\FM_{a}(t)}]
\right)
&\le
3n^2 \cdot\,
\mathbb{E}\!\left[
\left|
\bigcup_{a' \in \AG}
\left(
\mathcal{S}^{\mathrm{inc}}_{a'}\cup \mathcal{S}^{\mathrm{rej}}_{a'}
\right)
\right|
\right].
\label{ineq:genupdate-regret}
\end{align}

Before bounding the committing regret, we introduce the following notation.
Let
\[
\bigcup_{a \in \AG} \left(\mathcal{S}^{\mathrm{inc}}_a \cup \mathcal{S}^{\mathrm{rej}}_a\right)
= \{t_1,t_2,\ldots\}
\]
denote the set of switching times, ordered increasingly.
The $j$-th committing phase is defined as the interval
\[
[t'_j : t_{j+1}],
\qquad
\text{where } t'_j \doteq t_j + 3n^2 + 1,
\]
which corresponds to the period between the $j$-th and $(j+1)$-th distributed executions of $\FGS$.
Accordingly, the collection of committing intervals is
\[
\mathcal{K}
\doteq
\left\{
[t'_j,\, t_{j+1}] :
j \in \left[\left|\!\bigcup_{a \in \AG}(\mathcal{S}^{\mathrm{inc}}_a \cup \mathcal{S}^{\mathrm{rej}}_a)\right|-1\right]
\right\}.
\]

Now for the committing regret we write 
{\small\begin{align}
\mathbb{E}\!\left[\underline{R_{a}}(\mathcal{C})\right]
&= \sum_{t \in \mathcal{C}}
\left(
u_{a,\underline{f^*_a}}
-
\mathbb{E}[X_{a,\FM_{a}(t)}]
\right) \nonumber\\
&=
\mathbb{E}\!\left[
\sum_{j=1}^{|\mathcal{K}|}
\sum_{t \in [t'_j,\, t_{j+1}]}
\Big(
u_{a,f}
-
X_{a,\FM_{a}(t'_j)}(t)
\Big) 
\right],\nonumber\\
&\overset{(a)}{\le} \mathbb{E}\!\left[
\left|
\bigcup_{a' \in \AG}
\left(
\mathcal{S}^{\mathrm{inc}}_{a'} \cup \mathcal{S}^{\mathrm{rej}}_{a'}
\right)
\right|
\right] +
\mathbb{E}\!\left[
\sum_{j \in [|\mathcal{K}|]}
\bigl(t_{j+1}-1-t'_j\bigr)
\;\middle|\;
\FM_{a}(t'_j) \in \mathcal{L}_{a,\underline{f^*_a}}
 \land (a'',\FM_{a''}(t'_j))_{\forall a'' \in \AG} \text{ is not stable }\right],  \nonumber\\
&\overset{(b)}{\le}
\mathbb{E}\!\left[
\left|
\bigcup_{a' \in \AG}
\left(
\mathcal{S}^{\mathrm{inc}}_{a'} \cup \mathcal{S}^{\mathrm{rej}}_{a'}
\right)
\right]
\right]+
\mathbb{E}\!\left[\left|\Bar{\Gamma}^{(n)}_{\AG} \cup \Bar{\Gamma}_{\FI}^{(n)}\right|\right].\label{ineq:finaldec}
\end{align}
}

Inequality~(a) is obtained by separating, within each committing phase, the single signaling round during which an agent intentionally abstains from applying in order to signal a switching condition. Each committing phase contributes at most one such round, and therefore the total number of signaling rounds is bounded by
\(
\mathbb{E}\!\left[
\left|
\bigcup_{a' \in \AG}
\left(
\mathcal{S}^{\mathrm{inc}}_{a'} \cup \mathcal{S}^{\mathrm{rej}}_{a'}
\right)
\right|
\right].
\)
The second term captures the expected regret accumulated during the remaining rounds of each committing phase. Since we analyze the pessimal regret of agent~$a$, regret is incurred only when $a$ is matched to a firm that it strictly prefers less than its pessimal stable match $\underline{f^*_a}$, i.e., only when
\(
\FM_a(t'_j) \in \mathcal{L}_{a,\underline{f^*_a}}.
\)
In this case, the perfect matching reached at the end of the most recent updating stage, $(a'',\FM_{a''}(t'_j))_{\forall a'' \in \AG}$, is necessarily unstable, which justifies the conditioning in the second term of~(a).

Inequality~(b) follows from~\cref{lemma:sortedcommit}, since the matching
\((a'', \FM_{a''}(t'_j))_{\forall a'' \in \AG}\) is unstable.
Specifically, consider the committing interval \([t'_j,\, t_{j+1}]\).
If every time step in this interval belongs to
\(\Bar{\Gamma}^{(n)}_{\AG} \cup \Bar{\Gamma}^{(n)}_{\FI}\),
then the entire interval is charged to the second term of~\eqref{ineq:finaldec}, while the signaling round is accounted for separately in the first term.
Otherwise, there exists a time step \(t'' \in [t'_j,\, t_{j+1}]\) such that
\(t'' \in \Gamma^{(n)}_{\AG} \cap \Gamma^{(n)}_{\FI}\).
By~\cref{lemma:sortedcommit}, this time step must be the signaling round $t_{j+1}$ that triggers the next updating phase, and hence it is counted in the first term of~\eqref{ineq:finaldec}.

The first term of~\eqref{ineq:finaldec} is bounded by~\cref{lemma:dec}. We therefore bound the remaining term $\mathbb{E}\!\left[\left|\Bar{\Gamma}^{(n)}_{\AG} \cup \Bar{\Gamma}^{(n)}_{\FI}\right|\right]$ as follows:
\begin{align}
\mathbb{E}\!\left[\left|\Bar{\Gamma}^{(n)}_{\AG} \cup \Bar{\Gamma}_{\FI}^{(n)}\right|\right]
\le
\mathbb{E}\!\left[\left|\Bar{\Gamma}_{\AG}^{(n)}\right|\right]
+
\mathbb{E}\!\left[\left|\Bar{\Gamma}_{\FI}^{(n)}\right|\right] 
&\le
\sum_{a' \in \AG}
\sum_{f' \in \FI^{(n)}_{a'}}
\mathbb{E}\!\left[\left|\Bar{\mathcal{E}}_{a',f'}\right|\right]
+
\sum_{f'' \in \FI}
\sum_{a'' \in \AG}
\mathbb{E}\!\left[\left|\Bar{\mathcal{E}}_{f'',a''}\right|\right]. \label{eq:deccomitreg}
\end{align}
The final inequality follows by expanding the definitions of $\Bar{\Gamma}^{(n)}_{\AG}$ and $\Bar{\Gamma}^{(n)}_{\FI}$ and applying a union bound.

Thus, with respect to~\eqref{ineq:genupdate-regret} and~\eqref{eq:deccomitreg}, we obtain
{\small\begin{align*}
\mathbb{E}\!\left[\underline{R}_{a}(T)\right]
&=
\mathbb{E}\!\left[\underline{R}_{a}(\mathcal{U})\right]
+
\mathbb{E}\!\left[\underline{R}_{a}(\mathcal{C})\right]\\
&\overset{(a)}\leq
\left(3n^2+1\right)\cdot\mathbb{E}\!\left[
\left|
\bigcup_{a' \in \AG}
\left(
\mathcal{S}^{\mathrm{inc}}_{a'} \cup \mathcal{S}^{\mathrm{rej}}_{a'}
\right)
\right|
\right]
+
\sum_{a' \in \AG}
\sum_{f' \in \FI^{(n)}_{a'}}
\mathbb{E}\!\left[\left|\Bar{\mathcal{E}}_{a',f'}\right|\right]
+
\sum_{f'' \in \FI}
\sum_{a'' \in \AG}
\mathbb{E}\!\left[\left|\Bar{\mathcal{E}}_{f'',a''}\right|\right]\\
&\in O\!\left(n^4m^2\Delta^{-2}\right),
\end{align*}
}
where $\Delta \doteq \min_{a\in\AG, f\in\FI}\min(\Delta_{a,f}, \Delta_{f,a})$.
\end{proof}

\section{Deferred Concepts and Proofs from Section~\ref{sec:non-coopmain}}
This section collects additional definitions and deferred proofs from Section~\ref{sec:non-coopmain} to streamline the main presentation.

\subsection{Extended Coordination-Free Decentralized Algorithm with Anonymous Hiring Changes $\VFP(t)$ as Feedback}\label{sec:non-coop}

We now design \cref{alg:ancdrrap} for the anonymous hiring-changes feedback $\VFP(t)$. Under this feedback, an agent treats any firm that appears in $\VFP(t)$ (either because it hired a new agent or became vacant) as a local signal that it may be worth reconsidering. Accordingly, an agent may revisit firms ranked above its previous match whenever those firms appear in $\VFP(t)$, viewing such events as opportunities to compete with the new hire and improve its match.

\begin{algorithm}[h]
\caption{Coordination-Free Decentralized Learning}
\label{alg:ancdrrap}
\begin{algorithmic}[1]
    \State \textbf{Input:} $a$, $\FI$, anonymous hiring changes $\VFP$
    \State \textbf{Initialize:} $r_{a,f}(1) \gets 0$ for all $f \in \FI$
    \For{$t \in \mathcal{T}$}
        \State Construct $\CSFP_a(t)$ according to \eqref{def:candidateset}
        \State $\FA_a(t) \gets \arg\max_{f \in \CSFP_a(t)} \hat{u}_{a,f}(t)$
        \State Interview with $\INT_a(t) = \{\FA_a(t), \RR_a(t)\}$ and apply to $\FA_a(t)$
        \State \textsc{UpdateAgentRejVars}$(a,t)$
    \EndFor
\end{algorithmic}
\end{algorithm}

 To choose $\FA_a(t)$, the agent first constructs a \emph{candidate set} $\CSFP_a(t)$ consisting of firms that either have never rejected $a$, or have exhibited a hiring change since the last time they rejected $a$ to hire another agent:

{\small
\setlength{\abovedisplayskip}{-10pt}
\begin{align}\label{def:candidateset}
\CSFP_a(t) \doteq 
\left\{\, f : \exists\, t' \in [r_{a,f}(t),t),f \in \VFP(t') \,\right\}
\end{align}
}
{\setlength{\abovedisplayskip}{-12pt}
The agent then applies to the firm in $\CSFP_a(t)$ with the highest estimated utility,
\(
\FA_a(t) = \argmax_{f \in \CSFP_a(t)} \hat{u}_{a,f}(t).
\)
}



For the agent-dependent bound in \cref{thm:noncodecreg}, In the worst case,
$\sum_{j'\in[i]}|\mathcal{H}_{a_{j'},f_{j'}}|=O(i^2)$, so agent $a_i$ has regret $O(i^3m^2)$; hence the worst-case over agents is $O(n^3m^2)$.

To handle general (not necessarily structured) markets in a coordination-free manner, we extend~\cref{alg:ancdrr} to the randomized~\cref{alg:Eancdrr}, which uses $k=3$ interviews per round and a fixed randomization parameter $\lambda\in(0,1)$. 


The exponentially large factor $\epsilon^{-1}$ reflects the need to rule out blocking-pair resolution cycles that may arise in general markets under coordination-free dynamics; such cycles do not occur in $\alpha$-reducible markets. Similar large constants also appear in the $O(\log T)$ bounds of~\cite{liu2021bandit}, albeit under more informative feedback than our $\VFP$.

\subsection{Crucial Structural Observation}
\begin{restatable}{observation}{obsdecnoncopp}\label{obs:decnoncop2}
By any time $t$, each agent $a \in \AG$ must have been rejected by every firm $f \in \hat{\mathcal{H}}_{a,\FA_a(t)}(t)$ since the most recent hiring change at time $t'$, i.e., $f \in \VFP(t')$ and $r_{a,f}(t) \ge t'$.
\end{restatable}
\begin{proof}
Algorithm~\ref{alg:ancdrr} starts with $\CSFP_a(1) = \FI$. By definition of $\CSFP_a(t)$, whenever the winner of a firm $f$ changes at time $t'$, we have $f \in \CSFP_a(t'+1)$. Thus, at $t'+1$, $\FA_a(t'+1)$ is chosen so that $\hat{u}_{a,\FA_a(t'+1)} \geq \hat{u}_{a,f}$, and $a$ applies to a more preferred firm than $f$.  
Since $\FA_a$ moves downward over time toward less preferred firms, once $f \in \hat{\mathcal{H}}_{a,\FA_a(t)}$, there must exist some $t''$ with $t'<t''<t$ where $\FA_a(t'')=f$. At that round, $a$ applies to $f$, is rejected, and does not reconsider $f$ until its winner changes again.
\end{proof}

\subsection{Proof Sketch of~\cref{thm:noncodecreg} for the agent $a_1$ in $\alpha$-reducible Markets} \label{Sec:genuncertainproofsketchalpha}
\begin{proof}[Proof Sketch]
 Our proof has three steps: (i) exclude the (time-independent) \emph{invalid} rounds for the pair $(a_1,f_1)$ and reduce regret to counting them; (ii) show that on any two consecutive \emph{valid} rounds, $a_1$ is matched to $f_1$ in the second; and (iii) conclude that within each maximal valid interval $a_1$ can incur regret only in its first round, and the number of such intervals is at most the number of invalid rounds plus one. Concretely, recalling $|\mathcal{H}_{a_1,f_1}|=0$, we obtain the time-independent bound
\(
\mathbb{E}[R_{a_1}(\mathcal{T})]\le 2\,\mathbb{E}\!\left[|\Bar{\mathcal{E}}_{a_1,f_1}\cup \Bar{\mathcal{E}}_{f_1,a_1}|\right]+1,
\)
where $\Bar{\mathcal{E}}_{a_1,f_1}\cup \Bar{\mathcal{E}}_{f_1,a_1}$ is the set of rounds in which either $a_1$ or $f_1$ has an invalid estimated list w.r.t.\ their optimal match, and its expected size is time-independent by \cref{lemma:valid-rr}.

Consider two consecutive rounds $t$ and $t+1$ with both lists valid w.r.t.\ the optimal match, i.e.,
$t,t+1\in \mathcal{E}_{a_1,f_1}\cap \mathcal{E}_{f_1,a_1}$.
We show $\FM_{a_1}(t+1)=f_1$ (while $\FM_{a_1}(t)$ may or may not equal $f_1$). This implies that within any maximal \emph{valid} interval (a consecutive block of rounds outside $\Bar{\mathcal{E}}_{a_1,f_1}\cup \Bar{\mathcal{E}}_{f_1,a_1}$), $a_1$ can incur regret in at most its first round; moreover, since valid and invalid intervals alternate, the number of valid intervals is at most $|\Bar{\mathcal{E}}_{a_1,f_1}\cup \Bar{\mathcal{E}}_{f_1,a_1}|+1$.

To prove $\FM_{a_1}(t+1)=f_1$, first note that if $\FA_{a_1}(t)=f_1$, then $\FM_{a_1}(t)=f_1$ immediately. Otherwise, by the definition of the candidate set $\CSFP_{a_1}(t)$ and \cref{obs:decnoncop2}, at the most recent time $a_1$ applied to $f_1$ and was rejected, firm $f_1$ had already hired some agent $a_i$ with $i\neq 1$, and this agent remained matched to $f_1$ up to round $t$. By the strategic rejection condition $\mathcal{S}_{f_1}$ in~\eqref{def:uncertainfirmstra}, agent $a_1$ is now estimated preferable to $a_i$, and since $r_{f_1,a_1}(t)\ge c_{f_1}(t)$, firm $f_1$ sets $\NH_{f_1}(t)=0$. This abstention is revealed to $a_1$ as a hiring change, prompting $a_1$ to reconsider $f_1$ next round, hence $\FA_{a_1}(t+1)=\FM_{a_1}(t+1)=f_1$.
\end{proof}

\subsection{Proof of~\cref{thm:noncodecreg} in $\alpha$-reducible Markets} \label{Sec:genuncertainproofalpha}
\thmnoncodecregmain*
\begin{proof}
We present the proof for the case of uncertain firms; the case of certain firms then follows directly by setting the firm-uncertainty parameters to zero. Now int the $\alpha$-reducible market $\MM$, for agent $a_i$, we write
{\small
\begin{align}
    \mathbb{E}\!\left[R_{a_i}(T)\right]
    &\leq \mathbb{E}\!\left[\left|\left\{\,t \in \mathcal{T} : \FM_{a_i}(t) \in \emptyset \cup \mathcal{L}_{a_i,f_i}\,\right\}\right|\right],\nonumber\\
    &\leq \mathbb{E}\!\left[\left|\Bar{\mathcal{E}}_{a_i,f_i} \cup \left\{\,t \in \mathcal{E}_{a_i,f_i} : \FM_{a_i}(t) \in \emptyset \cup \mathcal{L}_{a_i,f_i}\,\right\}\right|\right],\nonumber\\
    &\leq \mathbb{E}\!\left[\left|\Bar{\mathcal{E}}_{a_i,f_i}
        \cup \left\{\,t \in \mathcal{E}_{a_i,f_i} : \FA_{a_i}(t) \underset{\EPL_{a_i}(t)}{<} f_i\,\right\}
        \cup \left\{\,t \in \mathcal{E}_{a_i,f_i} : \FA_{a_i}(t) \underset{\EPL_{a_i}(t)}{\geq} f_i \,\land\, \FM_{a_i}(t)=\emptyset\,\right\}\right|\right],\nonumber\\
    &\overset{(a)}{\leq}\; \mathbb{E}\!\left[\left|\Bar{\mathcal{E}}_{a_i,f_i}
        \cup \left(\Bar{\mathcal{E}}_{f_i,a_i} \cup \bigcup_{j \in [i-1]}\!\left(\Bar{\mathcal{E}}_{a_j,f_j} \cup \Bar{\mathcal{E}}_{f_j,a_j}\right)\right)
        \cup \left\{\,t \in \mathcal{E}_{a_i,f_i} : \FA_{a_i}(t) \underset{\EPL_{a_i}}{\geq} f_i \,\land\, \FM_{a_i}(t)=\emptyset\,\right\}\right|\right],\nonumber\\
    &\leq \mathbb{E}\!\left[\left| \bigcup_{j \in [i]} \left(\Bar{\mathcal{E}}_{a_j,f_j} \cup \Bar{\mathcal{E}}_{f_j,a_j} \right)\cup\left\{\, t \in \bigcap_{j' \in [i]} \!\left( \mathcal{E}_{a_{j'},f_{j'}} \cap \mathcal{E}_{f_{j'},a_{j'}} \right) : \FM_{a_i}(t)=\emptyset \,\right\}\right|\right],\nonumber\\
       & \leq \mathbb{E}\!\left[\left| \bigcup_{j \in [i]} \left(\Bar{\mathcal{E}}_{a_j,f_j} \cup \Bar{\mathcal{E}}_{f_j,a_j} \cup \Bar{\Phi}_{j}\right)\right|\right]+ \mathbb{E}\!\left[\left|\left\{\, t \in \bigcap_{j' \in [i]} \Phi_{j'} : \FM_{a_i}(t)=\emptyset \,\right\}\right|\right] \label{eq:proofcoordfree1},
\end{align}
}
where $\Phi_i$ is defined as  
{\small
\begin{align}
\Phi_i \doteq \left\{ 
t \in \!\left( \mathcal{E}_{a_i,f_i} \cap \mathcal{E}_{f_i,a_i} \right) :
\left(\nexists f \in \FI \; \text{s.t.} \; f \in \mathcal{H}_{a_i,f_i} \land f \in \hat{\mathcal{L}}_{a_i,f_i}(t)\right)
\land
\left(\nexists a \in \AG \; \text{s.t.} \; a \in \mathcal{H}_{f_i,a_i} \land a \in \hat{\mathcal{L}}_{f_i,a_i}(t)\right)
\right\}. \label{def:phi}
\end{align}
}
That is, $\Phi_i$ denotes the set of time steps during which both $f_i$ and $a_i$ possess valid preference lists with respect to their optimal stable match, and all peers preferred to $f_i$ or $a_i$ under the ground truth are also ranked higher in their estimated lists (though not necessarily in the correct internal order).

Then the inequality (a) holds because at each time step within the set
\(
\left\{t \in \mathcal{E}_{a_i,f_i} : \FA_{a_i}(t) \underset{\EPL_{a_i}(t)}{<} f_i \right\},
\)
the chosen firm $\FA_{a_i}(t)$ is estimated to be less preferred than $f_i$.  
By Observation~\ref{obs:decnoncop2}, this implies that agent $a_i$ must have applied to $f_i$ at some earlier round $t' < t$ and been rejected and by definition of $\CSFP_{a}(t)$, the hired agent $\FFM_{f_i}(t)$ has not changed since then. We now examine why the most recent time $t'$ this rejection occurred which then hold that then hired agent $\FFM_{f_i}(t')$ remains the same at $t$.  

i) If $\EPL_{f_i}(t')$ was invalid at that round, then $\EPL_{f_i}(t)$ must also remain invalid at all subsequent rounds, since the agent hired by $f_i$ at time $t'$, namely $\FFM_{f_i}(t')$, has not changed until $t$. Hence, $t \in \Bar{\mathcal{E}}_{f_i,a_i}$.  

ii) If $\EPL_{f_i}(t')$ was valid, then $a_i$ must have been blocked by some other agent $a_j \in \mathcal{H}_{f_i,a_i}$ with $j \in [i-1]$.  
In this case, the hired agent $\FFM_{f_i}(t')$ has also not changed throughout $[t',t]$, so the blocking of $a_i$ persists.  
This persistence implies that either some higher-layer agent $a_j$ has an invalid preference list $\EPL_{a_j}(t'')$, or some firm with higher index $f_j$ has an invalid preference list $\EPL_{f_j}(t'')$ for all $t'' \in [t',t]$.  
Otherwise, by $\alpha$-reducibility, $\FFM_{f_i}(t')$ would eventually have been forced to leave $f_i$, since $f_i$ would be suboptimal for that agent, contradicting the fact that $\FFM_{f_i}(t')$ remained unchanged until $t$.
Putting i) and ii) together, we conclude that every such $t$ must fall in one of the following categories: (i) $a_i$ or $f_i$ has an invalid list, or (ii) some higher-ranked agent $a_j$ or firm $f_j$ has an invalid list. Formally,
\[
\left\{t \in \mathcal{E}_{a_i,f_i}: \FA_{a_i}(t) \underset{\EPL_{a_i}(t)}{<} f_i\right\}\subseteq
\left(\Bar{\mathcal{E}}_{f_i,a_i} \cup \bigcup_{j \in [i-1]} \bigl(\Bar{\mathcal{E}}_{a_j,f_j} \cup \Bar{\mathcal{E}}_{f_j,a_j}\bigr)\right).
\]

To establish the regret bound for the case of uncertain firms, we bound inequality~\eqref{eq:proofcoordfree1} differently. For the first term, we have
\begin{align}
    \mathbb{E}\!\left[\left| \bigcup_{j \in [i]} 
        \left(\Bar{\mathcal{E}}_{a_j,f_j} 
        \cup \Bar{\mathcal{E}}_{f_j,a_j} 
        \cup \Bar{\Phi}_{j}\right)\right|\right] 
    &\leq  
    \sum_{j \in [i]} 
    \frac{4m\exp\!\left(-\tfrac{\overline{\Delta}_{a_j}^2}{2m}\right)}
    {1 - \exp\!\left(-\tfrac{\overline{\Delta}_{a_j}^2}{2m}\right)}
    + 
    \frac{4n\exp\!\left(-\tfrac{\underline{\Delta}_{f_j}^2}{2m}\right)}
    {1 - \exp\!\left(-\tfrac{\underline{\Delta}_{f_j}^2}{2m}\right)} \nonumber\\
    &\in \sum_{j\in[i]} O\!\left(m^2\overline{\Delta}_{a_j}^{-2} + nm\,\underline{\Delta}_{f_j}^{-2}\right).\label{ineq:certainalphanoncoord3} 
\end{align}

To bound the second term of inequality~\eqref{eq:proofcoordfree1}, recall that, by Observation~\ref{obs:rejpol}, strategic certain firms always hire. Hence, for such firms, we can bound the second term as
{\small
\begin{align}
   \mathbb{E}\!\left[\left|\left\{\, t \in \bigcap_{j' \in [i]} \Phi_{j'} : \FM_{a_i}(t)=\emptyset \,\right\}\right|\right]
   \leq \left(1 + \sum_{j'\in [i]}\left|\mathcal{H}_{a_{j'},f_{j'}}\right|\right)
   \cdot\left(\mathbb{E}\!\left[\left| \bigcup_{j \in [i]} 
   \left(\Bar{\mathcal{E}}_{a_j,f_j} \cup \Bar{\mathcal{E}}_{f_j,a_j} 
        \cup \Bar{\Phi}_{j}\right)\right|\right]+1\right).
   \label{ineq:certainalphancoop1}
\end{align}
}
This inequality holds as follows. Consider a time interval $[t_1, t_2]$ such that, for all $t' \in [t_1:t_2]$, we have $t' \in \bigcap_{j' \in [i]} \Phi_{j'}$. We now argue that there are at most 
$1 + \sum_{j' \in [i]} \left|\mathcal{H}_{a_{j'},f_{j'}}\right|$ initial consecutive rounds within this interval during which agent $a_i$ is not matched with $f_i$, i.e., 
$\FM_{a_i}(t') \neq f_i$ for all 
$t' \in \left[t_1 : t_1 + 1 + \sum_{j' \in [i]} \left|\mathcal{H}_{a_{j'},f_{j'}}\right|\right]$. 

For agent $a_1$, as outlined in the proof sketch of Theorem~\ref{thm:noncodecreg} in Section~\ref{sec:non-coop}, the $\alpha$-reducibility structure at time step $t_1$ implies that either $\FA_{a_1}(t_1) = f_1$ or $\NH_{f_1}(t_1) = 1$. Consequently, $\FM_{a_1}(t') = f_1$ for all $t' \in [t_1+1, t_2]$, according to the decision-making rules of agents under Algorithm~\ref{alg:ancdrr} and firms under Algorithm~\ref{alg:fdrr}. Thus, agent $a_1$ incurs regret only at time step $t_1$.

We now proceed by induction to establish the argument for agent $a_i$. 
Before doing so, we highlight a crucial observation. 
Let $t^*_i$ denote the last time step within the interval $[t_1:t_2]$ at which there exists some agent $a_{j}$ for $j \in [i-1]$ such that $\FM_{a_{j}}(t_i) \neq f_{j}$. 
Then, for all subsequent rounds $t' \in [t^*_i+1:t_2]$, the strategic uncertain firm $f_{i}$ will hire $a_{i}$ whenever it receives an application from it. 
This follows because, due to the validity condition of the preference lists of both sides, $a_{i}$ remains the most preferred agent among those who still consider $f_{i}$ as an active option, while all agents in $\mathcal{H}_{f_{i},a_{i}}$ have already secured their stable matches and no longer apply to $f_{i}$.

By the inductive hypothesis, the latest time at which any agent $a_{j'}$ for ${j'} \in [i-1]$ is not matched with its optimal firm happens within at most $1+\sum_{k' \in [i-1]} \left|\mathcal{H}_{a_{k'},f_{k'}}\right|$ rounds, thus $t^*_{i} \leq 1+\sum_{k' \in [i-1]} \left|\mathcal{H}_{a_{k'},f_{k'}}\right|$. Hence, during the interval $t'' \in \left[t_1 + 1+ \sum_{k' \in [i-1]}\left|\mathcal{H}_{a_{k'},f_{k'}}\right| : t_2\right]$, we have $\FM_{a_k}(t'') = f_k$ for all $k \in [i-1]$. 

Then, from the definition of $\bigcap_{j' \in [i]} \Phi_{j'}$, it follows that throughout $\left[t_1 + 1+\sum_{k' \in [i-1]}\left|\mathcal{H}_{a_{k'},f_{k'}}\right| : t_2\right]$, there are at most $|\mathcal{H}_{a_i,f_i}|$ consecutive rounds in which $a_i$ applies to firms and gets rejected. This is because, starting from time step $t_1 +1+ \sum_{k' \in [i-1]}\left|\mathcal{H}_{a_{k'},f_{k'}}\right|$, all firms in $\mathcal{H}_{a_{k'},f_{k'}}$ are matched with their optimal agents and no longer change their hires during $[t_1 + \left|\mathcal{H}_{a_{k'},f_{k'}}\right| + 1 : t_2]$. Therefore, under Algorithm~\ref{alg:ancdrr}, agent $a_i$ applies to each such firm at most once in sequence until eventually reaching $f_i$. Hence, $t_1 + 1+ \sum_{j' \in [i]} \left|\mathcal{H}_{a_{j'},f_{j'}}\right|$ is the last time step at which $a_i$ incurs regret within $[t_1:t_2]$, resulting in a total regret of $1+\sum_{j' \in [i]} \left|\mathcal{H}_{a_{j'},f_{j'}}\right|$ during the interval. Finally, as the total number of such intervals $[t_1:t_2]$ is bounded by  $\mathbb{E}\!\left[\left| \bigcup_{j \in [i]} \left(\Bar{\mathcal{E}}_{a_j,f_j} \cup\Bar{\mathcal{E}}_{f_j,a_j} 
        \cup  \Bar{\Phi}_{j}\right)\right|\right]+1$, we can prove the inequality~\eqref{ineq:certainalphancoop1}. 

Thus, with respect to inequalit~\eqref{eq:proofcoordfree1} for an strategic certain agent $a_i$, the regret is bounded by merging inequalitues~\eqref{ineq:certainalphanoncoord3} and~\eqref{ineq:certainalphancoop1} and get
{\small
\begin{align}
   \mathbb{E}\!\left[R_{a_i}(T)\right] &\leq 4\cdot\left(\sum_{j'\in [i]}\left|\mathcal{H}_{a_{j'},f_{j'}}\right|+2\right)\cdot\left(\sum_{j \in [i]} 
    \frac{4m\exp\!\left(-\tfrac{\overline{\Delta}_{a_j}^2}{2m}\right)}
    {1 - \exp\!\left(-\tfrac{\overline{\Delta}_{a_j}^2}{2m}\right)}
    + 
    \frac{4n\exp\!\left(-\tfrac{\underline{\Delta}_{f_j}^2}{2m}\right)}
    {1 - \exp\!\left(-\tfrac{\underline{\Delta}_{f_j}^2}{2m}\right)} \right)+ \sum_{j'\in [i]}\left|\mathcal{H}_{a_{j'},f_{j'}}\right|+1\nonumber\\
    &\leq 4\cdot\left(\sum_{j'\in [i]}\left|\mathcal{H}_{a_{j'},f_{j'}}\right|+2\right)\cdot\sum_{j \in [i]}\left( O\!\left(m^2\overline{\Delta}_{a_j}^{-2} + nm\,\underline{\Delta}_{f_j}^{-2}\right)\right)+ \sum_{j'\in [i]}\left|\mathcal{H}_{a_{j'},f_{j'}}\right|+1\nonumber\\
    &\in O\!\left(\left(\sum_{j \in [i]} |\mathcal{H}_{a_{j},f_{j}}| + 1\right)\cdot \left(\sum_{j\in[i]}\left(m^2\overline{\Delta}_{a_j}^{-2} + nm\,\underline{\Delta}_{f_j}^{-2}\right) + 1\right)\right).\nonumber
\end{align}
}
\end{proof}



\subsection{Extension of Coordination-Free Algorithm~\ref{alg:ancdrr} to General Markets}\label{sec:noncoopgen}
We now extend our results to general (unstructured) markets and focus on bounding the pessimal regret. We first present an example showing that, consistent with standard impossibility phenomena for decentralized learning in matching markets, when the underlying market admits multiple stable matchings, interviews alone may not suffice to ensure convergence to the agent-optimal stable matching. Consequently, in markets with potentially multiple stable outcomes, our goal is to bound the pessimal regret, namely the regret incurred when the dynamics converge to a stable matching that is suboptimal from the agents’ perspective. We then address two algorithmic necessities that arise in this general setting. First, we show that moving beyond $k=2$ interviews per round (in particular, allowing $k=3$) is necessary to avoid unresolved application cycles caused by concurrent attempts to eliminate blocking pairs. Second, we show that even with $k=3$, convergence can fail under the anonymous firm-side feedback $\VFP$ unless agents are allowed to apply to more than one firm in a round and then select among the firms that admit them, so as to remain matched without inducing spurious vacancy signals.

\subsubsection{\textbf{Impossibility of Guaranteeing Convergence to the Agent-Optimal Stable Matching}}
\begin{example}\label{examp:drrs-4}
Let the market consist of agents $\AG=\{a_1,a_2,a_3\}$ and firms $\FI=\{f_1,f_2,f_3\}$, with true preference lists
\[
\begin{array}{lr}
\PL_{a_1}: ( f^*_1 \succ f_2 \succ f^\dagger_3 ) & 
\PL_{f_1}: ( a^\dagger_2 \succ a_3 \succ a^*_1 ) \\[3pt]
\PL_{a_2}: ( f^*_2 \succ f_3 \succ f^\dagger_1 ) & 
\PL_{f_2}: ( a^\dagger_3 \succ a_1 \succ a^*_2 ) \\[3pt]
\PL_{a_3}: ( f^*_3 \succ f^\dagger_2 \succ f_1 ) & 
\PL_{f_3}: ( a^\dagger_1 \succ a_2 \succ a^*_3 ) .
\end{array}
\]
Under these preferences, multiple stable matchings exist. The agent-optimal (firm-pessimal) stable matching is $(a_1,f_1),(a_2,f_2),(a_3,f_3)$ (marked by $^*$), whereas the agent-pessimal (firm-optimal) stable matching is $(a_1,f_3),(a_2,f_1),(a_3,f_2)$ (marked by $^\dagger$).

Now consider a time step $t$ in which Algorithm~\ref{alg:ancdrr} is executed using the estimated preferences $\EPL(t)$. Suppose agent $a_1$ has not yet learned its true preference list, while agents $a_2$ and $a_3$ have already converged. The estimated lists and current interview sets are
\[
\begin{array}{lr}
\EPL_{a_1}(t): ( f_2 \succ f_1 \succ f_3 ) & I_{a_1}(t)=\{f_3,\RR_{a_1}\} \\[3pt]
\EPL_{a_2}(t): ( f_2 \succ f_3 \succ f_1 ) & I_{a_2}(t)=\{f_1,\RR_{a_2}\} \\[3pt]
\EPL_{a_3}(t): ( f_3 \succ f_2 \succ f_1 ) & I_{a_3}(t)=\{f_2,\RR_{a_3}\}.
\end{array}
\]
Thus, $a_1$ is the only agent with an incorrect estimate at time $t$. At a later time $t'>t$, round-robin interviewing corrects $a_1$’s estimate to $\EPL_{a_1}(t'):( f_1 \succ f_2 \succ f_3 )$. However, since the hired agents at $f_1$ and $f_2$ remain unchanged from time $t$, agent $a_1$ continues to set $\FA_{a_1}(t')=f_3$. Consequently, the dynamics converge to the agent-pessimal stable matching $(a_1,f_3),(a_2,f_1),(a_3,f_2)$. This shows that when multiple stable matchings exist, decentralized learning may fail to reach the agent-optimal outcome, and one can only guarantee regret with respect to the agent-pessimal stable matching.
\end{example}


\subsubsection{\textbf{Moving from $k=2$ to $k=3$}}
This section explains why allowing each agent to interview $k=3$ firms per round is necessary to obtain time-independent regret bounds in general (unstructured) markets. The core issue is concurrency: multiple agents may attempt to resolve blocking pairs simultaneously, and their actions can interfere. The following~\cref{examp:k3} shows that when $k=2$, Algorithm~\ref{alg:ancdrr} can enter an \emph{application cycle} that never resolves, even after all agents’ estimated preferences fully aligned with the ground truth.

\begin{example}\label{examp:k3}
Consider a market with agents $\AG=\{a_1,a_2\}$ and firms $\FI=\{f_1,f_2\}$. Suppose that from some time $t'$ onward the estimated preferences are exactly aligned with the ground truth and remain fixed, i.e., for all $t\ge t'$ we have $\EPL_a(t)=\PL_a$ for all $a\in\AG$ and $\EPL_f(t)=\PL_f$ for all $f\in\FI$. In particular, for every $t\ge t'$,
\[
\begin{array}{lr}
\EPL_{a_1}(t): ( f_1^* \succ f_2 ) & \EPL_{f_1}(t): ( a_2 \succ a_1 ) \\[3pt]
\EPL_{a_2}(t): ( f_2 \succ f_1^* ) & \EPL_{f_2}(t): ( a_1 \succ a_2 ) .
\end{array}
\]
Here, the $*$ superscripts indicate the firms $\FA_{a_i}(t)$ chosen by agents under Algorithm~\ref{alg:ancdrr}. 

Fix any round $t\ge t'$, and assume that at time $t-1$ agent $a_2$ applied to $f_2$ and was rejected because $f_2$ admitted $a_1$; hence $r_{a_2,f_2}(t)\ge t-1$, and since $f_2\notin \VFP(t-1)$ under this anonymous feedback, we have $f_2\notin \CSFP_{a_2}(t)$. Starting from round $t$, the algorithm can enter the following cycle:
\begin{itemize}
\item At time $t$, both agents apply to $f_1$, and $f_1$ admits $a_2$ and rejects $a_1$; hence $f_2\in \VFP(t)$.
\item At time $t+1$, agent $a_1$ applies to $f_2$ since $f_1\notin \CSFP_{a_1}(t+1)$ after the rejection, and $a_2$ also applies to $f_2$ because it is unmatched and $f_2\in \VFP(t)$. Firm $f_2$ admits $a_1$ and rejects $a_2$, so $f_1\in \VFP(t+1)$.
\item At time $t+2$, by the same reasoning both agents apply to $f_1$ again, and the process repeats indefinitely.
\end{itemize}
\end{example}

Thus, \cref{examp:k3} shows that even with perfectly learned preferences, Algorithm~\ref{alg:ancdrr} with $k=2$ can cycle forever. Such a cycle could be broken, for instance, by allowing agents to deviate from Algorithm~\ref{alg:ancdrr} with some probability and stick with $\FA_a(t-1)$. Instead, we adopt a deterministic remedy by enlarging each agent’s choice set: we allow agents to apply among multiple interviewed firms, including the firm they applied to in the previous round. Concretely, we add a third option by requiring each agent to always include $\FA_a(t-1)$ in its interview set. Accordingly, in our analysis we treat agents as also reconsidering $\FA_a(t-1)$ at every round.

We then extend Algorithm~\ref{alg:ancdrr} to $k=3$ by setting the interview set to include the current candidate, the previous application, and a round-robin firm, namely $$\INT_a(t)=\{\FA_a(t),\FA_a(t-1),\RR_a(t)\}.$$ This captures both remedies discussed above: a randomized deviation that lets the agent occasionally stick with $\FA_a(t-1)$, and a deterministic enlargement of the agent’s option set. In contrast to the $k=2$ version—where the agent effectively has a single action, applying to $\FA_a(t)$—under $k=3$ the agent interviews three firms and may apply to either $f_a(t)$ or $\FA_a(t-1)$, choosing between them based on the updated information revealed by $\INT_a(t)$ to choose $\FA_a(t)$.

\subsubsection{\textbf{Necessity of applying to more than one firm per round for coordination-free decentralized learning}}
We explain why agents must be allowed to apply to \emph{more than one} firm per round (i.e., apply to a set of interviewed firms), and why merely increasing the interview budget from $k=2$ to $k=3$ is still insufficient for convergence in general. The key obstacle is the \emph{anonymous} nature of the firm-side feedback $\VFP$: observing that a firm becomes vacant does not reveal \emph{which} agent left or whether the vacancy was caused by a strategic rejection. As a result, agents may repeatedly ``chase'' the same vacancy and create persistent application cycles, even when all estimated preference lists have already converged to the ground truth. The following example illustrates this phenomenon and shows how allowing an agent to apply to a \emph{set} of firms in a round can eliminate such cycles by preventing spurious vacancies from appearing in $\VFP$.

\begin{example}\label{examp:multappl}
Consider a market with agents $\AG=\{a_1,a_2\}$ and firms $\FI=\{f_1,f_2\}$. Suppose that from some time $t'$ onward the estimated preferences are exactly aligned with the ground truth and remain fixed, i.e., for all $t\ge t'$ we have $\EPL_a(t)=\PL_a$ for all $a\in\AG$ and $\EPL_f(t)=\PL_f$ for all $f\in\FI$. In particular, for every $t\ge t'$,
\[
\begin{array}{lr}
\EPL_{a_1}(t): ( f_1 \succ f_2^* ) & \EPL_{f_1}(t): ( a_2 \succ a_1 ) \\[3pt]
\EPL_{a_2}(t): ( f_2 \succ f_1^* ) & \EPL_{f_2}(t): ( a_1 \succ a_2 ) .
\end{array}
\]
Here, the $^*$ superscripts indicate the firms $\FA_{a_i}(t)$ chosen by agents under Algorithm~\ref{alg:ancdrr}. We also assume that $f_2$ rejected $a_2$ at some earlier time, so $f_2$ remains in $a_2$'s candidate set until $a_2$ applies to it again.

Fix any round $t\ge t'$. Consider the following two cases.
\begin{enumerate}
\item \textbf{Applying to a single firm can create a cycle.} Suppose that at time $t+1$ agent $a_1$ interviews $\INT_{a_1}(t+1)=\{f_1,f_2,\RR_{a_1}(t+1)\}$ and applies only to $f_1$, i.e., $\FA_{a_1}(t+1)=f_1$, while agent $a_2$ keeps applying to $f_2$. Since $f_1$ prefers $a_2$ to $a_1$, firm $f_1$ rejects $a_1$, and hence $f_2$ is reported in $\VFP(t+1)$ as a vacant firm. Because $\VFP(t+1)$ is anonymous, agent $a_2$ cannot distinguish whether $f_2$ appeared due to a strategic rejection or because its current match applied elsewhere; consequently, $a_2$ must keep $f_2$ in its candidate set and eventually applies to $f_2$ again. This induces a back-and-forth dynamic as in \cref{examp:k3}, yielding an application cycle that can persist indefinitely despite perfectly learned preferences.

\item \textbf{Applying to a set breaks the cycle.} Suppose instead that agents may apply to a \emph{set} of firms in a round. In the same situation at time $t+1$, let $a_1$ apply to both firms, i.e., $\FA_{a_1}(t+1)=\{f_1,f_2\}$. Even if $f_1$ rejects $a_1$, the simultaneous application to $f_2$ allows $a_1$ to remain (or become) matched to $f_2$ in the same round. Consequently, $f_2$ does \emph{not} appear in $\VFP(t+1)$, so $a_2$ receives no spurious vacancy signal and is not pulled into chasing $f_2$ again.
\end{enumerate}
\end{example}

In summary, under anonymous firm-side feedback, allowing an agent to apply to multiple firms in a round can suppress spurious vacancy signals and thereby eliminate persistent application cycles that would otherwise lead to linear regret.
.

\subsubsection{\textbf{Extended~\cref{alg:ancdrr}}}
With $k=3$, each agent $a$ forms the interview set as
\[
\INT_a(t)=\{f_a(t),\FA_a(t-1),\RR_a(t)\},
\]
where $f_a(t)$ is chosen exactly as in the $k=2$ version of Algorithm~\ref{alg:ancdrr}, namely
\[
f_a(t)=\argmax_{f\in \CSFP_a(t)} \hat{u}_{a,f}(t).
\]
The candidate set $\CSFP_a(t)$ is defined in \eqref{def:candidateset} as
\[
\CSFP_a(t)\doteq \{\, f : r_{a,f}(t)=0 \ \lor\ \exists\, t' \in [r_{a,f}(t),t),\ f\in \VFP(t') \,\}.
\]

Under this extended algorithm, each agent $a$ maintains a fixed probability parameter $\lambda\in(0,1)$ and, at the end of each round, updates its application as follows:
\[
\FA_a(t)=
\begin{cases}
\{f_a(t),\FA_a(t-1)\}, & \text{with probability }\lambda,\\
\FA_a(t-1), & \text{with probability }1-\lambda.
\end{cases}
\]

When $\FA_a(t)$ is a singleton, the application dynamics coincide with the original Algorithm~\ref{alg:ancdrr}: agent $a$ applies only to that firm and is matched according to the firm’s accept/reject decision. When $\FA_a(t)$ is a set, the agent applies to \emph{all} firms in $\FA_a(t)$ simultaneously. The agent then prioritizes $f_a(t)$: if $f_a(t)$ admits $a$, then $a$ accepts that offer and is matched to $f_a(t)$ for that round; otherwise, $a$ attempts to remain matched with $\FA_a(t-1)$. If $\FA_a(t-1)$ also rejects $a$, then $a$ remains unmatched in that round. 

Finally, the rejection variables are updated exactly as before: whenever agent $a$ is rejected by a firm $f$ in a \emph{non-strategic} manner, the rejection indicator $r_{a,f}(\cdot)$ is updated accordingly (whereas strategic rejections are treated separately and do not trigger the same update). The following Algorithm~\ref{alg:Eancdrr} presents the pseudocode for this extended procedure.

\begin{algorithm}[t]
\caption{Extended Decentralized Learning with $\VFP(t)$ (Agent $a$)}
\label{alg:Eancdrr}
\begin{algorithmic}[1]
    \State \textbf{Input:} $a$, $\FI$, anonymous hiring changes $\VFP$
    \State \textbf{Initialize:} $r_{a,f}(1) \gets 0$ for all $f \in \FI$
    \For{$t \in \mathcal{T}$}
        \State Construct $\CSFP_a(t)$ according to \eqref{def:candidateset}
        \State $f_a(t) \gets \arg\max_{f \in \CSFP_a(t)} \hat{u}_{a,f}(t)$
        \State $\INT_a(t) \gets \{f_a(t), \FA_a(t-1), \RR_a(t)\}$ and interview
        \State 
        \[
        \FA_a(t) =
        \begin{cases}
            \{f_a(t), \FA_a(t-1)\}, & \text{with probability } \lambda, \\
            \FA_a(t-1), & \text{with probability } 1 - \lambda
        \end{cases}
        \]
        \State Apply to all firms in $\FA_a(t)$; accept an offer from $f_a(t)$ if admitted, else accept an offer from $\FA_a(t-1)$ if admitted; otherwise remain unmatched
        \State \textsc{UpdateAgentRejVars}$(a,t)$
    \EndFor
\end{algorithmic}
\end{algorithm}

\subsubsection{\textbf{Crucial Theorems, Lemmas, and Observations}}

We begin by defining the notion of an \emph{agent-consistent} (a.k.a.\ player-consistent) blocking pair, adapted from~\cite{abeledo1995paths,liu2021bandit}.

\begin{definition}[Agent-consistent blocking pair]\label{def:agentconsistent}
A blocking pair $( a_i,f_j)$ in a matching $\mu$ is \emph{player-consistent} if $f_j \underset{\PL_{a_i}}{>} f$ for every firm $f$ such that $( a_i,f )$ is a blocking pair in $\mu$.
\end{definition}

We next recall a crucial result of \cite{abeledo1995paths} (stated here for completeness) showing that, under fixed preferences, repeatedly resolving blocking pairs in a player-consistent order reaches a stable matching.

\begin{theorem}[\cite{abeledo1995paths}]\label{thm:abeledopath}
Given any unstable matching $\mathcal{G}$, there exists a sequence of blocking pairs of length at most $n^4$ such that resolving this sequence reaches a stable matching. Moreover, this sequence can be chosen to resolve blocking pairs in a player-consistent order, i.e., every blocking pair $( a_i,f_j)$ resolved in the current matching $\mu$ is player-consistent with respect to $\mathcal{G}$. 
\end{theorem}

We now present \cref{lemma:noncoopgeninterval}, which bounds the expected maximum number of rounds for which the matching $(a,\FM_a(t))_{a\in\AG}$ remains unstable within any interval $[t_1:t_2]$ such that $t\in \Gamma^{(m)}_\AG \cap \Gamma^{(n)}_\FI$ for all $t\in[t_1:t_2]$ under~\cref{alg:Eancdrr}. By \cref{def:globaltopkalign}, these are precisely the time steps at which the estimated preference lists of all agents and firms coincide with the ground truth. This lemma plays a central role in the proof of \cref{thm:noncodecreg}.

\begin{lemma}\label{lemma:noncoopgeninterval}
In a matching market $\MM$ with strategic firms following~\cref{alg:fdrr} and agents following~\cref{alg:Eancdrr}, consider any interval $[t_1:t_2]$ such that $t\in \Gamma^{(m)}_\AG \cap \Gamma^{(n)}_\FI$ for all $t\in[t_1:t_2]$. Then, for any $\lambda \in (0,1)$,
\[
\mathbb{E}\!\left[
\sum_{t \in [t_1:t_2]}
\mathbbm{1}\left\{(a,\FM_{a}(t))_{\forall a \in \AG} \text{ is not stable}\right\}
\right]
\le 1
+
\frac{n^4m+nm}{\bigl(\lambda\cdot(1-\lambda)^{n-1}\bigr)^{n^4 m + nm}}.
\]
\end{lemma}
\begin{proof}
We first account for rounds in $[t_1:t_2]$ in which a firm-side strategic rejection can make the realized matching $(a,\FM_a(t))_{a\in\AG}$ unstable, regardless of the stability of the intended matching $(a,\FA_a(t))_{a\in\AG}$. We claim that for each firm $f\in\FI$, at most one such rejection can occur over $[t_1:t_2]$, and it can only occur at the first round $t_1$. Indeed, by the definition of $\Gamma^{(n)}_\FI$, for every $t\in[t_1:t_2]$ each firm’s estimated list coincides with its ground-truth list, so a firm never rejects in order to hire an agent that is worse under the ground truth. Hence, whenever the hiring state of $f$ changes during $[t_1:t_2]$, it either (i) hires an agent it truly prefers to its current match, or (ii) becomes vacant. Therefore, under~\cref{alg:fdrr} the only round at which a firm may strategically reject is $t_1$, which explains the additive $1$ term in the lemma statement.

It remains to bound
\[
\mathbb{E}\!\left[
\sum_{t \in [t_1+1:t_2]}
\mathbbm{1}\left\{(a,\FM_{a}(t))_{\forall a \in \AG} \text{ is not stable}\right\}
\right].
\]
Fix any $t\in[t_1+1:t_2]$ and let $\mathcal{G}(t)=(a,\FM_a(t))_{a\in\AG}$ be the realized matching at time $t$. By \cref{thm:abeledopath}, if $\mathcal{G}(t)$ is unstable then there exists a sequence $S=\{s_1,\dots,s_{n^4}\}$ of agent-consistent blocking pairs such that resolving these pairs in order reaches a stable matching.

We lower bound the probability of resolving this sequence by exhibiting a sufficient event for resolving each blocking pair. Assume inductively that $s_1,\dots,s_{i-1}$ have already been resolved, and let $s_i=( a',f'')$. Define $\mathcal{X}_i$ as the event that, for at most $m$ consecutive rounds after $s_1,\dots,s_{i-1}$ are resolved, agent $a'$ chooses the ``move'' option (hence applies according to $f_{a'}(\cdot)$) in every round, while every other agent $a\neq a'$ chooses the ``stay'' option (hence applies to $\FA_a(\cdot-1)$) in every round. Conditioned on $\mathcal{X}_i$, the agent $a'$ keeps advancing through its candidate set and, in the worst case, reaches and applies to $f''$ within $m$ rounds; when $a'$ applies to $f''$, the agent-consistency of $s_i$ implies that $f''$ admits $a'$ and rejects its current match, thereby resolving the blocking pair.

By independence of agents' randomized choices across rounds,
\[
\Pr(\mathcal{X}_i)\;\ge\;\bigl(\lambda\cdot(1-\lambda)^{n-1}\bigr)^{m}.
\]
Let
\[
\mathcal{X}\;\doteq\;\bigcap_{i=1}^{n^4}\mathcal{X}_i.
\]
By construction, $\mathcal{X}$ implies that all blocking pairs in $S$ are resolved in order, and hence the matching becomes stable after at most $n^4m$ rounds. Moreover,
\[
\Pr(\mathcal{X})\;\ge\;\bigl(\lambda\cdot(1-\lambda)^{n-1}\bigr)^{n^4m}.
\]

After $\mathcal{X}$ occurs, there exists a time $t'$ at which the matching has no blocking pair. However, due to the anonymous revelation of hiring changes in $\VFP$, an agent may still have firms in its candidate set that it prefers to its current match $\FM_a(t')$. Since there is no blocking pair, any such application is rejected provided the other agents remain with their stable matches. Under~\cref{alg:Eancdrr}, when an agent applies to a preferred firm while also applying to its current match $\FM_a(t')$, it remains matched with $\FM_a(t')$ upon rejection, and neither firm necessarily appears in $\VFP$. Thus, it may take up to an additional $nm$ rounds for all agents to exhaust such futile applications and reach a state in which they keep applying to (and remaining matched with) the stable matching.

Formally, define $\mathcal{J}$ as the event that, after $\mathcal{X}$ occurs, within the next $nm$ rounds the process reaches a time $\tau$ such that $\mathcal{G}(\tau)$ is stable and $\mathcal{G}(t)=\mathcal{G}(\tau)$ for all subsequent rounds in the interval. By the same independence argument,
\[
\Pr(\mathcal{J}\mid \mathcal{X})\;\ge\;\bigl(\lambda\cdot(1-\lambda)^{n-1}\bigr)^{nm}.
\]
Therefore,
\[
\Pr(\mathcal{X}\cap\mathcal{J})
\;=\;
\Pr(\mathcal{J}\mid \mathcal{X})\,\Pr(\mathcal{X})
\;\ge\;
\bigl(\lambda\cdot(1-\lambda)^{n-1}\bigr)^{n^4m+nm}.
\]
Hence, with probability at least $p\doteq \bigl(\lambda\cdot(1-\lambda)^{n-1}\bigr)^{n^4m+nm}$, within the next $n^4m+nm$ rounds the process reaches a stable matching that remains unchanged thereafter. Viewing this as a Bernoulli trial with success probability $p$, the expected number of rounds (starting from $t_1+1$) until such a stable-and-absorbing state is reached is at most $1/p$. Consequently,
\[
\mathbb{E}\!\left[
\sum_{t \in [t_1:t_2]}
\mathbbm{1}\left\{(a,\FM_{a}(t))_{\forall a \in \AG} \text{ is not stable}\right\}
\right]
\le
1
+
\frac{n^4m+nm}{\bigl(\lambda\cdot(1-\lambda)^{n-1}\bigr)^{n^4m+nm}}.
\]
\end{proof}

\subsubsection{Proof of~\cref{thm:noncodecreg} for General Markets}\label{proof:noncoordgen}
\thmnoncodecregmain*

\begin{proof}
Since we aim to bound the \emph{pessimal regret} $\mathbb{E}[\underbar{R}_a(T)]$, it suffices to measure regret with respect to agent~$a$’s pessimal stable match, denoted by $\underline{f^*_a}$. Accordingly, throughout the proof we measure regret relative to $\underline{f^*_a}$. Once the algorithm converges to a stable matching, agent~$a$ is matched either to $\underline{f^*_a}$ or to a strictly better stable match; therefore, all subsequent rounds incur zero (or negative) regret relative to $\underline{f^*_a}$. Hence, it suffices to bound the regret incurred prior to convergence in order to bound $\mathbb{E}[\underbar{R}_a(T)]$. We also denote by $\overline{a^*_f}$ the firm-optimal stable match of firm~$f$.

Next, with respect to~\cref{def:globaltopkalign}, consider the sets of time steps $\Gamma^{(m)}_{\AG}$ and $\Gamma^{(n)}_{\FI}$ in which all agents’ and all firms’ estimated preference lists, respectively, coincide with the ground truth. We first account for regret incurred on the complement rounds $\Bar{\Gamma}^{(m)}_\AG \cup \Bar{\Gamma}^{(n)}_\FI$, and then bound the maximum regret within each maximal interval $[t_i:t'_i]$ such that every $t\in[t_i:t'_i]$ satisfies $t\in \Gamma^{(m)}_\AG \cap \Gamma^{(n)}_\FI$.

Let $\mathcal{Y}$ denote the set of starting time steps of these maximal aligned intervals. Formally,
\[
\mathcal{Y}
\doteq
\left\{
t \in \mathcal{T} :
t \in \Gamma^{(m)}_{\AG}\cap \Gamma^{(n)}_{\FI}
\ \text{and}\
\bigl(t=1 \ \text{or}\ t-1 \notin \Gamma^{(m)}_{\AG}\cap \Gamma^{(n)}_{\FI}\bigr)
\right\}.
\]
For each $t_i\in\mathcal{Y}$, let
\[
t'_i \doteq \max\left\{t' \in [T] : [t_i:t'] \subseteq \Gamma^{(m)}_{\AG}\cap \Gamma^{(n)}_{\FI}\right\}
\]
denote the corresponding end time of the maximal interval starting at $t_i$.

\begin{align}
\mathbb{E}\!\left[\underbar{R}_a(T)\right]
    &= \mathbb{E}\!\left[\underbar{R}_a\!\left(\Bar{\Gamma}^{(m)}_\AG \cup \Bar{\Gamma}^{(n)}_\FI\right)\right]
      +\mathbb{E}\!\left[\underbar{R}_a\!\left(\Gamma^{(m)}_\AG \cap \Gamma^{(n)}_\FI\right)\right]\nonumber\\
    &= \mathbb{E}\!\left[\underbar{R}_a\!\left(\Bar{\Gamma}^{(m)}_\AG \cup \Bar{\Gamma}^{(n)}_\FI\right)\right]
      +\sum_{t_i \in \mathcal{Y}}\mathbb{E}\!\left[\underbar{R}_a\!\left([t_i,t'_i]\right)\right]\nonumber\\
    &\leq \mathbb{E}\!\left[\left|\Bar{\Gamma}^{(m)}_\AG \cup \Bar{\Gamma}^{(n)}_\FI\right|\right]
      +\mathbb{E}\!\left[\sum_{t_i \in \mathcal{Y}} \underbar{R}_a\!\left([t_i,t'_i]\right)\right]\nonumber\\
    &\leq \mathbb{E}\!\left[\left|\Bar{\Gamma}^{(m)}_\AG \cup \Bar{\Gamma}^{(n)}_\FI\right|\right]
      +\mathbb{E}\!\left[\sum_{t_i \in \mathcal{Y}} \sum_{t \in [t_i,t'_i]} \mathbbm{1}\left\{\FM_a(t) \in \emptyset \cup \mathcal{L}_{a,\underline{f^*_a}}\right\}\right]\nonumber\\
    &\leq \mathbb{E}\!\left[\left|\Bar{\Gamma}^{(m)}_\AG \cup \Bar{\Gamma}^{(n)}_\FI\right|\right]
      +\mathbb{E}\!\left[\sum_{t_i \in \mathcal{Y}} \sum_{t \in [t_i,t'_i]} \mathbbm{1}\left\{(a',\FM_{a'}(t))_{\forall a' \in \AG} \text{ is not stable}\right\}\right]\nonumber\\
    &\leq \mathbb{E}\!\left[\left|\Bar{\Gamma}^{(m)}_\AG \cup \Bar{\Gamma}^{(n)}_\FI\right|\right]
      +\mathbb{E}\!\left[\left|\mathcal{Y}\right|\cdot \max_{t_i \in \mathcal{Y}}\left( \sum_{t \in [t_i,t'_i]} \mathbbm{1}\left\{(a',\FM_{a'}(t))_{\forall a' \in \AG} \text{ is not stable}\right\}\right)\right]\nonumber\\
    &\overset{(a)}{\leq} \mathbb{E}\!\left[\left|\Bar{\Gamma}^{(m)}_\AG \cup \Bar{\Gamma}^{(n)}_\FI\right|\right]
      +\mathbb{E}\!\left[\left|\mathcal{Y}\right|\right]\cdot
       \mathbb{E}\!\left[ \max\left( \sum_{t \in [t_1,t_2]} \mathbbm{1}\left\{(a',\FM_{a'}(t))_{\forall a' \in \AG} \text{ is not stable}\right\}\right)\right]\nonumber\\
    &\overset{(b)}{\leq} \mathbb{E}\!\left[\left|\Bar{\Gamma}^{(m)}_\AG \cup \Bar{\Gamma}^{(n)}_\FI\right|\right]
      +\Bigl(\mathbb{E}\!\left[\left|\Bar{\Gamma}^{(m)}_\AG \cup \Bar{\Gamma}^{(n)}_\FI\right|\right]+1\Bigr)\cdot
       \mathbb{E}\!\left[ \max\left( \sum_{t \in [t_1,t_2]} \mathbbm{1}\left\{(a',\FM_{a'}(t))_{\forall a' \in \AG} \text{ is not stable}\right\}\right)\right]. \label{ineq:regnoncoopgenfinale}
\end{align}
Here, (a) follows by upper bounding
\(
\max_{t_i\in\mathcal{Y}} \underbar{R}_a([t_i,t'_i])
\)
with the corresponding quantity over an arbitrary interval $[t_1,t_2]$ satisfying $t''\in \Gamma^{(m)}_\AG \cap \Gamma^{(n)}_\FI$ for all $t''\in[t_1,t_2]$, and then applying $\mathbb{E}[XY]\le \mathbb{E}[X]\mathbb{E}[Y]$ for independent nonnegative random variables. Inequality (b) follows since each maximal interval $[t_i:t'_i]\subseteq \Gamma^{(m)}_\AG \cap \Gamma^{(n)}_\FI$ has, whenever defined, its boundary time steps $t_i-1$ and $t'_i+1$ belonging to $\Bar{\Gamma}^{(m)}_\AG \cup \Bar{\Gamma}^{(n)}_\FI$ (and only one boundary exists when $t_i=1$ or $t'_i=T$). Hence, the number of such maximal intervals is at most $\left|\Bar{\Gamma}^{(m)}_\AG \cup \Bar{\Gamma}^{(n)}_\FI\right|+1$.

By~\cref{lemma:noncoopgeninterval}, we obtain
\begin{align}
\mathbb{E}\!\left[
\max_{t_i \in \mathcal{Y}}
\left(
\sum_{t \in [t_i,t'_i]}
\mathbbm{1}\left\{(a',\FM_{a'}(t))_{\forall a' \in \AG}\ \text{is not stable}\right\}
\right)
\right]
\leq
1
+
\frac{n^4m+nm}{\bigl(\lambda\cdot(1-\lambda)^{n-1}\bigr)^{n^4m+nm}}.\label{neq:regnoncoopgenfinale1}
\end{align}
It remains to bound $\mathbb{E}\!\left[\left|\Bar{\Gamma}^{(m)}_\AG \cup \Bar{\Gamma}^{(n)}_\FI\right|\right]$. We write
\begin{align}
 \mathbb{E}\!\left[\left|\Bar{\Gamma}^{(m)}_\AG \cup \Bar{\Gamma}^{(n)}_\FI\right|\right]
 &\leq
 \mathbb{E}\!\left[\left|\Bar{\Gamma}^{(m)}_\AG \right|\right]
 +
 \mathbb{E}\!\left[\left|\Bar{\Gamma}^{(n)}_\FI\right|\right]\nonumber\\
&\leq
\mathbb{E}\!\left[\left|\left\{
t \in \mathcal{T} :
\exists a \in \AG\ \exists f \in \FI,\;
t \in \Bar{\mathcal{E}}_{a,f}
\right\}\right|\right]
+
\mathbb{E}\!\left[\left|\left\{
t \in \mathcal{T} :
\exists a' \in \AG\ \exists f' \in \FI,\;
t \in \Bar{\mathcal{E}}_{f',a'}
\right\}\right|\right]\nonumber\\
&\leq
\sum_{a\in \AG}\sum_{f\in\FI}
\frac{4|\mathcal{L}_{a,f}|
 \exp\!\left(-\tfrac{\Delta^2}{2m}\right)}
 {1 - \exp\!\left(-\tfrac{\Delta^2}{2m}\right)}
+
\frac{4|\mathcal{L}_{f,a}|
 \exp\!\left(-\tfrac{\Delta^2}{2m}\right)}
 {1 - \exp\!\left(-\tfrac{\Delta^2}{2m}\right)}\nonumber\\
&\leq \sum_{a\in \AG}\sum_{f\in\FI} O\!\left(m\!\left(|\mathcal{L}_{a,f}|+|\mathcal{L}_{f,a}|\right)\Delta^{-2}\right)\nonumber\\
&\in O\!\left(nm^3\Delta^{-2}\right), \label{neq:regnoncoopgenfinale2}
\end{align}
where the last inequality follows by~\cref{lemma:valid-rr}.

Finally, combining~\eqref{ineq:regnoncoopgenfinale}, \eqref{neq:regnoncoopgenfinale1}, and~\eqref{neq:regnoncoopgenfinale2}, we obtain
\[
\mathbb{E}\!\left[\underbar{R}_a(T)\right]
\in
O\left(\frac{n^5m^4\Delta^{-2}}{\bigl(\lambda\cdot(1-\lambda)^{n-1}\bigr)^{n^4m+nm}}\right).
\]
\end{proof}

\section{Discussion on Incentive Compatibility in Decentralized Learning}
\label{sec:ic}
We study double-sided learning in a matching market $\MM$, where both agents and firms update their behavior over time based on their own observations. Since stability is a shared objective on both sides, a natural question is whether the decentralized dynamics we propose are \emph{incentive compatible} in the following sense: if all other players follow a fixed prescribed algorithm, does any single agent or firm have an incentive to deviate in order to obtain a more favorable stable outcome?

Concretely, we consider agents following the decentralized Algorithms~\ref{alg:drr},~\ref{alg:ancdrr}, and~\ref{alg:Eancdrr}, together with firms implementing the strategic rejection policy in~\cref{alg:fdrr}. In markets that may admit multiple stable matchings, where agent-optimal/agent-pessimal and firm-optimal/firm-pessimal stable matches are well-defined, we ask whether following these algorithms---given that others do the same---leads each side to converge to the stable matching that is optimal (or pessimal) from its own perspective.

We next give a general notion of \emph{incentive compatibility} for a player on either side of the market, defined relative to fixed algorithms adopted by all other players, and under the informational constraints of our learning model.

\begin{definition}[Incentive compatibility under partial information]\label{def:IC}
Consider a matching market $\MM$ in which the underlying (ground-truth) preference profile is unknown to all players. At each time step, each agent and firm only maintains its own estimated preference list (formed from its observed samples), and players observe outcomes only through the firm-side feedback stream (e.g., $\VF(t)$ or $\VFP(t)$, depending on the model).
Fix a pair of policies $(\pi_{\AG},\pi_{\FI})$, where $\pi_{\AG}$ specifies the agents' learning policy and $\pi_{\FI}$ specifies the firms' rejection (and, if applicable, learning) policy. 
For a side $S \in \{\AG,\FI\}$, we say that policy $\pi_S$ is \emph{incentive compatible} for side $S$ (relative to $\pi_{-S}$) if, under these information constraints, when all players on side $S$ follow $\pi_S$ and all players on the opposite side follow $\pi_{-S}$, the induced dynamics converge to the $S$-optimal stable matching (equivalently, the $( -S)$-pessimal stable matching).
\end{definition}

We now discuss the incentive compatibility of agents and firms under our decentralized learning dynamics, distinguishing between markets with a unique stable matching and markets with multiple stable matchings.

\subsection{Incentive compatibility in matching markets with a unique stable matching}\label{sec:ic-unique}

When $\MM$ admits a \emph{unique} stable matching, incentive considerations are comparatively clean: convergence to a stable outcome uniquely determines the optimal stable outcome for both sides. We formalize this under our decentralized information structure and show that, conditional on convergence, the prescribed agent- and firm-side policies are incentive compatible in the sense of \cref{def:IC}.

\begin{proposition}[Incentive compatibility under a unique stable matching]\label{prop:ic-unique}
Suppose $\MM$ admits a unique stable matching $\mu^\star$ under the ground-truth preferences (this includes $\alpha$-reducible markets). Assume the decentralized information structure where each player knows only its own estimated preference list and observes only the firm-side feedback. Fix any agent learning policy
$\pi_1 \in \{\textnormal{Algorithms~\ref{alg:drr},~\ref{alg:ancdrr},~\ref{alg:Eancdrr}}\}$
and the firm rejection policy $\pi_2=\textnormal{Algorithm~\ref{alg:fdrr}}$.
If, when all agents follow $\pi_1$ and all firms follow $\pi_2$, the induced dynamics converge to a stable matching, then they converge to $\mu^\star$, and $\pi_1$ is incentive compatible for agents and $\pi_2$ is incentive compatible for firms with respect to \cref{def:IC}.
\end{proposition}

\begin{proof}
Since $\MM$ has a unique stable matching, every stable outcome equals $\mu^\star$. Hence, under $(\pi_1,\pi_2)$, convergence to any stable matching implies convergence to $\mu^\star$. Because no alternative stable matching exists, no agent or firm can improve its stable outcome by deviating in order to steer the dynamics toward a different stable limit; deviations can only affect transient behavior (e.g., delaying coordination or prolonging suboptimal commitments). Therefore, following $\pi_1$ (resp., $\pi_2$) is incentive compatible for agents (resp., firms) in the sense of \cref{def:IC}.
\end{proof}

\subsection{Incentive compatibility in matching markets with multiple stable matchings}\label{sec:ic-multiple}

When $\MM$ admits \emph{multiple} stable matchings, incentive considerations become inherently subtle: a guarantee of convergence to \emph{a} stable matching does not identify \emph{which} stable outcome is selected, and different stable matchings can yield different payoffs for both sides (the number of stable matchings may be exponential in $n$ and $m$). Thus, a policy profile may be stability-preserving while still failing to be incentive compatible in the sense of \cref{def:IC}, which singles out a particular stable outcome (agent-pessimal / firm-optimal in our definition).

\paragraph{Agents.}
In our decentralized learning model, agents are the proposing side and each player has only partial information (it knows only its own estimated preference list and observes only the firm-side feedback). Under this information structure, even if all agents follow the prescribed learning policy and firms follow the prescribed rejection policy, convergence need not be to the \emph{agent-optimal} stable matching; see \cref{examp:drrs-4}. Therefore, in markets with multiple stable matchings, the prescribed agent learning policy is \emph{not} incentive compatible for agents under \cref{def:IC}: an agent whose objective is to secure its agent-optimal stable partner cannot generally view ``following the algorithm'' as guaranteeing convergence to that target stable outcome.

\paragraph{Firms.}
A similar selection issue arises on the firm side. Our firm-side rejection policy~\cref{alg:fdrr} is designed to
preserve stability under limited feedback by triggering corrective updating phases when a firm detects local
inconsistencies; it is \emph{not} designed to perform outcome-steering strategic rejections.
Thus, in general markets with multiple stable matchings, convergence to \emph{some} stable matching under
\cref{alg:fdrr} does not imply convergence to the agent-pessimal (equivalently, firm-optimal) stable matching.

Importantly, even when firms are \emph{certain}, the literature shows that strategic rejections can change which
stable matching is selected and may improve firms' final matches; see, e.g., \cite{kupfer2018influence}.
In particular, there exist instances where allowing firms to use strategic rejections leads to a stable outcome
that is strictly better for all firms than the outcome obtained without such rejections.
Since under certainty our policy does not employ such outcome-improving strategic rejections, a firm may have an
incentive to deviate in markets with multiple stable matchings, in an attempt to steer the dynamics toward a
more firm-preferred stable outcome. Consequently, without an additional guarantee that the induced dynamics
converge to the firm-optimal stable matching, following~\cref{alg:fdrr} is \emph{not} incentive compatible for
firms in the sense of \cref{def:IC}.

\begin{proposition}[Lack of incentive compatibility under multiple stable matchings]\label{prop:ic-multiple-neg}
Suppose $\MM$ admits more than one stable matching under the ground-truth preferences, and consider the decentralized information structure where each player knows only its own estimated preference list and observes only the firm-side feedback. Fix any agent learning policy
$\pi_1 \in \{\textnormal{Algorithms~\ref{alg:drr},~\ref{alg:ancdrr},~\ref{alg:Eancdrr}}\}$
and the firm rejection policy $\pi_2=\textnormal{Algorithm~\ref{alg:fdrr}}$.

\begin{itemize}
    \item There exists an instance (e.g., \cref{examp:drrs-4}) in which the induced dynamics under $(\pi_1,\pi_2)$ converge to a stable matching that is \emph{not} agent-optimal, then $\pi_1$ is not incentive compatible for agents in the sense of \cref{def:IC}.
    \item There exists an instance in which the induced dynamics under $(\pi_1,\pi_2)$ converge to a stable matching that is \emph{not} agent-pessimal (equivalently, not firm-optimal), then $\pi_2$ is not incentive compatible for firms in the sense of \cref{def:IC}.
\end{itemize}

\end{proposition}

\noindent
Taken together with \cref{prop:ic-unique}, this highlights a sharp contrast: under a unique stable matching, convergence to stability pins down a unique long-run outcome and yields incentive compatibility for both sides, whereas under multiple stable matchings, stable-outcome selection becomes endogenous and the same decentralized policies need not be incentive compatible for \emph{general} underlying markets.




\section{Optimality of the Results}\label{sec:remarks1}
In this section, we discuss the optimality of our regret guarantees for both the centralized and decentralized algorithms. Our goal is to position the obtained bounds relative to natural lower-bound benchmarks, and to clarify which gaps are inherent to the learning problem versus artifacts of coordination and information constraints. In particular, we compare our upper bounds to immediate information-theoretic baselines and outline the remaining gaps as concrete lower-bound questions.

\subsection{Optimality of the Centralized Bounds~\cref{thm:ciarr}}\label{sec:centopt}

We prove that the regret guarantee of Algorithm~\ref{alg:ciarr} is near-optimal by establishing a matching lower bound (up to a factor $m$) that applies to any policy achieving stable outcomes.
\begin{restatable}{proposition}{propcia}\label{prop:ciarr}
~\cref{alg:ciarr}'s regret $O(nm^2\cdot\min\!\{\overline{\Delta}_{\AG},\underline{\Delta}_{\FI}\}^{-2})$ is within a factor $m$ of the information-theoretic optimal regret.
\end{restatable}
\begin{proof}
Consider any policy $\pi$ under which the market converges to a stable matching with sublinear regret.
To certify stability, each agent $a$ must learn, for every firm $f$ and competing agent $a'$, whether $f$ prefers $a$ to $a'$, since these comparisons determine blocking pairs.
This information requires at least one round in which $a$ applies to $f$, so any policy must incur $\Omega(nm)$ regret per agent over the $O(nm)$ such triplets $(a,a',f)$.
Since \cref{alg:ciarr} achieves $O(nm^2)$ regret, it is within a factor $m$ of optimal.
The $\min\!\{\overline{\Delta}_{\AG},\underline{\Delta}_{\FI}\}^{-2}$ factor reflects the minimum number of samples needed to distinguish any two partners on the opposite side.
\end{proof}

By~\cref{thm:ciarr}, our centralized algorithm~\cref{alg:ciarr} achieves a time-independent regret bound of $O(nm^2)$. This guarantee is within a factor $m$ of the immediate benchmark $O(nm)$, and~\cref{prop:ciarr} shows that \cref{alg:ciarr} is nearly optimal under that benchmark.

We further conjecture that the $O(nm^2)$ rate is in fact asymptotically optimal for learning in this setting. Intuitively, learning ordinal preferences requires, for each of the $n$ agents, distinguishing between $\Theta(m^2)$ firm pairs. In the worst case, this can demand a constant number of observations per pair (depending on the gaps between the firms' expected rewards), which suggests an inherent $\Omega(nm^2)$ difficulty.

\begin{conjecture}\label{conj:ciarr-opt}
The $O(nm^2)$ regret bound of the centralized~\cref{alg:ciarr} is asymptotically optimal.
\end{conjecture}

\subsection{Optimality of the Coordinated Decentralized Bounds~\cref{thm:decregmain} }\label{opt:coord}
Compared to the centralized algorithm~\cref{alg:ciarr}, coordinated decentralized learning necessarily incurs an additional overhead due to distributed executions of $\FGS$. In the centralized setting, $\FGS$ is effectively run in a single round at each time step, which underlies the time-independent $O(nm^2)$ regret guarantee of~\cref{thm:ciarr}. By contrast, under the coordinated decentralized algorithm~\cref{alg:drr}, agents must coordinate to jointly run $\FGS$, and each such distributed execution can take up to $O(n^2)$ rounds, during which agents may incur regret.

Consequently, a direct comparison suggests an $O(n^2)$ multiplicative overhead relative to the centralized benchmark: combining the $O(n^2)$ duration of a distributed $\FGS$ run with the $O(nm^2)$ learning cost yields an $O(n^3m^2)$ regret bound, matching the guarantee in $\alpha$-reducible markets in~\cref{thm:decregmain}. Under this interpretation, the coordinated decentralized guarantee remains nearly optimal up to the inherent coordination cost.

For general (unstructured) markets,~\cref{thm:decregmain} yields an $O(n^4m^2)$ bound, which is an additional factor $n$ larger than the $\alpha$-reducible guarantee. Equivalently, relative to the immediate $O(nm)$ benchmark (corresponding to the idealized centralized baseline that resolves matching each round), the general-market bound is larger by a factor $O(n^3m)$.

\subsection{Optimality of the Coordination-Free Decentralized Bounds~\cref{thm:noncodecreg}}\label{opt:noncoord}
As noted earlier, the richer firm-side feedback $\VFP$ enables the design of coordination-free decentralized learning. In structured markets (in particular, $\alpha$-reducible markets), the market-dependent regret bound of~\cref{thm:noncodecreg} for~\cref{alg:ancdrr} is at most $O(n^3m^2)$ in the worst case, and therefore remains within an $O(n^2m)$ factor of the immediate benchmark $O(nm)$.

A key difference from the coordinated algorithm is that coordination-free executions do not incur a fixed $O(n^2)$ overhead from explicitly coordinating a distributed run of $\FGS$. Instead, the regret bound in~\cref{thm:noncodecreg} adapts to the intrinsic time complexity of $\FGS$ under the underlying market structure: in markets where $\FGS$ completes in $O(n)$ rounds, the same analysis yields an $O(n^2m^2)$ bound. Hence, for $\alpha$-reducible markets, the coordination-free guarantees can be viewed as adaptively nearly optimal with respect to the number of rounds required by $\FGS$ under the ground-truth preferences. This leads to a substantial improvement over the coordination-free approach of~\cite{maheshwari2022decentralized} in such structured markets.

For general (unstructured) markets, the extension~\cref{alg:Eancdrr} admits a time-independent bound in~\cref{thm:noncodecreg} whose leading term is dominated by a $\lambda$-dependent constant that is exponential in $n$ and $m$, namely
\(
\frac{1}{\bigl(\lambda\cdot(1-\lambda)^{n-1}\bigr)^{n^4m+nm}}.
\)
While this dependence is consistent with the phenomena highlighted in~\cite{liu2021bandit}, it is far larger than the immediate lower bound $O(n^3m)$ for coordination-free decentralized learning. As discussed by~\cite{liu2021bandit,abeledo1995paths}, such large constants are unavoidable in fully general markets under coordination-free information constraints as simultaneous blocking pair resolution may cause infinite cycles. Establishing tight lower bounds for coordination-free algorithms in unstructured matching markets therefore remains an interesting and challenging open problem.

\section{Related Work}\label{sec:remarks2}
We present here the main ways our results strengthen prior work: we obtain \emph{time-independent} guarantees under \emph{limited firm-side feedback}, while explicitly modeling \emph{firm uncertainty} via a strategic action space and a rejection policy. This yields a more general framework for decentralized double-sided learning, and our firm-side rejection mechanism can be incorporated into existing learning-to-stability approaches to preserve convergence even when firms update their estimates over time.

\subsection{Applicability of Strategic Firm's Rejection Policy~\cref{alg:fdrr}}
A key conceptual contribution of our framework is the firm-side rejection policy~\cref{alg:fdrr}, which is designed to handle \emph{firm uncertainty} while preserving convergence to stability under decentralized learning. The policy uses controlled (strategic) rejections as a mechanism for re-initiating correction phases exactly when a firm detects that its current local estimates may be inconsistent with the true preference ordering relevant to stability.

Beyond its role in our algorithms, \cref{alg:fdrr} can be viewed as a modular primitive that can be adapted to prior decentralized matching-bandit frameworks \cite{liu2021bandit,maheshwari2022decentralized,zhangbandit}. Many existing works establish convergence to a stable matching under the assumption that firms' preferences are fixed and known (or, more generally, that firms act as passive acceptors without learning). Under such assumptions, stability is guaranteed only because the firm side introduces no additional uncertainty into the dynamics. Our rejection policy provides a principled way to relax this restriction: it allows firms to be uncertain and to learn their own preferences, while still enabling the overall process to converge to a stable outcome by triggering corrective updates whenever firm-side estimation errors would otherwise derail stability.

To our knowledge, this use of strategic firm-side rejections specifically to incorporate \emph{learning and uncertainty on the firm side} while maintaining convergence-to-stability guarantees is first introduced in this work. Consequently, incorporating \cref{alg:fdrr} into prior approaches offers a direct path to strengthening their results by extending them from settings with fully known firm preferences to more realistic environments in which firms are uncertain and must learn.

\subsection{Improvements with Limited Firm-Side Feedback}
The firm-side feedback model studied here, together with our treatment of firm uncertainty via an extended firm action space, yields substantial improvements over the guarantees in~\cite{liu2021bandit}. In particular, their model assumes that firms have fixed preference lists that are known to both firms and agents throughout the learning process, which is a strong and often unrealistic assumption. They also assume a much richer feedback structure in which agents observe not only vacancies and hiring changes, but also the identity of each firm’s current match at the end of every round.

In contrast, our bounds are obtained under anonymous firm-side feedback and while allowing firms to be both strategic and uncertain. Taken together, these features lead to nearly optimal regret guarantees under significantly weaker informational assumptions, improving upon~\cite{liu2021bandit} along multiple dimensions.

\subsection{Model and Result Differences with Recent Work with Two-Sided Uncertainty}
An ETC-then-$\FGS$ algorithm for two-sided online learning was proposed in \cite{pagare2023two}. Although their ETC algorithmic structure is similar to ours, our problem setup differs from theirs in several key respects: (a) firms in our setting can strategically defer hiring, which complicates the matching process; (b) agents can coordinate through mutual signaling; and (c) agents undergo an interview phase prior to matching. These distinctions necessitate fundamentally different algorithmic techniques, and our algorithms further achieve horizon-independent regret, substantially improving upon their logarithmic bounds.


\section{Future Directions and Open Questions}\label{sec:remarks3}
Having detailed our improvements over prior work, we now outline several future directions that remain as interesting open problems.

\begin{itemize}
    \item \textbf{Firm-side incentives beyond the unique-stable regime.}
    Our analysis focuses on agent regret and introduces the firm-side rejection policy~\cref{alg:fdrr}, which guarantees convergence to stability under limited firm-side feedback. However, as discussed in \cref{sec:ic-unique,sec:ic-multiple}, this policy is incentive compatible for firms only in markets with a unique stable matching (cf.~\cref{prop:ic-unique}), while incentive alignment becomes subtle when multiple stable matchings exist.
    A natural open question is to design \emph{alternative} firm-side rejection policies that (i) preserve convergence to stability under the same decentralized information structure (each firm knows only its own estimated preference list and observes the firm-side feedback) and (ii) are more closely aligned with firm objectives in multi-stable markets, i.e., can bias equilibrium selection toward firm-preferred stable outcomes.
    This direction connects to the influence of firm-side strategic rejections in markets with certain firms (e.g., \cite{kupfer2018influence}) and raises the challenge of extending such outcome-improvement phenomena to an \emph{online} learning-to-stability setting, potentially also under our extended firm action space that models uncertainty.

    \item \textbf{Reward-adaptive application decisions and stability notions.}
    In our model, agents choose whom to apply/interview based on their current estimates (e.g., empirical means) \emph{before} observing the realized interview reward in that round.
    An interesting extension is a \emph{reward-sensitive} variant in which an agent can condition its final application decision on the realized interview outcomes within the same round (e.g., after interviewing a subset of firms).
    This change can reduce regret by enabling within-round adaptation, but it also suggests stronger, reward-contingent notions of stability and benchmarking: the relevant comparison may no longer be regret to a fixed stable matching (agent-optimal/pessimal), but to an outcome concept that accounts for the additional adaptivity enabled by realized rewards.
    Analyzing this setting would likely require tracking not only mean gaps but also distributional properties (e.g., variance or tail behavior), and may call for estimators and concentration tools beyond the empirical-mean-based arguments used throughout our paper.
\end{itemize}

\end{document}